\documentclass[11pt,a4paper]{article}

\usepackage{amsmath,amssymb,amsthm,mathrsfs}
\usepackage{stmaryrd}
\usepackage{hyperref}
\usepackage{graphicx}
\usepackage{xcolor}
\usepackage{soul}
\usepackage[thicklines]{cancel}
\usepackage{fullpage}
\usepackage[inline]{enumitem}
\usepackage{tikz}
\usepackage{pgfplots}
\pgfplotsset{compat=1.18}
\usepackage{subcaption}
\usepackage{titlesec}
\usepackage{cancel}
\usepackage{mathtools}

\theoremstyle{plain}
\newtheorem{theorem}{Theorem}[section]
\newtheorem{lemma}[theorem]{Lemma}
\newtheorem{proposition}[theorem]{Proposition}

\theoremstyle{definition}

\newtheorem{ass}[theorem]{Assumption}

\newtheorem{remark}[theorem]{Remark}
\numberwithin{equation}{section}

\newcommand{\dd}{\,\mathrm{d}}

\renewcommand{\epsilon}{\varepsilon}

\allowdisplaybreaks

\begin{document}
    \title{Explicit Signal-Adaptive Sequential Optimal Execution Quotes
    \addtocounter{footnote}{-1}\thanks{The author gratefully acknowledges financial support from the IDEA League Fellowship. The author would also like to thank Andrew Allan for helpful discussions at various stages of this project, and Mathieu Rosenbaum for valuable comments at the NUS Quantitative Finance Conference 2025, which helped improve the manuscript.}
    }

 \author{
Fenghui Yu
\thanks{Delft Institute of Applied Mathematics, TU Delft, 2628 CD Delft, The Netherlands.
(E-mail: fenghui.yu@tudelft.nl).}}

    \maketitle

\begin{abstract}
This paper develops a unified explicit solution theory for optimal execution through sequential limit-order placement in
a limit order book. Rather than controlling only the trading speed of a
metaorder, we determine how individual limit orders should be quoted over time.
The model incorporates signal-dependent drift, price impact, inventory risk, and
execution risk, with fills modeled by point processes whose intensities depend
on the submitted quotes. We formulate four execution criteria: expected terminal wealth, expected terminal wealth with running inventory penalty, CARA utility of terminal wealth, and CARA utility with running inventory penalty. For general
price-impact and inventory-penalty functions, we derive the corresponding HJB equations and
show that all four problems reduce to a triangular finite-dimensional
structure which can be solved explicitly, leading to fully explicit value functions and optimal quotes across all cases. We also prove well-posedness, admissibility, and verification results. The
explicit formulas reveal connections between quoting strategies under different
criteria, support long-horizon asymptotic analysis, and show numerically that
signal-dependent drift can substantially affect optimal execution. 
    \end{abstract}

        \noindent{\bf Keywords:}
Optimal execution; optimal quoting; limit order book; stochastic control with controlled jumps; explicit solutions; fill risk; price impact; inventory risk; market signals.
    \section{Introduction}

Optimal execution is a central problem in market microstructure and mathematical
finance. Its aim is to determine how a large order should be executed over time
while balancing execution costs, price impact, inventory exposure, and execution
risk. In practice, this problem naturally appears at several levels. At the
macro level, a trader decides how to split a metaorder into smaller child orders
over a finite horizon, leading to an optimal scheduling problem. This viewpoint
has been extensively studied through models in which the trading speed is the
control variable; see, for instance, Almgren and Chriss
\cite{almgren2001optimal} in discrete time and Cartea and Jaimungal
\cite{cartea2016incorporating} in continuous time. At the micro level, once a
child order has been specified, the trader must decide where to place it in the
limit order book. This gives rise to the optimal order placement, or optimal
quoting, problem, in which the control is the quote submitted to the book. A
third level concerns venue selection and smart order routing.

This paper focuses on the interaction between the first two levels. Rather than
separating the scheduling of a metaorder from the placement of its individual
child orders, we study an integrated sequential execution problem in which a
trader repeatedly chooses limit-order quotes over a finite horizon in order to
liquidate a target inventory. The chosen quotes determine both the aggressiveness
of execution and the probability of execution. More aggressive quotes increase
the likelihood of fills but reduce execution revenues, whereas more passive
quotes improve prices but expose the trader to fill risk and inventory risk. We
refer to this integrated problem as an optimal execution quoting problem.

The resulting model combines several risks that are often treated separately in
the literature. At the level of metaorder execution, price impact and inventory
risk are essential determinants of the optimal execution schedule. At the level
of limit-order placement, execution is not guaranteed, and fill risk must be
incorporated explicitly. In our framework, fills are modeled through point
processes whose intensities depend on the submitted quotes, so that the cash and
inventory processes exhibit jumps when orders are executed. This leads to a
nonlinear stochastic control problem in which the control enters the jump
intensity.

A further important feature of our model is that the reference price (can be interpreted as mid-price) is allowed
to have a non-zero drift. Many tractable optimal quoting models assume a
martingale mid-price, mainly for analytical convenience. This assumption excludes
predictive information about short-term price movements. In practice, however,
signals extracted from order-flow imbalance, limit-order-book dynamics, or
data-driven forecasting methods may contain information about the future drift
of the reference price. Such signals are increasingly important in modern
execution and market-making problems; see, for example, Kolm et al.
\cite{kolm2023deep} for deep learning methods for extracting alpha signals.
Allowing for a non-zero drift therefore makes it possible to study how market
signals affect optimal quoting decisions. Our numerical results show that the
optimal quotes may change substantially when such drift information is included.

The literature on optimal execution and optimal quoting is extensive. Classical
optimal execution models usually control the trading rate and study the tradeoff
between price impact and inventory risk. In contrast, optimal quoting models,
originating from market-making problems such as Ho and Stoll \cite{HO198147} and
Avellaneda and Stoikov \cite{Avellaneda01042008}, control bid and ask quotes
through their effect on execution intensities. Further refinements have
incorporated adverse selection, price discovery, and other microstructure
features; see, for example, Barzykin et al.
\cite{barzykin2025optimalquotingadverseselection,Barzykin2025}. The present
paper belongs to the class of models in which liquidation is carried out through
sequential limit-order placement, thereby combining execution and quoting.

The works particularly relevant to ours include Bayraktar and Ludkovski
\cite{Bayraktar2014}, Guéant et al. \cite{Guant2012}, Guéant and Lehalle
\cite{Guant2015}, Cartea et al. \cite{Cartea2014}, Cartea and Jaimungal
\cite{cartea2015optimal}, and the treatment in Cartea et al.
\cite{cartea2015algorithmic}. These papers study related liquidation or quoting
problems under specific assumptions on the price dynamics, execution
intensities, and objective functions. Bayraktar and Ludkovski
\cite{Bayraktar2014} consider an expected terminal wealth criterion in a setting without an
explicit stochastic mid-price dynamic. Guéant et al. \cite{Guant2012} and Guéant
and Lehalle \cite{Guant2015} study liquidation under CARA utility objectives;
although their formulation allows for diffusion price dynamics, explicit
solutions are obtained only under restrictive assumptions such as zero drift and
zero volatility in the relevant reduced equations. Cartea and Jaimungal
\cite{cartea2015optimal} and Cartea et al. \cite{cartea2015algorithmic} include
expected terminal wealth, price impact, and running inventory penalties, but assume a
zero-drift mid-price and do not obtain fully explicit solutions in the general
setting. Cartea et al. \cite{Cartea2014} allow for a non-zero drift in the
reference price, but the resulting control problem is not solved explicitly.

The main contribution of this paper is to provide a unified and explicit
solution theory for optimal execution quoting with non-zero drift and volatility
in the reference price dynamics. The drift may be interpreted as a signal-driven
predictor of short-term price movements. Our framework also incorporates price
impact, fill risk, and running inventory risk. The price impact and inventory
penalty are modeled by general functions, which allows the framework to recover
several commonly used specifications as special cases. We consider four execution criteria:
\[
\begin{array}{ll}
\text{(I) expected terminal wealth,}
&
\text{(II) expected terminal wealth with inventory penalty,}
\\[0.3em]
\text{(III) CARA utility of terminal wealth,}
&
\text{(IV) CARA utility with inventory penalty.}
\end{array}
\]
For each criterion,
we derive the corresponding Hamilton--Jacobi--Bellman (HJB) equation and show that,
after a suitable transformation, the problem reduces to a
triangular finite-dimensional system. This triangular structure makes it
possible to obtain explicit value functions and explicit optimal quotes, which
can be computed through divided-difference formulas. To the best of our
knowledge, this is the first unified explicit solution theory for an optimal
execution quoting problem of this generality, covering signal-dependent price
dynamics, quote-dependent execution intensities, price impact, inventory risk,
and multiple execution criteria within a single tractable framework.

Closed-form solutions for stochastic control problems with controlled jump
intensities are rare. Even in classical optimal quoting and market-making
models, explicit solutions are usually unavailable except under restrictive
assumptions or through approximation schemes. Guéant et al. \cite{Guant2013},
for example, develop approximation methods for utility-based quoting problems,
while Boyce et al. \cite{Boyce2025} use related approximation ideas for
objectives involving running inventory risk. By contrast, the present paper
derives exact explicit solutions for a broad family of execution quoting
problems. Existing explicit solutions obtained in more restrictive settings,
such as those in Guéant et al. \cite{Guant2012} and Cartea et al.
\cite{cartea2015algorithmic}, are recovered as special cases by choosing the
corresponding model parameters and objective functions.

In addition to explicit solvability, we provide a rigorous verification theory.
We establish well-posedness of the controlled execution model, admissibility of
the feedback quotes, and verification results showing that the candidate value
functions indeed solve the corresponding stochastic control problems. These
results are important because the controlled processes contain jumps and the
control appears in the fill intensities, making the verification argument more
delicate than in standard diffusion control problems.

The explicit formulas allow us to compare the optimal quotes generated by
different criteria. In particular, we identify conditions under which the quotes
arising from a quadratic running inventory penalty coincide with, or differ only
by a constant shift from, those obtained under a CARA utility objective. This
provides a precise sense in which a quadratic running inventory penalty can be
interpreted as a reduced-form certainty-equivalent approximation of exponential
utility under price risk. 

These explicit formulas also enable a long-horizon asymptotic analysis of the
optimal quotes. We show that the asymptotic behavior is governed by the dominant
growth modes of the triangular system: depending on the signal drift, inventory
penalties, volatility, and risk aversion, the optimal quote may grow linearly
with the remaining horizon, grow only logarithmically in the fully degenerate
case, or remain bounded at the linear scale, meaning that it is sublinear in the remaining
horizon. These regimes clarify when the
trader becomes increasingly patient, remains boundedly patient, or adjusts only
gradually as the execution horizon increases.

Finally, we investigate the role of signals through numerical experiments. We
compare optimal quotes under different drift specifications, including constant,
decaying, and latent signals. The results illustrate that ignoring drift
information can lead to substantially different execution strategies. Since the
solution is explicit, uncertainty in the estimated drift can also be propagated
directly into the optimal quotes.

The remainder of the paper is organized as follows. Section
\ref{sec:model} introduces the execution model, the admissible controls, and the
four optimization criteria. It also specifies the assumptions under which the
control problems are well posed. Section \ref{sec:solution} derives the HJB
equations and shows how the four problems reduce to triangular systems with
explicit solutions. Section \ref{sec:verification_admissibility} proves
well-posedness, admissibility, and verification results. Section
\ref{sec:comparison_optimal_quotes} compares the optimal quotes
across the four objectives and discusses their financial interpretation. Section
\ref{sec:long_horizon_asymptotics} studies the long-horizon asymptotics of the
optimal quotes. Section \ref{sec:numerics} presents numerical experiments under
different signal specifications. Section \ref{sec:conclusion} concludes.

\section{The model and control problems}
\label{sec:model}

In this section, we introduce the market model and formulate the optimal execution problems studied in the paper. We focus on the liquidation of a sell order through sequentially submitted limit orders. The corresponding acquisition problem with buy limit orders can be treated analogously, with the appropriate changes in signs.

Let \(T>0\) be a fixed trading horizon. Consider a filtered probability space
$(\Omega,\mathcal F,(\mathcal F_t)_{0\leq t\leq T},\mathbb P)$
satisfying the usual conditions. All stochastic processes appearing below are assumed to be defined on this stochastic basis. The agent initially holds \(Q_0\in\mathbb N\) units of an asset and aims to liquidate this inventory over the time interval \([0,T]\).

We denote by \(M=(M_t)_{0\leq t\leq T}\) the reference price of the asset. This reference price may be interpreted as the mid-price, or alternatively as the best bid price, depending on the precise market convention adopted. We assume that \(M\) follows the diffusion dynamics
\begin{equation}
\label{eq:mid-price_1}
    \dd M_t
    =
    g(s_t)\,\dd t
    +
    \sigma_t\,\dd W_t,
\end{equation}
where \(W=(W_t)_{0\leq t\leq T}\) is a standard Brownian motion. The process \(s=(s_t)_{0\leq t\leq T}\) denotes a generic market signal, for instance a signal extracted from the limit order book. The function
$g:\mathbb R^d\to\mathbb R$
captures the predictive effect of this signal on the future price movement. In particular, \(g(s_t)\) represents the drift component of the reference price induced by the signal \(s_t\). The process \(\sigma=(\sigma_t)_{0\leq t\leq T}\) denotes the volatility of the reference price.

The agent controls the quote depth
\[
    \delta=(\delta_t)_{0\leq t\leq T},
\]
where \(\delta_t\) specifies the distance from the reference price at which a sell limit order is posted. Executions are modeled by a counting process
\[
    N^\delta=(N_t^\delta)_{0\leq t\leq T},
\]
where \(N_t^\delta\) records the number of executed orders up to time \(t\). We assume that each execution corresponds to the sale of one unit of inventory. Hence the remaining inventory is
\begin{equation}
\label{eq:inventory_process}
    Q_t^\delta
    =
    Q_0-N_t^\delta,
    \qquad 0\leq t\leq T.
\end{equation}

Conditional on the control \(\delta\), the execution process \(N^\delta\) is assumed to have stochastic intensity
\begin{equation}
\label{eq:intensity_process}
    \lambda_t^\delta
    =
    \lambda e^{-\kappa\delta_t}\mathbf 1_{\{Q_{t-}^\delta>0\}},
    \qquad
    \lambda>0,\quad \kappa>0.
\end{equation}
The exponential form captures the standard trade-off between price improvement and execution risk: a larger quote depth gives a more favorable execution price, but reduces the arrival intensity of executions. More precisely, for small \(\Delta t>0\),
\[
    \mathbb P\bigl(
        N_{t+\Delta t}^\delta-N_t^\delta=1
        \,\big|\,
        \mathcal F_t
    \bigr)
    =
    \lambda e^{-\kappa\delta_t}\Delta t
    +
    o(\Delta t),
\]
on the event \(\{Q_{t-}^\delta>0\}\). Thus \(\lambda e^{-\kappa\delta_t}\) is an execution intensity, rather than an execution probability.

The associated cash process \(X^\delta=(X_t^\delta)_{0\leq t\leq T}\) is given by
\begin{equation}
\label{cash1}
    \dd X_t^\delta
    =
    \bigl(M_t-a+b\delta_t\bigr)\,\dd N_t^\delta,
\end{equation}
where \(a\geq 0\) and \(b>0\). Therefore, whenever one unit of inventory is executed, the agent receives the execution price
\[
    M_t-a+b\delta_t.
\]
The term \(-a+b\delta_t\) represents the adjustment of the execution price relative to the reference price. The parameter \(a\) may be interpreted as a fixed adverse price adjustment, while \(b\delta_t\) captures the improvement obtained by posting the order further away from the reference price. The condition \(b>0\) ensures that the quote depth affects the execution price and therefore enters the optimization problem nontrivially.

The liquidation time is defined by
\begin{equation}
\label{eq:liquidation_time}
    \tau
    =
    T\wedge
    \inf\{t\geq 0:Q_t^\delta=0\}.
\end{equation}
Thus, trading stops either when the inventory has been fully liquidated or when the terminal time \(T\) is reached. Equivalently, after the inventory has reached zero, the execution intensity is set to zero and the inventory remains at zero.

\begin{remark}[Relation to the standard no-impact specification]
If the price adjustment caused by the order is ignored, one may set \(a=0\) and \(b=1\). In that case, the cash dynamics reduce to
\begin{equation}
\label{cash1_degenerate}
    \dd X_t^\delta
    =
    (M_t+\delta_t)\,\dd N_t^\delta,
\end{equation}
which corresponds to the standard specification in which the execution price of a sell limit order is the reference price plus the posted depth; see, for example, \cite{cartea2015optimal}. The more general specification \eqref{cash1} allows for an additional execution-price adjustment while retaining the same point-process structure for order arrivals.
\end{remark}

We now impose the standing assumptions used throughout the paper.

\begin{ass}
\label{ass:standing}
The following assumptions hold:
\begin{enumerate}[label=(A\arabic*)]
    \item The initial inventory satisfies \(Q_0\in\mathbb N\), and each execution reduces the inventory by one unit. 

    \item The execution-intensity parameters satisfy
    \(
        \lambda>0,\) and 
        \(\kappa>0.\)
    The execution-price parameters satisfy
    \(a\geq 0,\) and 
        \(b>0.\)

    \item The admissible control set is
    \[
        \mathcal A
        :=
        \left\{
        \delta:\Omega\times[0,T]\to[\delta_{\min},\delta_{\max}]
        \;:\;
        \delta \text{ is predictable}
        \right\},
    \]
    where
    \(
        -\infty<\delta_{\min}<\delta_{\max}<\infty.
    \)

    \item The processes \(g(s_t)\) and \(\sigma_t\) are progressively measurable and bounded. Moreover, \(\sigma_t\geq 0\) for all \(t\in[0,T]\).

    \item The terminal liquidation price-impact penalty satisfies
    \(I:\{0,\dots,Q_0\}\to\mathbb R_+\) and \(I(0)=0\).
    If an inventory-running penalty is included in the objective, it is
    described by a function
    \(J:\{0,\dots,Q_0\}\to\mathbb R_+\), with \(J(0)=0\).

    \item In the cases with exponential utility, the risk-aversion parameter satisfies $
        \gamma>0.$
\end{enumerate}
\end{ass}

Under these assumptions, the state variables of the control problem are the cash
\(X^\delta\), the reference price \(M\), and the remaining inventory
\(Q^\delta\), while the control variable is the quote depth \(\delta\). The
inventory is measured in execution units: one unit may represent the chosen
trade size or lot size. More general execution sizes can therefore be
incorporated by rescaling the inventory accordingly. The predictability requirement means that the quote depth is chosen using
information available just before time \(t\). This ensures that
\(\lambda_t^\delta=\lambda e^{-\kappa\delta_t}\mathbf 1_{\{Q_{t-}^\delta>0\}}\)
is a valid predictable intensity for the execution process and rules out
anticipative strategies that would react to an execution at the same time it
occurs. Financially, \(I\) and \(J\) are typically taken to be nondecreasing in
    the inventory level, but monotonicity is not needed for the results below.

The different optimization
criteria considered below lead to different value functions and HJB equations,
but all are based on the same controlled execution dynamics described above.

\subsection{Case I: Optimization without risk aversion}
\label{subsec:case_I_without_risk_aversion}

We first consider the risk-neutral liquidation problem. The agent maximizes the expected terminal wealth obtained from limit-order executions and, if necessary, from the terminal liquidation of any remaining inventory. More precisely, the objective is
\begin{equation}
\label{eq:initial_value_function_without_risk_aversion}
    H(0,x,M,Q_0)
    =
    \sup_{\delta\in\mathcal A}
    \mathbb E\left[
        X_\tau^\delta
        +
        Q_\tau^\delta
        \bigl(M_\tau-I(Q_\tau^\delta)\bigr)
        \,\middle|\,
        X_0^\delta=x,\,
        M_0=M,\,
        Q_0^\delta=Q_0
    \right].
\end{equation}
The terminal payoff consists of the cash accumulated through executed limit orders and the liquidation value of the residual inventory. The function \(I\) represents the terminal liquidation price-impact penalty per share, so that any remaining inventory \(Q_\tau^\delta\) is liquidated at the penalized reference price
\[
    M_\tau-I(Q_\tau^\delta).
\]
This penalty discourages large residual positions at the terminal time; a common specification is the linear per-share penalty
\[
    I(q)=\alpha q,
    \qquad \alpha\geq 0.
\]
Throughout the paper, we use the term terminal liquidation penalty to refer to this terminal price-impact adjustment.

\subsection{Case II: Optimization with running inventory risk}
\label{subsec:case_II_running_inventory_risk}

We next extend the risk-neutral objective by incorporating a running inventory penalty. This term penalizes the agent for carrying inventory over time and therefore captures the exposure to inventory risk during the liquidation horizon. Let
\[
    J:\{0,\dots,Q_0\}\to\mathbb R_+
\]
be the running inventory penalty function. We assume that \(J\) is nonnegative and increasing, with \(J(0)=0\).

The initial value function is defined by
\begin{equation}
\label{eq:initial_value_function_inventory_risk}
    R(0,x,M,Q_0)
    =
    \sup_{\delta\in\mathcal A}
    \mathbb E\left[
        X_\tau^\delta
        +
        Q_\tau^\delta
        \bigl(M_\tau-I(Q_\tau^\delta)\bigr)
        -
        \int_0^\tau J(Q_s^\delta)\,\dd s
        \,\middle|\,
        X_0^\delta=x,\,
        M_0=M,\,
        Q_0^\delta=Q_0
    \right].
\end{equation}
Compared with Case I, the additional term
\[
    \int_0^\tau J(Q_s^\delta)\,\dd s
\]
penalizes the agent for maintaining a nonzero position before liquidation is completed. A common specification is the quadratic inventory penalty
\[
    J(q)=\beta q^2,
    \qquad
    \beta\geq 0.
\]

\subsection{Case III: Optimization with exponential utility}
\label{subsec:case_III_exponential_utility}

We now consider an execution criterion based on exponential utility. This formulation incorporates risk aversion through the concavity of the utility function and penalizes uncertainty in terminal wealth. Let \(\gamma>0\) denote the coefficient of absolute risk aversion. The initial value function is defined by
\begin{equation}
\label{eq:initial_value_function_exponential_utility}
    U(0,x,M,Q_0)
    =
    \sup_{\delta\in\mathcal A}
    \mathbb E\left[
        -\exp\left\{
            -\gamma
            \left(
                X_\tau^\delta
                +
                Q_\tau^\delta
                \bigl(M_\tau-I(Q_\tau^\delta)\bigr)
            \right)
        \right\}
        \,\middle|\,
        X_0^\delta=x,\,
        M_0=M,\,
        Q_0^\delta=Q_0
    \right].
\end{equation}
The random variable inside the utility function is the terminal wealth obtained from executed limit orders together with the penalized liquidation value of any residual inventory.

\subsection{Case IV: Optimization with exponential utility and running inventory risk}
\label{subsec:case_IV_exponential_utility_inventory_risk}

We finally combine the exponential utility criterion of Case III with the running inventory penalty of Case II. This formulation accounts both for risk aversion with respect to terminal execution wealth and for the cost of carrying inventory during the liquidation period.

The running inventory penalty is treated as a monetary cost deducted from terminal wealth before applying the exponential utility function. Thus, the initial value function is defined by
\begin{equation}
\label{eq:initial_value_function_utility_inventory}
\begin{aligned}
    F(0,x,M,Q_0)
    =
    \sup_{\delta\in\mathcal A}
    \mathbb E\Bigg[
        &-\exp\Bigg\{
            -\gamma
            \Bigg(
                X_\tau^\delta
                +
                Q_\tau^\delta
                \bigl(M_\tau-I(Q_\tau^\delta)\bigr)
                -
                \int_0^\tau J(Q_s^\delta)\,\dd s
            \Bigg)
        \Bigg\} \\
        &\qquad\qquad
        \Bigg|\,
        X_0^\delta=x,\,
        M_0=M,\,
        Q_0^\delta=Q_0
    \Bigg].
\end{aligned}
\end{equation}
Here \(\gamma>0\) is the coefficient of absolute risk aversion, and \(J(q)\) represents the instantaneous cost of holding inventory level \(q\).

\section{Solution to the optimal control problems}
\label{sec:solution}

In this section, we derive finite-dimensional characterizations of the optimal controls. To obtain explicit formulas, we first work under a frozen-coefficient specification, where the signal and volatility are fixed at given values:
\[
    s_t\equiv s,
    \qquad
    \sigma_t\equiv \sigma.
\]
Accordingly, throughout this section we write \(g(s)\) for the constant drift coefficient \(g(s_t)\). Under this specification, the HJB equations for the four cases reduce to triangular finite-dimensional systems of ordinary differential equations (ODE) indexed by the inventory level \(q\in\{0,\dots,Q_0\}\). Since the inventory is finite-state and the cash dependence can be separated explicitly, the resulting equations are finite-dimensional.

Although the formulas are derived under fixed signal and volatility values, this
frozen-coefficient formulation also explains the signal-adaptive nature of the
resulting quotes. At each decision time, newly observed or estimated values of
the signal \(s\) and volatility \(\sigma\) can be inserted directly into the
explicit quote formulas. The resulting feedback quote therefore adapts to the
current market state without requiring the HJB system to be solved again after
each update. In this sense, the closed-form solutions provide statewise adaptive
execution strategies.

If \(s_t\) and \(\sigma_t\) are deterministic time-dependent functions, the same
reductions remain valid, with time-dependent coefficients in the corresponding
triangular ODE systems. 
By contrast, if \(s_t\) is itself stochastic and Markovian, the
signal must be included as an additional state variable and its infinitesimal
generator must be added to the HJB equation. In that case, a direct plug-in
substitution \(s_t=s\) is not justified as an exact dynamic-programming
reduction, but the frozen-coefficient formulas still provide adaptive quotes
based on the currently observed or estimated signal and volatility. In the numerical experiments with time-dependent signals, we use a frozen-signal implementation. 

\subsection{Optimal solution to Case I}
\label{subsec:solution_case_I}

We first solve the risk-neutral liquidation problem. For \(t\in[0,T]\), define
the conditional value function
\begin{equation}
\label{eq:value_function_without_risk_aversion}
    H(t,x,M,q)
    =
    \sup_{\delta\in\mathcal A}
    \mathbb E_{t,x,M,q}
    \left[
        X_\tau^\delta
        +
        Q_\tau^\delta
        \bigl(M_\tau-I(Q_\tau^\delta)\bigr)
    \right],
\end{equation}
where \(\mathbb E_{t,x,M,q}[\cdot]\) denotes conditional expectation given
\[
    X_{t-}^\delta=x,
    \qquad
    M_t=M,
    \qquad
    Q_{t-}^\delta=q.
\]
The use of left limits for \(X^\delta\) and \(Q^\delta\) reflects the fact
that cash and inventory jump at execution times.

The dynamic programming principle gives the following coupled HJB system:
\begin{align}
\label{eq:HJB_equations_drift}
\left\{
\begin{array}{ll}
\displaystyle
\partial_t H(t,x,M,q)
+
g(s)\partial_M H(t,x,M,q)
+
\frac{1}{2}\sigma^2\partial_{MM}H(t,x,M,q)
& \\[0.6em]
\displaystyle\qquad
+
\sup_{\delta\in\mathcal A}
\left\{
    \lambda e^{-\kappa\delta}
    \left[
        H\bigl(t,x+M-a+b\delta,M,q-1\bigr)
        -
        H(t,x,M,q)
    \right]
\right\}
=0,
& q\geq 1,\; t<T,
\\[1.2em]
\displaystyle
H(T,x,M,q)
=
x+q\bigl(M-I(q)\bigr),
& q\geq 1,
\\[0.8em]
\displaystyle
H(t,x,M,0)
=
x,
& t\leq T.
\end{array}
\right.
\end{align}
The differential operator
\[
    g(s)\partial_M+\frac12\sigma^2\partial_{MM}
\]
is the infinitesimal generator of the reference price process. The jump term
corresponds to the execution of one unit of inventory: the cash increases by
\(M-a+b\delta\), while the inventory decreases from \(q\) to \(q-1\). Thus the
supremum over \(\delta\) captures the trade-off between price improvement and
execution intensity.

\begin{remark}[Linear terminal price impact]
For the linear per-share terminal penalty
\[
    I(q)=\alpha q,
    \qquad \alpha\geq0,
\]
the terminal condition becomes
\[
    H(T,x,M,q)=x+qM-\alpha q^2.
\]
The HJB equation is otherwise unchanged. 
\end{remark}

For later use, we denote by
\[
    \Pi_{[\delta_{\min},\delta_{\max}]}(y)
    :=
    \min\{\delta_{\max},\max\{\delta_{\min},y\}\},
    \qquad y\in\mathbb R,
\]
the projection onto the admissible quote interval
\([\delta_{\min},\delta_{\max}]\).

\begin{lemma}[Optimal quote in Case I]
\label{lem:optimal_quote_case_I}
Suppose that the unconstrained maximizer is admissible. Then the value function
admits the representation
\begin{equation}
\label{eq:ansatz_case_I}
    H(t,x,M,q)
    =
    x+qM+h(t,q),
\end{equation}
where
\[
    h(t,q)=\frac{b}{\kappa}\log w(t,q),
    \qquad w(t,q)>0.
\]
The function \(w\) solves the triangular linear system
\begin{align}
\label{eq:PDE_no_risk}
\left\{
\begin{array}{ll}
\displaystyle
\partial_t w(t,q)
+
\frac{\kappa}{b}g(s)q\,w(t,q)
+
\lambda\exp\left\{-\frac{\kappa a}{b}-1\right\}w(t,q-1)
=0,
& q\geq 1,\; t<T,
\\[1.1em]
\displaystyle
w(T,q)
=
\exp\left\{
    -\frac{\kappa}{b}qI(q)
\right\},
& q\geq 1,
\\[1.1em]
\displaystyle
w(t,0)=1.
&
\end{array}
\right.
\end{align}
The unconstrained optimal feedback quote is
\begin{equation}
\label{eq:optimal_delta_case1_w}
    \delta^{\mathrm{unc}}(t,q)
    =
    \frac{1}{\kappa}
    \left(
        1+\log\frac{w(t,q)}{w(t,q-1)}
    \right)
    +
    \frac{a}{b}.
\end{equation}
If the admissible quote set is \([\delta_{\min},\delta_{\max}]\), then the
admissible feedback quote is obtained by projection:
\begin{equation}
\label{eq:optimal_delta_feedback_case_I}
    \delta^*(t,q)
    =
    \Pi_{[\delta_{\min},\delta_{\max}]}
    \bigl(\delta^{\mathrm{unc}}(t,q)\bigr).
\end{equation}
\end{lemma}

\begin{proof}
The terminal and boundary conditions in \eqref{eq:HJB_equations_drift} suggest
the affine ansatz
\[
    H(t,x,M,q)=x+qM+h(t,q).
\]
The term \(x+qM\) is the mark-to-market wealth, while \(h(t,q)\) represents the
additional value generated by optimally choosing the quote depth for the
remaining liquidation problem.

Using
\[
    \partial_t H=\partial_t h,
    \qquad
    \partial_M H=q,
    \qquad
    \partial_{MM}H=0,
\]
and substituting the ansatz into \eqref{eq:HJB_equations_drift}, we obtain, for
\(q\geq1\),
\begin{align}
\label{eq:HJB_equations_drift_ansatz}
\left\{
\begin{array}{ll}
\displaystyle
\partial_t h(t,q)
+
g(s)q
+
\sup_{\delta\in\mathcal A}
\left\{
    \lambda e^{-\kappa\delta}
    \left[
        -a+b\delta+h(t,q-1)-h(t,q)
    \right]
\right\}
=0,
& t<T,
\\[1.1em]
\displaystyle
h(T,q)=-qI(q),
& q\geq 1,
\\[0.6em]
\displaystyle
h(t,0)=0.
&
\end{array}
\right.
\end{align}

For fixed \((t,q)\), define
\[
    \Delta h(t,q):=h(t,q)-h(t,q-1).
\]
The term to be maximized is
\[
    \lambda e^{-\kappa\delta}
    \left[
        -a+b\delta-\Delta h(t,q)
    \right].
\]
The first-order condition for an interior maximizer is
\[
\begin{aligned}
0
&=
\frac{\partial}{\partial\delta}
\left(
    \lambda e^{-\kappa\delta}
    \left[
        -a+b\delta-\Delta h(t,q)
    \right]
\right)
\\
&=
\lambda e^{-\kappa\delta}
\left(
    b
    -
    \kappa
    \left[
        -a+b\delta-\Delta h(t,q)
    \right]
\right).
\end{aligned}
\]
Hence the unconstrained maximizer is
\begin{equation}
\label{eq:optimal_delta_feedback_unconstrained_case_I_h}
    \delta^{\mathrm{unc}}(t,q)
    =
    \frac{1}{\kappa}
    +
    \frac{a}{b}
    +
    \frac{h(t,q)-h(t,q-1)}{b}.
\end{equation}

Substituting \eqref{eq:optimal_delta_feedback_unconstrained_case_I_h} into
\eqref{eq:HJB_equations_drift_ansatz} yields
\begin{align}
\label{eq:HJB_equations_drift_feedback_solution}
\left\{
\begin{array}{ll}
\displaystyle
\partial_t h(t,q)
+
g(s)q
+
\frac{\lambda b}{e\kappa}
\exp\left\{
    -\frac{\kappa a}{b}
    -
    \frac{\kappa}{b}
    \bigl(h(t,q)-h(t,q-1)\bigr)
\right\}
=0,
& q\geq 1,\; t<T,
\\[1.1em]
\displaystyle
h(T,q)=-qI(q),
& q\geq 1,
\\[0.6em]
\displaystyle
h(t,0)=0.
&
\end{array}
\right.
\end{align}

Now set
\[
    h(t,q)=\frac{b}{\kappa}\log w(t,q),
    \qquad w(t,q)>0.
\]
Then
\[
    h(t,q)-h(t,q-1)
    =
    \frac{b}{\kappa}
    \log\frac{w(t,q)}{w(t,q-1)}.
\]
Therefore the feedback quote becomes
\[
    \delta^{\mathrm{unc}}(t,q)
    =
    \frac{1}{\kappa}
    \left(
        1+\log\frac{w(t,q)}{w(t,q-1)}
    \right)
    +
    \frac{a}{b}.
\]
Substituting the logarithmic transformation into
\eqref{eq:HJB_equations_drift_feedback_solution} gives
\[
\begin{aligned}
0
&=
\frac{b}{\kappa}
\frac{\partial_t w(t,q)}{w(t,q)}
+
g(s)q
+
\frac{\lambda b}{e\kappa}
\exp\left\{-\frac{\kappa a}{b}\right\}
\frac{w(t,q-1)}{w(t,q)}.
\end{aligned}
\]
Multiplying by \(\kappa w(t,q)/b\) gives \eqref{eq:PDE_no_risk}. The terminal
and boundary conditions follow from
\[
    h(T,q)=-qI(q),
    \qquad
    h(t,0)=0,
\]
which imply
\[
    w(T,q)=\exp\left\{-\frac{\kappa}{b}qI(q)\right\},
    \qquad
    w(t,0)=1.
\]
If the admissible set is bounded, the admissible maximizer is the projection of
the unconstrained maximizer onto \([\delta_{\min},\delta_{\max}]\).
\end{proof}

The coefficients of \(w(t,q)\) and \(w(t,q-1)\) in \eqref{eq:PDE_no_risk} are
\[
    A_q^{(H)}
    :=
    \frac{\kappa}{b}g(s)q,
    \qquad
    C^{(H)}
    :=
    \lambda\exp\left\{-\frac{\kappa a}{b}-1\right\}.
\]
Thus Case I already has the common triangular form introduced in Section \ref{subsec:common_recursive_structure} below.

\subsection{Optimal solution to Case II}
\label{subsec:solution_case_II}

We next solve the problem with running inventory risk. For \(t\in[0,T]\), define
the conditional value function
\begin{equation}
\label{eq:value_function_inventory_risk}
    R(t,x,M,q)
    =
    \sup_{\delta\in\mathcal A}
    \mathbb E_{t,x,M,q}
    \left[
        X_\tau^\delta
        +
        Q_\tau^\delta
        \bigl(M_\tau-I(Q_\tau^\delta)\bigr)
        -
        \int_t^\tau J(Q_s^\delta)\,\dd s
    \right],
\end{equation}
where \(\mathbb E_{t,x,M,q}[\cdot]\) denotes conditional expectation given
\[
    X_{t-}^\delta=x,
    \qquad
    M_t=M,
    \qquad
    Q_{t-}^\delta=q.
\]

The dynamic programming principle gives the coupled HJB system
\begin{align}
\label{eq:HJB_equations_inventory_risk}
\left\{
\begin{array}{ll}
\displaystyle
\partial_t R(t,x,M,q)
+
g(s)\partial_M R(t,x,M,q)
+
\frac{1}{2}\sigma^2\partial_{MM}R(t,x,M,q)
-
J(q)
& \\[0.6em]
\displaystyle\qquad
+
\sup_{\delta\in\mathcal A}
\left\{
    \lambda e^{-\kappa\delta}
    \left[
        R\bigl(t,x+M-a+b\delta,M,q-1\bigr)
        -
        R(t,x,M,q)
    \right]
\right\}
=0,
& q\geq 1,\; t<T,
\\[1.2em]
\displaystyle
R(T,x,M,q)
=
x+q\bigl(M-I(q)\bigr),
& q\geq 1,
\\[0.8em]
\displaystyle
R(t,x,M,0)
=
x,
& t\leq T.
\end{array}
\right.
\end{align}
Compared with Case I, the HJB contains the additional term \(-J(q)\). This
reflects the instantaneous cost of holding inventory level \(q\). Financially,
this running penalty lowers the value of waiting and therefore encourages faster
liquidation.

\begin{remark}[Quadratic running inventory penalty]
A common specification is
\[
    J(q)=\beta q^2,
    \qquad \beta\geq0.
\]
Then the running-cost term in \eqref{eq:HJB_equations_inventory_risk} becomes
\(-\beta q^2\). When \(\beta=0\), the formulation reduces to Case I.
\end{remark}

\begin{lemma}[Optimal quote in Case II]
\label{lem:optimal_quote_case_II}
Suppose that the unconstrained maximizer is admissible. Then the value function
admits the representation
\begin{equation}
\label{eq:ansatz_running_risk}
    R(t,x,M,q)
    =
    x+qM+r(t,q),
\end{equation}
where
\[
    r(t,q)=\frac{b}{\kappa}\log w(t,q),
    \qquad w(t,q)>0.
\]
The function \(w\) solves the triangular linear system
\begin{align}
\label{eq:PDE_running_risk}
\left\{
\begin{array}{ll}
\displaystyle
\partial_t w(t,q)
+
\frac{\kappa}{b}
\bigl(g(s)q-J(q)\bigr)w(t,q)
+
\lambda\exp\left\{-\frac{\kappa a}{b}-1\right\}w(t,q-1)
=0,
& q\geq 1,\; t<T,
\\[1.1em]
\displaystyle
w(T,q)
=
\exp\left\{
    -\frac{\kappa}{b}qI(q)
\right\},
& q\geq 1,
\\[1.1em]
\displaystyle
w(t,0)=1.
&
\end{array}
\right.
\end{align}
The unconstrained optimal feedback quote is
\begin{equation}
\label{eq:optimal_delta_case2_w}
    \delta^{\mathrm{unc}}(t,q)
    =
    \frac{1}{\kappa}
    \left(
        1+\log\frac{w(t,q)}{w(t,q-1)}
    \right)
    +
    \frac{a}{b}.
\end{equation}
If the admissible quote set is \([\delta_{\min},\delta_{\max}]\), then the
admissible feedback quote is obtained by projection:
\begin{equation}
\label{eq:optimal_delta_feedback_running_risk}
    \delta^*(t,q)
    =
    \Pi_{[\delta_{\min},\delta_{\max}]}
    \bigl(\delta^{\mathrm{unc}}(t,q)\bigr).
\end{equation}
\end{lemma}

\begin{proof}
The argument follows the same steps as the proof of Lemma~\ref{lem:optimal_quote_case_I}: one substitutes the affine ansatz into the HJB equation, derives the first-order condition for the optimal quote, and then applies the logarithmic transformation to obtain the triangular linear system. The details are given in Appendix~\ref{app:optimal-quote-caseII-proof}.
\end{proof}

Compared with Case I, the running inventory penalty modifies the coefficient
multiplying \(w(t,q)\) in \eqref{eq:PDE_running_risk}: the drift contribution
\(g(s)q\) is replaced by the net term \(g(s)q-J(q)\). Thus the coefficients
multiplying \(w(t,q)\) and \(w(t,q-1)\) in \eqref{eq:PDE_running_risk} are
\[
    A_q^{(R)}
    :=
    \frac{\kappa}{b}\bigl(g(s)q-J(q)\bigr),
    \qquad
    C^{(R)}
    :=
    \lambda\exp\left\{-\frac{\kappa a}{b}-1\right\}.
\]
Financially, \(g(s)q\) represents the expected mark-to-market benefit of
holding inventory under a favourable signal, whereas \(J(q)\) represents the
cost of carrying inventory. Hence Case II has the common triangular form
introduced in Section~\ref{subsec:common_recursive_structure} below.

\subsection{Optimal solution to Case III}
\label{subsec:solution_case_III}

We now solve the problem with exponential utility. For \(t\in[0,T]\), define
the conditional value function
\begin{equation}
\label{eq:value_function_utility}
    U(t,x,M,q)
    =
    \sup_{\delta\in\mathcal A}
    \mathbb E_{t,x,M,q}
    \left[
        -\exp\left\{
            -\gamma
            \left(
                X_\tau^\delta
                +
                Q_\tau^\delta
                \bigl(M_\tau-I(Q_\tau^\delta)\bigr)
            \right)
        \right\}
    \right],
\end{equation}
where \(\mathbb E_{t,x,M,q}[\cdot]\) denotes conditional expectation given
\[
    X_{t-}^\delta=x,
    \qquad
    M_t=M,
    \qquad
    Q_{t-}^\delta=q.
\]

The
dynamic programming principle gives the coupled HJB system
\begin{align}
\label{eq:HJB_equations_utility}
\left\{
\begin{array}{ll}
\displaystyle
\partial_t U(t,x,M,q)
+
g(s)\partial_M U(t,x,M,q)
+
\frac{1}{2}\sigma^2\partial_{MM}U(t,x,M,q)
& \\[0.6em]
\displaystyle\qquad
+
\sup_{\delta\in\mathcal A}
\left\{
    \lambda e^{-\kappa\delta}
    \left[
        U\bigl(t,x+M-a+b\delta,M,q-1\bigr)
        -
        U(t,x,M,q)
    \right]
\right\}
=0,
& q\geq 1,\; t<T,
\\[1.2em]
\displaystyle
U(T,x,M,q)
=
-\exp\left\{
    -\gamma
    \left(
        x+q\bigl(M-I(q)\bigr)
    \right)
\right\},
& q\geq 1,
\\[1.0em]
\displaystyle
U(t,x,M,0)
=
-\exp\{-\gamma x\},
& t\leq T.
\end{array}
\right.
\end{align}
The controlled state dynamics are the same as in Case I, so the HJB has the
same jump-diffusion structure. The difference is the exponential utility
criterion, which changes the terminal and boundary conditions and introduces
risk aversion through the nonlinear dependence on cash and inventory value.

\begin{lemma}[Optimal quote in Case III]
\label{lem:optimal_quote_case_III}
Suppose that the unconstrained maximizer is admissible. Then the value function
admits the CARA representation
\begin{equation}
\label{eq:ansatz_utility}
    U(t,x,M,q)
    =
    -\exp\left\{
        -\gamma\bigl(x+qM+u(t,q)\bigr)
    \right\},
\end{equation}
where
\[
    u(t,q)=\frac{b}{\kappa}\log w(t,q),
    \qquad w(t,q)>0.
\]
Define
\begin{equation}
\label{eq:hat_lambda_utility}
    \widehat{\lambda}
    :=
    \lambda
    \left(
        \frac{\kappa}{\kappa+b\gamma}
    \right)^{\frac{\kappa}{b\gamma}+1}
    e^{-\kappa a/b}.
\end{equation}
Then \(w\) solves the triangular linear system
\begin{align}
\label{eq:PDE_utility}
\left\{
\begin{array}{ll}
\displaystyle
\partial_t w(t,q)
+
\frac{\kappa}{b}
\left(
    g(s)q
    -
    \frac{1}{2}\sigma^2\gamma q^2
\right)
w(t,q)
+
\widehat{\lambda}w(t,q-1)
=0,
& q\geq 1,\; t<T,
\\[1.2em]
\displaystyle
w(T,q)
=
\exp\left\{
    -\frac{\kappa}{b}qI(q)
\right\},
& q\geq 1,
\\[1.0em]
\displaystyle
w(t,0)=1.
&
\end{array}
\right.
\end{align}
The unconstrained optimal feedback quote is
\begin{equation}
\label{eq:optimal_delta_feedback_utility_w}
    \delta^{\mathrm{unc}}(t,q)
    =
    \frac{1}{b\gamma}
    \log\left(\frac{\kappa+b\gamma}{\kappa}\right)
    +
    \frac{a}{b}
    +
    \frac{1}{\kappa}
    \log\frac{w(t,q)}{w(t,q-1)}.
\end{equation}
If the admissible quote set is \([\delta_{\min},\delta_{\max}]\), then the
admissible feedback quote is obtained by projection:
\begin{equation}
\label{eq:optimal_delta_feedback_utility}
    \delta^*(t,q)
    =
    \Pi_{[\delta_{\min},\delta_{\max}]}
    \bigl(\delta^{\mathrm{unc}}(t,q)\bigr).
\end{equation}
\end{lemma}

\begin{proof}
The proof is analogous to that of Lemma~\ref{lem:optimal_quote_case_I}, with the exponential-utility ansatz replacing the affine ansatz. Substituting the ansatz into the HJB equation, optimizing over the quote depth, and applying the logarithmic transformation yields the stated triangular system and feedback quote. The details are provided in Appendix~\ref{app:optimal-quote-caseIII-proof}.
\end{proof}

Compared with Cases I and II, the exponential
utility criterion modifies the coefficient multiplying \(w(t,q)\) in
\eqref{eq:PDE_utility} by adding the certainty-equivalent inventory-risk term
\(
    -\frac12\sigma^2\gamma q^2.
\)
Thus the coefficients multiplying \(w(t,q)\) and \(w(t,q-1)\) in
\eqref{eq:PDE_utility} are
\[
    A_q^{(U)}
    :=
    \frac{\kappa}{b}
    \left(
        g(s)q-\frac12\sigma^2\gamma q^2
    \right),
    \qquad
    C^{(U)}
    :=
    \widehat{\lambda}.
\]
Financially, the term \(\frac12\sigma^2\gamma q^2\) represents the
certainty-equivalent cost of holding inventory under price volatility and
exponential utility. Hence Case III also has the common triangular form introduced
in Section~\ref{subsec:common_recursive_structure} below.

\subsection{Optimal solution to Case IV}
\label{subsec:solution_case_IV}

We finally solve the problem with exponential utility and running inventory risk.
For \(t\in[0,T]\), define the conditional value function
\begin{equation}
\label{eq:value_function_utility_inventory}
\begin{aligned}
    F(t,x,M,q)
    =
    \sup_{\delta\in\mathcal A}
    \mathbb E_{t,x,M,q}
    \Bigg[
        -\exp\Bigg\{
            -\gamma
            \Bigg(
                X_\tau^\delta
                +
                Q_\tau^\delta
                \bigl(M_\tau-I(Q_\tau^\delta)\bigr)
                -
                \int_t^\tau J(Q_s^\delta)\,\dd s
            \Bigg)
        \Bigg\}
    \Bigg],
\end{aligned}
\end{equation}
where \(\mathbb E_{t,x,M,q}[\cdot]\) denotes conditional expectation given
\[
    X_{t-}^\delta=x,
    \qquad
    M_t=M,
    \qquad
    Q_{t-}^\delta=q.
\]

We first explain the origin of the running-cost term in the HJB equation. Over
a short time interval \([t,t+\Delta t]\), if the inventory remains equal to
\(q\), then
\[
    \int_t^{t+\Delta t}J(Q_s^\delta)\,\dd s
    =
    J(q)\Delta t+o(\Delta t).
\]
Since this term is subtracted from the monetary payoff before applying
exponential utility, the contribution of this short interval to the payoff is
\[
    -J(q)\Delta t+o(\Delta t).
\]
Hence, inside the exponential utility, this generates the factor
\[
    \exp\{-\gamma(-J(q)\Delta t+o(\Delta t))\}
    =
    1+\gamma J(q)\Delta t+o(\Delta t).
\]
More precisely, the dynamic programming expansion gives
\[
\begin{aligned}
F(t,x,M,q)
\approx
\sup_{\delta\in\mathcal A}
\mathbb E_{t,x,M,q}
\Big[
\bigl(1+\gamma J(q)\Delta t\bigr)
F\bigl(
    t+\Delta t,
    X_{t+\Delta t}^\delta,
    M_{t+\Delta t},
    Q_{t+\Delta t}^\delta
\bigr)
\Big].
\end{aligned}
\]
Because \(J(q)\) is deterministic conditional on the current state
\((t,x,M,q)\), this multiplicative factor can be taken outside the conditional
expectation. For a fixed admissible control, suppressing the optimization
notation for readability,
\[
\begin{aligned}
F(t,x,M,q)
\approx
\bigl(1+\gamma J(q)\Delta t\bigr)
\mathbb E_{t,x,M,q}
\left[
F\bigl(
    t+\Delta t,
    X_{t+\Delta t}^\delta,
    M_{t+\Delta t},
    Q_{t+\Delta t}^\delta
\bigr)
\right].
\end{aligned}
\]
Let \(\mathcal L^\delta\) denote the infinitesimal generator of the controlled
state process without the running inventory penalty. Then
\[
\begin{aligned}
\mathbb E_{t,x,M,q}
\left[
F\bigl(
    t+\Delta t,
    X_{t+\Delta t}^\delta,
    M_{t+\Delta t},
    Q_{t+\Delta t}^\delta
\bigr)
\right]
=
F(t,x,M,q)
+
\Delta t\,\mathcal L^\delta F(t,x,M,q)
+
o(\Delta t).
\end{aligned}
\]
Combining the previous two expansions gives
\[
\begin{aligned}
F(t,x,M,q)
\approx
F(t,x,M,q)
+
\Delta t\,\mathcal L^\delta F(t,x,M,q)
+
\gamma J(q)\Delta t\,F(t,x,M,q)
+
o(\Delta t).
\end{aligned}
\]
After subtracting \(F(t,x,M,q)\), dividing by \(\Delta t\), and letting
\(\Delta t\to0\), the running inventory penalty contributes the zeroth-order
term
\[
    \gamma J(q)F(t,x,M,q)
\]
to the HJB equation. Since \(F<0\), this term lowers the value when \(J(q)>0\),
as expected for an inventory-carrying cost.

The
dynamic programming principle therefore yields
\begin{align}
\label{eq:HJB_equations_utility_inventory}
\left\{
\begin{array}{ll}
\displaystyle
\partial_t F(t,x,M,q)
+
g(s)\partial_M F(t,x,M,q)
+
\frac{1}{2}\sigma^2\partial_{MM}F(t,x,M,q)
+
\gamma J(q)F(t,x,M,q)
& \\[0.6em]
\displaystyle\qquad
+
\sup_{\delta\in\mathcal A}
\left\{
    \lambda e^{-\kappa\delta}
    \left[
        F\bigl(t,x+M-a+b\delta,M,q-1\bigr)
        -
        F(t,x,M,q)
    \right]
\right\}
=0,
& q\geq 1,\; t<T,
\\[1.2em]
\displaystyle
F(T,x,M,q)
=
-\exp\left\{
    -\gamma
    \left(
        x+q\bigl(M-I(q)\bigr)
    \right)
\right\},
& q\geq 1,
\\[1.0em]
\displaystyle
F(t,x,M,0)
=
-\exp\{-\gamma x\},
& t\leq T.
\end{array}
\right.
\end{align}
This HJB combines the effects of Cases II and III. The exponential utility
changes the terminal and boundary conditions, while the running inventory cost
appears as the additional zeroth-order term \(\gamma J(q)F\).

\begin{lemma}[Optimal quote in Case IV]
\label{lem:optimal_quote_case_IV}
Suppose that the unconstrained maximizer is admissible. Then the value function
admits the CARA representation
\begin{equation}
\label{eq:ansatz_utility_inventory}
    F(t,x,M,q)
    =
    -\exp\left\{
        -\gamma\bigl(x+qM+f(t,q)\bigr)
    \right\},
\end{equation}
where
\[
    f(t,q)=\frac{b}{\kappa}\log w(t,q),
    \qquad w(t,q)>0.
\]
Let \(\widehat{\lambda}\) be the effective execution coefficient defined in
\eqref{eq:hat_lambda_utility}. Then \(w\) solves
\begin{align}
\label{eq:PDE_utility_inventory}
\left\{
\begin{array}{ll}
\displaystyle
\partial_t w(t,q)
+
\frac{\kappa}{b}
\left[
    g(s)q
    -
    \frac{1}{2}\sigma^2\gamma q^2
    -
    J(q)
\right]
w(t,q)
+
\widehat{\lambda}w(t,q-1)
=0,
& q\geq 1,\; t<T,
\\[1.2em]
\displaystyle
w(T,q)
=
\exp\left\{
    -\frac{\kappa}{b}qI(q)
\right\},
& q\geq 1,
\\[1.0em]
\displaystyle
w(t,0)=1.
&
\end{array}
\right.
\end{align}
The unconstrained optimal feedback quote is
\begin{equation}
\label{eq:optimal_delta_feedback_utility_inventory_w}
    \delta^{\mathrm{unc}}(t,q)
    =
    \frac{1}{b\gamma}
    \log\left(\frac{\kappa+b\gamma}{\kappa}\right)
    +
    \frac{a}{b}
    +
    \frac{1}{\kappa}
    \log\frac{w(t,q)}{w(t,q-1)}.
\end{equation}
If the admissible quote set is \([\delta_{\min},\delta_{\max}]\), then
\[
    \delta^*(t,q)
    =
    \Pi_{[\delta_{\min},\delta_{\max}]}
    \bigl(\delta^{\mathrm{unc}}(t,q)\bigr).
\]
\end{lemma}

\begin{proof}
The proof follows the same structure as the proof of Lemma~\ref{lem:optimal_quote_case_III}. The only additional contribution is the running inventory cost, which produces the term \(-J(q)\) in the reduced certainty-equivalent equation. After optimizing over the quote depth and applying the logarithmic transformation, one obtains the stated triangular system and feedback quote. The details are provided in Appendix~\ref{app:optimal-quote-caseIV-proof}.
\end{proof}

Compared with Case III, the running inventory penalty further modifies the
coefficient multiplying \(w(t,q)\) in \eqref{eq:PDE_utility_inventory} by
subtracting \(J(q)\). Thus the coefficients multiplying \(w(t,q)\) and
\(w(t,q-1)\) in \eqref{eq:PDE_utility_inventory} are
\[
    A_q^{(F)}
    :=
    \frac{\kappa}{b}
    \left(
        g(s)q
        -
        \frac12\sigma^2\gamma q^2
        -
        J(q)
    \right),
    \qquad
    C^{(F)}
    :=
    \widehat{\lambda}.
\]
Financially, inventory is penalized both through the explicit carrying cost
\(J(q)\) and through the exponential-utility volatility penalty
\(\frac12\sigma^2\gamma q^2\). Hence Case IV also has the common triangular form
introduced in Section~\ref{subsec:common_recursive_structure} below.

\subsection{A common recursive structure}
\label{subsec:common_recursive_structure}

The transformations above show that all four cases lead to the same linear
triangular structure. More precisely, after the appropriate change of variables,
the function \(w(t,q)\) satisfies
\begin{align}
\label{eq:PDE_general}
\left\{
\begin{array}{ll}
\displaystyle
\partial_t w(t,q)
+
A_q w(t,q)
+
C w(t,q-1)
=0,
& q\geq 1,\; t<T,
\\[1.1em]
\displaystyle
w(T,q)=G_q,
& q\geq 1,
\\[0.6em]
\displaystyle
w(t,0)=1.
&
\end{array}
\right.
\end{align}
Here \(A_q\in\mathbb R\), \(C>0\), and \(G_q>0\) are deterministic
coefficients. We set
\[
    A_0=0,
    \qquad
    G_0=1,
\]
which is consistent with the boundary condition \(w(t,0)=1\). In the four cases
considered above, the terminal condition is
\begin{equation}
\label{eq:common_terminal_condition}
    G_q
    =
    \exp\left\{
        -\frac{\kappa}{b}qI(q)
    \right\},
    \qquad q=0,\dots,Q_0.
\end{equation}

The coefficients \(A_q\) and \(C\) depend on the objective criterion:
\begin{equation}
\label{eq:case_coefficients_A_C}
\begin{array}{lll}
\text{Case I:}
&
\displaystyle
A_q=\frac{\kappa}{b}g(s)q,
&
\displaystyle
C=\lambda\exp\left\{-\frac{\kappa a}{b}-1\right\},
\\[1.2em]
\text{Case II:}
&
\displaystyle
A_q=\frac{\kappa}{b}\bigl(g(s)q-J(q)\bigr),
&
\displaystyle
C=\lambda\exp\left\{-\frac{\kappa a}{b}-1\right\},
\\[1.2em]
\text{Case III:}
&
\displaystyle
A_q=\frac{\kappa}{b}
\left(
    g(s)q-\frac{1}{2}\sigma^2\gamma q^2
\right),
&
\displaystyle
C=\widehat{\lambda},
\\[1.2em]
\text{Case IV:}
&
\displaystyle
A_q=\frac{\kappa}{b}
\left(
    g(s)q-\frac{1}{2}\sigma^2\gamma q^2-J(q)
\right),
&
\displaystyle
C=\widehat{\lambda},
\end{array}
\end{equation}
where
\begin{equation}
\label{eq:hat_lambda_common}
    \widehat{\lambda}
    :=
    \lambda
    \left(
        \frac{\kappa}{\kappa+b\gamma}
    \right)^{\frac{\kappa}{b\gamma}+1}
    e^{-\kappa a/b}.
\end{equation}

Thus the four optimization problems differ only through the coefficient
\(A_q\), which encodes the financial trade-offs associated with signal drift,
running inventory cost, and price risk, and through the effective execution
coefficient \(C\). 

The system \eqref{eq:PDE_general} is triangular in the inventory variable:
for each \(q\geq1\), the equation for \(w(t,q)\) depends only on \(w(t,q)\) and
\(w(t,q-1)\). Therefore, given \(w(\cdot,q-1)\), the
variation-of-constants formula gives the recursive representation
\begin{equation}
\label{eq:w_recursive_solution_general}
    w(t,q)
    =
    e^{A_q(T-t)}G_q
    +
    C\int_t^T e^{A_q(u-t)}w(u,q-1)\,\dd u.
\end{equation}
This recursion provides a constructive way to compute \(w(t,q)\) successively
for \(q=1,\dots,Q_0\). In
Section~\ref{subsec:explicit_triangular_solution}, we show that this common
system can in fact be solved fully explicitly, even for more general terminal
conditions \(G_q\). The recursive representation is therefore useful both as a
constructive solution method and as an intermediate step toward the closed-form
solution.

\begin{remark}[Deterministic time-dependent coefficients]
\label{rem:time_dependent_coefficients}
If the signal or volatility is deterministic and time-dependent, the triangular
structure remains valid with \(A_q\) replaced by \(A_q(t)\). In that case
\eqref{eq:w_recursive_solution_general} becomes
\[
    w(t,q)
    =
    G_q
    \exp\left\{
        \int_t^T A_q(r)\,\dd r
    \right\}
    +
    C
    \int_t^T
    \exp\left\{
        \int_t^u A_q(r)\,\dd r
    \right\}
    w(u,q-1)\,\dd u .
\]
Thus the recursion remains valid, but the constant-coefficient
divided-difference formula in Section~\ref{subsec:explicit_triangular_solution} below 
no longer applies directly. If the signal is stochastic, it must be included as
an additional state variable in the HJB equation.
\end{remark}

\subsection{Explicit solution of the triangular system}
\label{subsec:explicit_triangular_solution}

The transformations introduced in the previous subsections show that all four
control problems reduce to the same triangular linear system. Therefore, we don't need to solve
the resulting ODEs separately for each objective, we can solve the general
system given by \eqref{eq:PDE_general}.

\subsubsection{Non-degenerate case}
\label{subsubsec:explicit_triangular_nondegenerate}

We first consider the non-degenerate case
\begin{equation}
\label{eq:nondegenerate_Aq}
    A_i\neq A_j,
    \qquad
    i\neq j,
    \qquad
    i,j\in\{0,\dots,Q_0\}.
\end{equation}
Under this assumption, the denominators in the representation below are
well defined, and the solution admits a finite-sum representation.

\begin{proposition}[Explicit solution in the non-degenerate case]
\label{prop:explicit_solution_triangular_system}
Suppose that \eqref{eq:nondegenerate_Aq} holds. For \(0\leq r\leq q\), define
\begin{equation}
\label{eq:Phi_rq_definition}
    \Phi_{r,q}(\tau)
    :=
    \sum_{i=r}^q
    \frac{e^{A_i\tau}}
    {\displaystyle\prod_{\substack{j=r\\ j\neq i}}^q (A_i-A_j)}.
\end{equation}
Then the solution of \eqref{eq:PDE_general} is
\begin{equation}
\label{eq:general_explicit_solution}
    w(t,q)
    =
    \sum_{r=0}^q
    C^{\,q-r}G_r\,\Phi_{r,q}(T-t),
    \qquad q=0,1,\dots,Q_0,
\end{equation}
where \(C\) is the coefficient in \eqref{eq:PDE_general}. Equivalently,
\begin{equation}
\label{eq:general_explicit_solution_expanded}
    w(t,q)
    =
    \sum_{r=0}^q
    C^{\,q-r}G_r
    \sum_{i=r}^q
    \frac{e^{A_i(T-t)}}
    {\displaystyle\prod_{\substack{j=r\\ j\neq i}}^q (A_i-A_j)}.
\end{equation}
\end{proposition}

\begin{proof}[Idea of proof]
Set
\[
    v_q(\tau):=w(T-\tau,q),
    \qquad 0\leq \tau\leq T.
\]
Then \eqref{eq:PDE_general} becomes
\[
    \frac{\dd}{\dd\tau}v_q(\tau)
    =
    A_qv_q(\tau)+Cv_{q-1}(\tau),
    \qquad
    v_q(0)=G_q,
    \qquad
    v_0(\tau)=1.
\]
For each \(q\geq 1\), this is a linear inhomogeneous ODE driven by
\(v_{q-1}\). The variation-of-constants formula gives
\[
    v_q(\tau)
    =
    e^{A_q\tau}G_q
    +
    C\int_0^\tau e^{A_q(\tau-u)}v_{q-1}(u)\,\dd u.
\]
Starting from \(v_0(\tau)=1\), one obtains \(v_q\) recursively. Substituting
the finite-sum expression for \(v_{q-1}\) into the integral gives
\[
\begin{aligned}
    v_q(\tau)
    &=
    e^{A_q\tau}G_q
    +
    \sum_{r=0}^{q-1}
    C^{q-r}G_r
    \sum_{i=r}^{q-1}
    \frac{
        e^{A_i\tau}-e^{A_q\tau}
    }{
        (A_i-A_q)
        \displaystyle\prod_{\substack{j=r\\ j\neq i}}^{q-1}(A_i-A_j)
    }.
\end{aligned}
\]
This expression can be rewritten as
\[
    v_q(\tau)
    =
    \sum_{r=0}^{q}
    C^{q-r}G_r
    \sum_{i=r}^{q}
    \frac{e^{A_i\tau}}
    {\displaystyle\prod_{\substack{j=r\\ j\neq i}}^{q}(A_i-A_j)}.
\]
The algebraic rearrangement follows from the standard divided-difference
identity. A detailed induction proof is provided in
Appendix~\ref{app:triangular-system-proof}. Finally, setting \(\tau=T-t\)
yields \eqref{eq:general_explicit_solution}.
\end{proof}

\begin{remark}[Coinciding coefficients]
The non-degeneracy assumption \eqref{eq:nondegenerate_Aq} is imposed only to
avoid vanishing denominators in
\eqref{eq:general_explicit_solution_expanded}. It is not needed for
well-posedness of the triangular ODE system. If some coefficients coincide, the
solution is obtained by taking the corresponding limits in the explicit
formula.

For example,
\[
    \lim_{A_1\to 0}
    \frac{e^{A_1(T-t)}-1}{A_1}
    =
    T-t.
\]
More generally,
\[
    \lim_{A_2\to A_1}
    \frac{
        e^{A_2(T-t)}-e^{A_1(T-t)}
    }{
        A_2-A_1
    }
    =
    (T-t)e^{A_1(T-t)}.
\]
Thus repeated coefficients lead to polynomial factors multiplying the
corresponding exponential terms, as in the standard repeated-root case for
linear ODE systems.
\end{remark}

\subsubsection{Degenerate case}
\label{subsubsec:explicit_triangular_degenerate}

We next consider the completely degenerate case
\[
    A_q=0,
    \qquad q=0,1,\dots,Q_0.
\]
Then \eqref{eq:PDE_general} reduces to
\begin{align}
\label{eq:general_triangular_zero_system}
\left\{
\begin{array}{ll}
\displaystyle
\partial_t w(t,q)+Cw(t,q-1)=0,
& q\geq 1,\; t<T,
\\[0.8em]
\displaystyle
w(T,q)=G_q,
& q\geq 1,
\\[0.6em]
\displaystyle
w(t,0)=G_0=1.
&
\end{array}
\right.
\end{align}

\begin{lemma}[Explicit solution in the degenerate case]
\label{lem:explicit_solution_degenerate_case}
The solution of \eqref{eq:general_triangular_zero_system} is
\begin{equation}
\label{eq:general_explicit_solution_zero}
    w(t,q)
    =
    \sum_{i=0}^q
    \frac{C^i}{i!}
    G_{q-i}(T-t)^i,
    \qquad q=0,1,\dots,Q_0.
\end{equation}
\end{lemma}

\begin{proof}[Proof]
Let \(\tau=T-t\) and define \(v_q(\tau)=w(T-\tau,q)\). Then
\eqref{eq:general_triangular_zero_system} becomes
\[
    \frac{\dd}{\dd\tau}v_q(\tau)
    =
    Cv_{q-1}(\tau),
    \qquad
    v_q(0)=G_q,
    \qquad
    v_0(\tau)=1.
\]
Thus
\[
    v_q(\tau)
    =
    G_q
    +
    C\int_0^\tau v_{q-1}(u)\,\dd u.
\]
Starting from \(v_0(\tau)=1\), induction on \(q\) gives
\[
    v_q(\tau)
    =
    \sum_{i=0}^q
    \frac{C^i}{i!}G_{q-i}\tau^i.
\]
Setting \(\tau=T-t\) yields \eqref{eq:general_explicit_solution_zero}.
\end{proof}

Under the terminal condition used in the optimal execution problems \eqref{eq:common_terminal_condition},
the degenerate solution becomes
\begin{equation}
\label{eq:general_explicit_solution_zero_execution}
    w(t,q)
    =
    \sum_{i=0}^q
    \frac{C^i}{i!}
    \exp\left\{
        -\frac{\kappa}{b}(q-i)I(q-i)
    \right\}
    (T-t)^i.
\end{equation}
In particular, for the linear per-share terminal penalty
\[
    I(q)=\alpha q,
    \qquad \alpha\geq 0,
\]
we obtain
\begin{equation}
\label{eq:general_explicit_solution_zero_linear_impact}
    w(t,q)
    =
    \sum_{i=0}^q
    \frac{C^i}{i!}
    \exp\left\{
        -\frac{\kappa}{b}\alpha(q-i)^2
    \right\}
    (T-t)^i.
\end{equation}

\subsection{Unified explicit optimal quotes under all four cases}

The coefficients \(A_q\) and \(C\) in the triangular system
\eqref{eq:PDE_general} depend on the objective criterion. The
following proposition summarizes the corresponding choices and the resulting
optimal feedback quotes.

\begin{proposition}[Explicit optimal quotes for the four cases]
\label{prop:explicit_optimal_quotes_four_cases}
Let \(w(t,q)\) be the solution of the triangular system
\eqref{eq:PDE_general}, with coefficients \(A_q\) and \(C\)
chosen according to \eqref{eq:case_coefficients_A_C}. Suppose the terminal
condition is
\[
    G_q
    =
    \exp\left\{
        -\frac{\kappa}{b}qI(q)
    \right\},
    \qquad q=0,\dots,Q_0.
\]
Then \(w(t,q)\) is given by \eqref{eq:general_explicit_solution} in the
non-degenerate case and by \eqref{eq:general_explicit_solution_zero} in the
fully degenerate case \(A_q=0\) for all \(q\).

For Cases I and II, the unconstrained optimal quote is
\begin{equation}
\label{eq:explicit_quote_cases_I_II}
    \delta^{\mathrm{unc}}(t,q)
    =
    \frac{1}{\kappa}
    \left(
        1+\log\frac{w(t,q)}{w(t,q-1)}
    \right)
    +
    \frac{a}{b}.
\end{equation}
For Cases III and IV, the unconstrained optimal quote is
\begin{equation}
\label{eq:explicit_quote_cases_III_IV}
    \delta^{\mathrm{unc}}(t,q)
    =
    \frac{1}{b\gamma}
    \log\left(\frac{\kappa+b\gamma}{\kappa}\right)
    +
    \frac{1}{\kappa}
    \log\frac{w(t,q)}{w(t,q-1)}
    +
    \frac{a}{b}.
\end{equation}
If the admissible quote set is \([\delta_{\min},\delta_{\max}]\), then the
admissible optimal quote is obtained by projection:
\begin{equation}
\label{eq:projected_explicit_quote}
    \delta^*(t,q)
    =
    \Pi_{[\delta_{\min},\delta_{\max}]}
    \bigl(\delta^{\mathrm{unc}}(t,q)\bigr).
\end{equation}
\end{proposition}

\begin{proof}
The result follows by collecting the reductions obtained in
Lemmas~\ref{lem:optimal_quote_case_I}--\ref{lem:optimal_quote_case_IV}.
These lemmas show that, after the corresponding change of variables, each value
function is characterized by the triangular system
\eqref{eq:PDE_general}, with coefficients \(A_q\) and \(C\) given in
\eqref{eq:case_coefficients_A_C} and terminal condition \eqref{eq:common_terminal_condition}.
Therefore \(w(t,q)\) is given by the explicit solution
\eqref{eq:general_explicit_solution} in the non-degenerate case and by
\eqref{eq:general_explicit_solution_zero} in the fully degenerate case.
The feedback quote formulas are precisely those obtained in the four case-wise
lemmas. 
\end{proof}

The proposition shows that the four problems differ only through the coefficient
\(A_q\), which collects the financial effects of predictable price drift,
running inventory cost, and mark-to-market risk, and through the effective
execution coefficient \(C\). Once \(w(t,q)\) has been computed from the explicit
formula, the optimal quote is obtained from the ratio \(w(t,q)/w(t,q-1)\). This
ratio represents the marginal continuation value of holding one additional
unit of inventory.

\begin{remark}[No-drift risk-neutral benchmark]
The classical no-drift optimal quote in Section 8.2 of \cite{cartea2015algorithmic} is
recovered as a special case of Case I. Indeed, take
\[
    a=0,
    \qquad
    b=1,
    \qquad
    g(s)=0,
    \qquad
    I(q)=\alpha q,
\]
where \(\alpha\geq0\) is the linear terminal impact parameter. Then the
triangular system is in the degenerate case \(A_q=0\), and
\[
    w(t,q)
    =
    \sum_{i=0}^{q}
    \frac{(\lambda e^{-1})^i}{i!}
    e^{-\kappa\alpha(q-i)^2}
    (T-t)^i.
\]
Consequently, the unconstrained optimal sell quote is
\[
    \delta^{\mathrm{unc}}(t,q)
    =
    \frac{1}{\kappa}
    \left[
        1+
        \log
        \frac{
        \displaystyle
        \sum_{i=0}^{q}
        \frac{(\lambda e^{-1})^i}{i!}
        e^{-\kappa\alpha(q-i)^2}
        (T-t)^i
        }{
        \displaystyle
        \sum_{i=0}^{q-1}
        \frac{(\lambda e^{-1})^i}{i!}
        e^{-\kappa\alpha(q-1-i)^2}
        (T-t)^i
        }
    \right].
\]
This recovers the closed-form no-drift benchmark for sell limit-order
placement derived in \cite{cartea2015algorithmic}. Hence the result in
\cite{cartea2015algorithmic} is obtained as the no-drift degenerate case of
our general solution.
\end{remark}

\begin{remark}[No-drift, no-volatility exponential-utility benchmark]
The exponential-utility optimal quote in \cite{Guant2012} is recovered as a
special case of Case III. Indeed, take
\[
    a=0,
    \qquad
    b=1,
    \qquad
    g(s)=0,
    \qquad
    \sigma=0,
    \qquad
    I(q)=\alpha,
\]
where \(\alpha\geq0\) is a constant terminal liquidation penalty per share.
Then the triangular system is in the degenerate case \(A_q=0\), and
\[
    w(t,q)
    =
    \sum_{j=0}^q
    \frac{\widehat{\lambda}^{\,j}}{j!}
    e^{-\kappa\alpha(q-j)}
    (T-t)^j,
\]
where
\[
    \widehat{\lambda}
    =
    \lambda
    \left(
        \frac{\kappa}{\kappa+\gamma}
    \right)^{\frac{\kappa}{\gamma}+1}.
\]
Consequently,
\[
\begin{aligned}
    \delta^{\mathrm{unc}}(t,q)
    &=
    -\alpha
    +
    \frac{1}{\kappa}
    \log\left(
        1+
        \frac{
            \dfrac{\widehat{\lambda}^{\,q}}{q!}(T-t)^q
        }{
            \displaystyle
            \sum_{j=0}^{q-1}
            \frac{\widehat{\lambda}^{\,j}}{j!}
            e^{-\kappa\alpha(q-j)}
            (T-t)^j
        }
    \right)
    +
    \frac{1}{\gamma}
    \log\left(1+\frac{\gamma}{\kappa}\right).
\end{aligned}
\]
This is precisely the closed-form optimal ask quote obtained in
\cite{Guant2012}, after translating notation. Hence the result of
\cite{Guant2012} is recovered as the no-drift, no-volatility degenerate case of
our general solution.
\end{remark}

\section{Verification and admissibility}
\label{sec:verification_admissibility}

In this section we justify that the controlled execution model is well posed
and that the candidate value functions obtained above solve the corresponding
control problems. As in Section~\ref{sec:solution}, we work under the
constant-coefficient assumption \(s_t\equiv s\) and \(\sigma_t\equiv\sigma\).
The verification arguments extend directly to deterministic time-dependent
coefficients, with the corresponding time-dependent generator. If the signal is
stochastic and Markovian, it must be included as an additional state variable in
the HJB equation.

Since the inventory takes values in the finite set
\(\{0,\dots,Q_0\}\), the verification argument reduces to a finite family of
equations coupled through execution jumps. The proof flow is: well-posedness of
the controlled process, characterization of Hamiltonian maximizers, predictability and admissibility of feedback controls, positivity of the transformed functions \(w\), and then verification theorem.

Throughout this section we work under Assumption~\ref{ass:standing}. We also
use the projection convention
\[
    \Pi_{[\delta_{\min},\delta_{\max}]}(y)
    :=
    \min\{\delta_{\max},\max\{\delta_{\min},y\}\},
    \qquad y\in\mathbb R.
\]

\begin{proposition}[Well-posedness on the finite inventory grid]
\label{prop:well_posedness_finite_inventory}
For each \(\delta\in\mathcal A\), there exists a unique controlled process
\((X^\delta,M,Q^\delta)\) on \([0,T]\) satisfying the dynamics introduced in
Section~\ref{sec:model}. Moreover,
\(Q_t^\delta\in\{0,\dots,Q_0\},
    \qquad 0\leq t\leq T,
\)
and \(Q^\delta\) has at most \(Q_0\) downward jumps of size one. In particular,
the objective functionals in Cases I--IV are finite.
\end{proposition}

\begin{proof}
Fix \(\delta\in\mathcal A\). Since \(\delta\) is predictable and takes values
in \([\delta_{\min},\delta_{\max}]\), the controlled intensity
\(\lambda_t^\delta
    =
    \lambda e^{-\kappa\delta_t}\mathbf 1_{\{Q_{t-}^\delta>0\}}
\)
is predictable and bounded by
\[
    \overline\lambda:=\lambda e^{-\kappa\delta_{\min}}.
\]
Thus the corresponding counting process can be constructed by the standard
time-change representation and is non-explosive on \([0,T]\). Since the
intensity is set to zero after the inventory reaches zero, \(Q^\delta\) is
absorbed at zero and has at most \(Q_0\) downward jumps. The same construction
also gives pathwise uniqueness of \(N^\delta\) and hence of \(Q^\delta\).

Given \(M\), define
\[
    X_t^\delta
    =
    x+
    \int_0^t
    (M_u-a+b\delta_u)\,\dd N_u^\delta .
\]
Since \(N^\delta\) has at most \(Q_0\) jumps, this integral is well defined and
\(X^\delta\) is uniquely determined. Therefore the controlled state process
\((X^\delta,M,Q^\delta)\) is unique up to indistinguishability.

It remains to check finiteness of the objectives. Since \(M\) is driven by
the Brownian diffusion
\(\dd M_t=g(s_t)\,\dd t+\sigma_t\,\dd W_t\) with bounded drift and
volatility, \(M\) has finite moments and finite exponential moments of linear
order on the finite horizon. Since
\(N_T^\delta\leq Q_0\) and \(\delta\) is bounded, \(X_\tau^\delta\) is a finite
sum of terms involving \(M\) and bounded controls. Moreover, \(I\) and \(J\) are
bounded on the finite inventory set \(\{0,\dots,Q_0\}\). Hence the terminal
liquidation terms, running inventory costs, and exponential utility terms are
integrable. Therefore all four objective functionals are finite.
\end{proof}

\begin{lemma}[Hamiltonian maximizers]
\label{lem:Hamiltonian_maximizers}
Let \(\Delta\in\mathbb R\) be fixed.

\begin{enumerate}[label=(\roman*)]
    \item For Cases I--II, define
    \[
        \mathcal H_1(\delta)
        :=
        \lambda e^{-\kappa\delta}
        \bigl(-a+b\delta+\Delta\bigr).
    \]
    Then \(\mathcal H_1\) has the unique unconstrained maximizer
    \begin{equation}
    \label{eq:hamiltonian_max_case_I_II}
        \delta^{\mathrm{unc}}_1
        =
        \frac1\kappa+\frac{a-\Delta}{b}.
    \end{equation}
    On the constrained interval \([\delta_{\min},\delta_{\max}]\), the unique
    maximizer is
    \[
        \delta_1^\star
        =
        \Pi_{[\delta_{\min},\delta_{\max}]}
        \bigl(\delta^{\mathrm{unc}}_1\bigr).
    \]

    \item For Cases III--IV, define
    \[
        \mathcal H_2(\delta)
        :=
        e^{-\kappa\delta}
        \left(
            1-e^{-\gamma(-a+b\delta+\Delta)}
        \right).
    \]
    Then \(\mathcal H_2\) has the unique unconstrained maximizer
    \begin{equation}
    \label{eq:hamiltonian_max_case_III_IV}
        \delta^{\mathrm{unc}}_2
        =
        \frac{a-\Delta}{b}
        +
        \frac{1}{b\gamma}
        \log\left(\frac{\kappa+b\gamma}{\kappa}\right).
    \end{equation}
    On the constrained interval \([\delta_{\min},\delta_{\max}]\), the unique
    maximizer is
    \[
        \delta_2^\star
        =
        \Pi_{[\delta_{\min},\delta_{\max}]}
        \bigl(\delta^{\mathrm{unc}}_2\bigr).
    \]
\end{enumerate}
\end{lemma}

\begin{proof}
For Cases I--II,
\[
    \mathcal H_1'(\delta)
    =
    \lambda e^{-\kappa\delta}
    \left[
        b-\kappa(-a+b\delta+\Delta)
    \right].
\]
The bracket is strictly decreasing in \(\delta\), hence \(\mathcal H_1'\)
changes sign exactly once, from positive to negative. The unique critical point
is therefore the unique global maximizer and is given by
\eqref{eq:hamiltonian_max_case_I_II}. The constrained maximizer is the
projection of this point onto the compact admissible interval.

For Cases III--IV, set
\[
    z(\delta):=e^{-\gamma(-a+b\delta+\Delta)}.
\]
Then
\[
    \mathcal H_2'(\delta)
    =
    e^{-\kappa\delta}
    \left[
        -\kappa+(\kappa+b\gamma)z(\delta)
    \right],
    \qquad
    z'(\delta)=-b\gamma z(\delta)<0.
\]
Thus the bracket is strictly decreasing and crosses zero exactly once. The
unique critical point is therefore the unique global maximizer and is given by
\eqref{eq:hamiltonian_max_case_III_IV}. Again, the constrained maximizer is the
projection onto \([\delta_{\min},\delta_{\max}]\).
\end{proof}

\begin{remark}[Predictability of feedback quotes]
Under the constant-coefficient specification considered in
Section~\ref{sec:solution}, the candidate feedback quotes are deterministic
functions of \(t\) and \(q\). Evaluated at the predictable inventory process
\(Q_{t-}^\delta\), they therefore define predictable controls. If the
constrained maximizer is used, predictability is preserved because the
projection \(\Pi_{[\delta_{\min},\delta_{\max}]}\) is continuous.
\end{remark}

\begin{remark}[Interior admissibility of the explicit formulas]
\label{rem:interior_admissibility}
The explicit \(w\)-systems derived in Section~\ref{sec:solution} are obtained
by substituting the unconstrained Hamiltonian maximizers into the HJB equations.
Therefore, for these closed-form candidates to solve the constrained control
problem under the compact admissible set \(\mathcal A\), we impose the interior
admissibility condition
\[
    \delta^{\mathrm{unc}}(t,q)
    \in
    [\delta_{\min},\delta_{\max}],
    \qquad
    t\in[0,T],\quad q=1,\dots,Q_0.
\]
Under this condition, the projection is inactive and the unconstrained feedback
quote is admissible. If this condition fails, the HJB equation remains
well-posed on the compact control set, and the pointwise maximizer is the
projected quote, but the reduced linear \(w\)-system must be replaced by the
corresponding constrained Hamiltonian equation.
\end{remark}

\begin{lemma}[Positivity of the transformed solution]
\label{lem:positivity_w}
Let \(w\) solve the triangular system
\[
    \partial_t w(t,q)+A_qw(t,q)+Cw(t,q-1)=0,
    \qquad q=1,\dots,Q_0,
\]
with
\[
    w(T,q)=G_q>0,
    \qquad
    w(t,0)=1,
    \qquad
    C>0.
\]
Then
\[
    w(t,q)>0,
    \qquad
    (t,q)\in[0,T]\times\{0,\dots,Q_0\}.
\]
In particular, the feedback quotes involving
\[
    \log\frac{w(t,q)}{w(t,q-1)}
\]
are well defined.
\end{lemma}

\begin{proof}
The assertion is proved by induction on \(q\). The case \(q=0\) follows from
\(w(t,0)=1\). Suppose \(w(\cdot,q-1)>0\). By variation of constants,
\[
    w(t,q)
    =
    e^{A_q(T-t)}G_q
    +
    C\int_t^T e^{A_q(u-t)}w(u,q-1)\,\dd u .
\]
Both terms are strictly positive for \(t\in[0,T]\), since \(G_q>0\),
\(C>0\), and \(w(u,q-1)>0\). Hence \(w(t,q)>0\). The result follows by
induction.
\end{proof}

\begin{theorem}[Verification theorem]
\label{thm:verification}
Consider one of Cases I--IV, and let \(V\) be the corresponding candidate value
function constructed in Section~\ref{sec:solution}. Assume that:
\begin{enumerate}[label=(V\arabic*)]
    \item \(V\) is \(C^{1,2}\) in \((t,M)\) for each fixed inventory state
    \(q\in\{0,\dots,Q_0\}\);
    \item \(V\) satisfies the corresponding HJB equation together with the
    terminal and boundary conditions;
    \item the feedback quote is admissible and attains the Hamiltonian
    maximizer pointwise. For the explicit linear \(w\)-systems, this is ensured
    by the interior admissibility condition in
    Remark~\ref{rem:interior_admissibility}.
\end{enumerate}
Then \(V\) coincides with the corresponding value function on
\([0,T]\times\mathbb R\times\mathbb R\times\{0,\dots,Q_0\}\). Moreover, the
feedback quote attaining the Hamiltonian maximum is optimal.
\end{theorem}

\begin{proof}
Fix an initial state \((t,x,M,q)\) and an admissible control
\(\delta\in\mathcal A\). Let
\[
    \Phi_s^\delta
    :=
    (X_s^\delta,M_s,Q_s^\delta),
    \qquad s\in[t,T],
\]
denote the controlled state process starting from
\(\Phi_t^\delta=(x,M,q)\). At an execution time, cash increases by
\(M_s-a+b\delta_s\) and inventory decreases by one unit. We denote the
corresponding post-execution state by
\[
    \Phi_{s-}^{\delta,\downarrow}
    :=
    \bigl(
        X_{s-}^\delta+M_s-a+b\delta_s,
        M_s,
        Q_{s-}^\delta-1
    \bigr).
\]
Since \(M\) is continuous in the present model, using \(M_s\) or \(M_{s-}\) is
equivalent. Let
\[
    \tau^\delta
    :=
    T\wedge \inf\{s\geq t:Q_s^\delta=0\}
\]
be the liquidation horizon starting from the state \((t,x,M,q)\). For the cases
with running inventory cost, define
\[
    A_s^\delta
    :=
    \int_t^s J(Q_u^\delta)\,\dd u,
    \qquad s\in[t,\tau^\delta].
\]

All expectations below are finite by
Proposition~\ref{prop:well_posedness_finite_inventory}: the inventory grid is
finite, the number of executions is bounded by \(Q_0\), the quote process is
bounded, and the affine and exponential-affine candidates satisfy the required
moment bounds under Assumption~\ref{ass:standing}.

We first consider Cases I and III, where there is no running inventory cost.
Applying It\^o's formula for jump-diffusions to \(V(s,\Phi_s^\delta)\) on
\([t,\tau^\delta]\) gives
\[
\begin{aligned}
    \dd V(s,\Phi_s^\delta)
    &=
    \Big[
        \partial_s V(s,\Phi_{s-}^\delta)
        +
        \mathcal L^M V(s,\Phi_{s-}^\delta)
        +
        \lambda_s^\delta
        \Big(
            V(s,\Phi_{s-}^{\delta,\downarrow})
            -
            V(s,\Phi_{s-}^\delta)
        \Big)
    \Big]\dd s
    +
    \dd\mathcal M_s^\delta ,
\end{aligned}
\]
where
\[
    \mathcal L^M V
    =
    g(s)\partial_MV
    +
    \frac12\sigma^2\partial_{MM}V,
\]
and \(\mathcal M^\delta\) is a martingale. The three drift terms arise from the
time derivative, the diffusion generator of the reference price, and the
compensator of the execution jump, respectively.

For the fixed admissible quote \(\delta_s\), the HJB equation implies that the
drift in the display above is nonpositive. Indeed, the HJB sets the same
expression with the supremum over all admissible quotes equal to zero, and
evaluating the Hamiltonian at one admissible quote cannot exceed that supremum.
Therefore \(V(s,\Phi_s^\delta)\) is a supermartingale on \([t,\tau^\delta]\) in
Cases I and III. Taking conditional expectations gives
\[
    V(t,x,M,q)
    \geq
    \mathbb E_{t,x,M,q}
    \left[
        V(\tau^\delta,\Phi_{\tau^\delta}^\delta)
    \right].
\]
The terminal condition is used if \(\tau^\delta=T\), while the boundary
condition at \(q=0\) is used if liquidation occurs before \(T\). Hence
\(V(\tau^\delta,\Phi_{\tau^\delta}^\delta)\) is exactly the terminal payoff in
Case I, or the exponential utility payoff in Case III.

In Case II, the running inventory cost is additive. Applying It\^o's formula to
\(V(s,\Phi_s^\delta)-A_s^\delta\) gives the drift
\[
\partial_s V(s,\Phi_{s-}^\delta)
+
\mathcal L^M V(s,\Phi_{s-}^\delta)
-
J(Q_{s-}^\delta)
+
\lambda_s^\delta
\Big(
    V(s,\Phi_{s-}^{\delta,\downarrow})
    -
    V(s,\Phi_{s-}^\delta)
\Big).
\]
This drift is nonpositive by the Case II HJB equation evaluated at
\(\delta_s\). Thus \(V(s,\Phi_s^\delta)-A_s^\delta\) is a supermartingale on
\([t,\tau^\delta]\), and therefore
\[
    V(t,x,M,q)
    \geq
    \mathbb E_{t,x,M,q}
    \left[
        V(\tau^\delta,\Phi_{\tau^\delta}^\delta)
        -
        A_{\tau^\delta}^\delta
    \right].
\]
Using the terminal and boundary conditions, the right-hand side is precisely the
Case II payoff under the control \(\delta\).

In Case IV, the running inventory cost is inside the exponential utility. The
HJB contains the zeroth-order term \(\gamma J(q)V\). Applying the product rule
to \(e^{\gamma A_s^\delta}V(s,\Phi_s^\delta)\) gives the drift
\[
e^{\gamma A_s^\delta}
\Big[
    \partial_s V(s,\Phi_{s-}^\delta)
    +
    \mathcal L^M V(s,\Phi_{s-}^\delta)
    +
    \gamma J(Q_{s-}^\delta)V(s,\Phi_{s-}^\delta)
    +
    \lambda_s^\delta
    \Big(
        V(s,\Phi_{s-}^{\delta,\downarrow})
        -
        V(s,\Phi_{s-}^\delta)
    \Big)
\Big].
\]
The bracket is nonpositive by the Case IV HJB equation evaluated at the
admissible quote \(\delta_s\). Hence
\(e^{\gamma A_s^\delta}V(s,\Phi_s^\delta)\) is a supermartingale on
\([t,\tau^\delta]\), and
\[
    V(t,x,M,q)
    \geq
    \mathbb E_{t,x,M,q}
    \left[
        e^{\gamma A_{\tau^\delta}^\delta}
        V(\tau^\delta,\Phi_{\tau^\delta}^\delta)
    \right].
\]
The terminal and boundary conditions give
\[
    e^{\gamma A_{\tau^\delta}^\delta}
    V(\tau^\delta,\Phi_{\tau^\delta}^\delta)
    =
    -\exp\left\{
        -\gamma
        \left(
            X_{\tau^\delta}^\delta
            +
            Q_{\tau^\delta}^\delta
            \bigl(M_{\tau^\delta}-I(Q_{\tau^\delta}^\delta)\bigr)
            -
            A_{\tau^\delta}^\delta
        \right)
    \right\}.
\]
Thus the right-hand side is exactly the Case IV performance functional under
\(\delta\).

Combining the preceding arguments, for each of the four cases we have
\[
    V(t,x,M,q)
    \geq
    \mathcal J^\delta(t,x,M,q),
\]
where \(\mathcal J^\delta(t,x,M,q)\) denotes the corresponding performance
functional under the admissible control \(\delta\). Taking the supremum over
\(\delta\in\mathcal A\), we obtain that \(V\) is an upper bound for the value
function.

Now let \(\delta^\star\) be a predictable admissible feedback quote which
attains the Hamiltonian maximum pointwise. Then the Hamiltonian evaluated at
\(\delta^\star\) equals the supremum in the HJB equation. Consequently, all the
nonpositive drift terms above vanish under \(\delta^\star\), and the
corresponding supermartingales become martingales. Therefore
\[
    V(t,x,M,q)
    =
    \mathcal J^{\delta^\star}(t,x,M,q).
\]
Hence \(V\) equals the value function, and \(\delta^\star\) is optimal.
\end{proof}

\section{Comparison of optimal quotes}
\label{sec:comparison_optimal_quotes}

The explicit formulas show that the optimal quote is determined by the ratio
\(
    w(t,q)/w(t,q-1).
\)
This ratio represents the marginal continuation value of holding one additional
unit of inventory. In Cases I and II, the unconstrained quote is
\[
    \delta^{\mathrm{unc}}(t,q)
    =
    \frac{1}{\kappa}
    +
    \frac{a}{b}
    +
    \frac{1}{\kappa}
    \log\frac{w(t,q)}{w(t,q-1)}.
\]
Thus a larger continuation-value ratio leads to a larger quote depth. For a sell
limit order, this means that the trader posts at a more favorable price but
accepts a lower execution intensity. Conversely, when the continuation value of
carrying inventory decreases, the trader quotes more aggressively in order to
increase the probability of execution. The discussion below separates three
effects: the signal-driven value of waiting, explicit inventory costs, and the
certainty-equivalent inventory cost generated by exponential utility.

Since the ratio \(w(t,q)/w(t,q-1)\) is generated by the triangular system, the
coefficient \(A_q\) provides a compact way to compare the economic forces that
shape the continuation value across the four cases. 

In Case I,
\(
    A_q=\frac{\kappa}{b}g(s)q.
\)
A positive signal drift \(g(s)>0\) increases the value of holding inventory,
because the reference price is expected to move favorably. This tends to
increase the quote depth. A negative signal has the opposite effect and induces
more aggressive liquidation.

In Case II, the coefficient becomes
\(
    A_q=\frac{\kappa}{b}\bigl(g(s)q-J(q)\bigr).
\)
The running inventory penalty therefore acts against the predictive drift. A
larger value of \(J(q)\) reduces the continuation value of holding inventory and
pushes the optimal quote closer to the reference price, increasing the execution
intensity. Economically, the trader is willing to give up some price improvement
in order to reduce inventory exposure.

In Case III, exponential utility introduces the volatility correction
\(
    -\frac{1}{2}\sigma^2\gamma q^2,
\)
so that
\[
    A_q=
    \frac{\kappa}{b}
    \left(
        g(s)q-\frac{1}{2}\sigma^2\gamma q^2
    \right).
\]
This negative correction represents the certainty-equivalent cost of
mark-to-market inventory risk. Its magnitude increases with volatility, risk
aversion, and the square of the remaining inventory, thereby lowering \(A_q\)
and reducing the continuation value of holding inventory. Hence, at the level of the reduced coefficient \(A_q\), exponential
utility induces the same type of quadratic inventory-risk correction as a
quadratic running inventory penalty.

This also explains why volatility affects the optimal quote only in the
exponential-utility cases. In Cases I and II, Brownian
fluctuations change the distribution of terminal wealth but not its conditional
expectation, and therefore \(\sigma\) drops out of the reduced equations. In
Cases III and IV, the nonlinear utility objective turns price uncertainty into
the certainty-equivalent penalty
\(
    \frac{1}{2}\sigma^2\gamma q^2
\),
which lowers the continuation value of inventory and pushes the optimal policy
toward more aggressive liquidation.

The connection between Cases II and III is especially transparent when
\(
    J(q)=\frac{1}{2}\sigma^2\gamma q^2.
\)
Then the coefficient \(A_q\) is the same in both cases:
\[
    \frac{\kappa}{b}\bigl(g(s)q-J(q)\bigr)
    =
    \frac{\kappa}{b}
    \left(
        g(s)q-\frac{1}{2}\sigma^2\gamma q^2
    \right).
\]
Thus a quadratic running inventory penalty reproduces, at the level of the
continuation-value equation, the certainty-equivalent inventory cost induced by
exponential utility. The two formulations are nevertheless not identical:
exponential utility also changes the execution coefficient in the linear
\(w\)-system from \(\lambda\exp\left\{-\frac{\kappa a}{b}-1\right\}\) to
\[
    \widehat{\lambda}
    =
    \lambda
    \left(
        \frac{\kappa}{\kappa+b\gamma}
    \right)^{\frac{\kappa}{b\gamma}+1}
    e^{-\kappa a/b},
\]
and replaces the risk-neutral quote constant \(1/\kappa\) by
\[
    \frac{1}{b\gamma}
    \log\left(\frac{\kappa+b\gamma}{\kappa}\right).
\]
Hence Cases II and III share the same inventory-risk contribution under this
choice of \(J\), but they do not generally produce identical optimal quotes. The two quotes coincide, or differ only by the
corresponding constant shift, only when these additional effects are matched or
negligible.

In the small-risk-aversion regime, the exponential-utility quote constant
satisfies
\[
    \frac{1}{b\gamma}
    \log\left(\frac{\kappa+b\gamma}{\kappa}\right)
    =
    \frac{1}{\kappa}
    -
    \frac{b\gamma}{2\kappa^2}
    +
    O(\gamma^2),
    \qquad \gamma\to0.
\]
Thus the exponential-utility quote approaches the risk-neutral quote as
\(\gamma\to0\). At the same time, the volatility penalty
\(
    \frac{1}{2}\sigma^2\gamma q^2
\)
is precisely the leading-order certainty-equivalent cost of inventory risk.
This explains why quadratic inventory penalties are frequently used as a
tractable proxy for risk aversion in optimal execution models.

Case IV combines the two inventory-risk effects:
\[
    A_q=
    \frac{\kappa}{b}
    \left(
        g(s)q
        -
        \frac{1}{2}\sigma^2\gamma q^2
        -
        J(q)
    \right).
\]
The running inventory cost and the exponential-utility volatility penalty are
additive in the coefficient \(A_q\). Financially, the trader is penalized both
for explicitly carrying inventory over time and for being exposed to
mark-to-market price uncertainty. Both effects reduce the value of waiting and
therefore push the optimal policy toward more aggressive liquidation.

\section{Long-horizon asymptotics of the optimal quotes}
\label{sec:long_horizon_asymptotics}

We now discuss the behaviour of the optimal quotes as the trading horizon becomes large. Since the feedback quotes are expressed through the ratio
\(
    w(t,q)/w(t,q-1),
\)
their long-horizon behaviour is determined by the asymptotics of the triangular system. It is convenient to work with the remaining time
\(
    \tau:=T-t.
\)
Thus the relevant asymptotic regime is \(\tau\to\infty\). This includes the initial quote \(t=0\) as \(T\to\infty\). If instead the remaining time \(T-t\) stays fixed, then no long-horizon asymptotic effect arises.

Recall that, with
\[
    v_q(\tau):=w(T-\tau,q),
\]
the common triangular system can be written as
\[
    \frac{\dd}{\dd\tau}v_q(\tau)
    =
    A_qv_q(\tau)+Cv_{q-1}(\tau),
    \qquad
    v_q(0)=G_q,
    \qquad
    v_0(\tau)=1.
\]
For each \(q\), define
\[
    \overline A_q
    :=
    \max_{0\le i\le q} A_i,
    \qquad A_0:=0.
\]

\subsection{Non-degenerate case}
\label{subsec:asymptotic_nondegenerate}

Assume that the coefficients are non-degenerate:
\[
    A_i\neq A_j,
    \qquad i\neq j.
\]
Then the maximum defining \(\overline A_q\) is attained at a unique index.

\begin{proposition}[Long-horizon growth rate in the non-degenerate case]
\label{prop:asymptotic_quotes_nondegenerate}
Assume that
\(
    A_i\neq A_j, i\neq j,
\)
for \(i,j\in\{0,\dots,Q_0\}\), and that \(C>0\), \(G_q>0\) for
\(q=0,\dots,Q_0\). Then, for each fixed \(q\geq1\),
\begin{equation}
\label{eq:log_ratio_growth_rate_nondegenerate}
    \lim_{T-t\to\infty}
    \frac{1}{T-t}
    \log\frac{w(t,q)}{w(t,q-1)}
    =
    \overline A_q-\overline A_{q-1}.
\end{equation}
Consequently, for all four cases,
\begin{equation}
\label{eq:quote_growth_rate_nondegenerate}
    \lim_{T-t\to\infty}
    \frac{\delta^{\mathrm{unc}}(t,q)}{T-t}
    =
    \frac{\overline A_q-\overline A_{q-1}}{\kappa}.
\end{equation}
In particular, if
\(
    \overline A_q>\overline A_{q-1},
\)
then the unconstrained quote grows linearly in the remaining horizon. If
\(
    \overline A_q=\overline A_{q-1},
\)
then
\[
    \lim_{T-t\to\infty}
    \frac{\delta^{\mathrm{unc}}(t,q)}{T-t}
    =
    0.
\]
\end{proposition}

\begin{proof}
Let \(\tau:=T-t\) and \(v_q(\tau):=w(T-\tau,q)\). By
Proposition~\ref{prop:explicit_solution_triangular_system}, \(v_q\) is a finite
linear combination of exponentials,
\[
    v_q(\tau)=\sum_{i=0}^q B_{q,i}e^{A_i\tau},
\]
where \(B_{q,i}\) denotes the coefficient of \(e^{A_i\tau}\). 
Since the \(A_i\)'s are distinct, the maximum
\(\overline A_q=\max_{0\leq i\leq q}A_i\) is attained at a unique index; denote
it by \(i_q\). Set \(B_q:=B_{q,i_q}\). By positivity of the triangular system,
together with \(C>0\) and \(G_r>0\), the dominant coefficient satisfies
\(B_q>0\). Therefore
\[
    v_q(\tau)=B_q e^{\overline A_q\tau}(1+o(1)),
    \qquad \tau\to\infty.
\]
Similarly, if \(i_{q-1}\) is the unique maximizer of
\(\max_{0\leq i\leq q-1}A_i\) and \(B_{q-1}:=B_{q-1,i_{q-1}}>0\), then
\[
    v_{q-1}(\tau)
    =
    B_{q-1}e^{\overline A_{q-1}\tau}(1+o(1)).
\]
Hence
\[
    \log\frac{v_q(\tau)}{v_{q-1}(\tau)}
    =
    \bigl(\overline A_q-\overline A_{q-1}\bigr)\tau+O(1).
\]
Dividing by \(\tau\) and using \(\tau=T-t\) gives
\eqref{eq:log_ratio_growth_rate_nondegenerate}.

For Cases I and II,
\[
    \delta^{\mathrm{unc}}(t,q)
    =
    \frac{1}{\kappa}
    +
    \frac{a}{b}
    +
    \frac{1}{\kappa}
    \log\frac{w(t,q)}{w(t,q-1)}.
\]
For Cases III and IV,
\[
    \delta^{\mathrm{unc}}(t,q)
    =
    \frac{1}{b\gamma}
    \log\left(\frac{\kappa+b\gamma}{\kappa}\right)
    +
    \frac{a}{b}
    +
    \frac{1}{\kappa}
    \log\frac{w(t,q)}{w(t,q-1)}.
\]
In both cases, the terms not involving \(w(t,q)/w(t,q-1)\) are finite constants
independent of the remaining time \(T-t\). Hence, after division by \(T-t\),
these constants vanish:
\[
    \lim_{T-t\to\infty}
    \frac{1}{T-t}
    \left(
        \frac{1}{\kappa}
        +
        \frac{a}{b}
    \right)
    =
    0,
\]
and similarly
\[
    \lim_{T-t\to\infty}
    \frac{1}{T-t}
    \left[
        \frac{1}{b\gamma}
        \log\left(\frac{\kappa+b\gamma}{\kappa}\right)
        +
        \frac{a}{b}
    \right]
    =
    0.
\]
Therefore, using \eqref{eq:log_ratio_growth_rate_nondegenerate}, we obtain
\[
    \lim_{T-t\to\infty}
    \frac{\delta^{\mathrm{unc}}(t,q)}{T-t}
    =
    \frac{1}{\kappa}
    \lim_{T-t\to\infty}
    \frac{1}{T-t}
    \log\frac{w(t,q)}{w(t,q-1)}
    =
    \frac{\overline A_q-\overline A_{q-1}}{\kappa},
\]
which proves \eqref{eq:quote_growth_rate_nondegenerate}.
\end{proof}

The proposition shows that the long-horizon growth of the optimal quote is
controlled by the increment
\(
    \overline A_q-\overline A_{q-1}.
\)

If \( \overline A_q>\overline A_{q-1}, \) then adding the \(q\)-th unit of inventory introduces a new dominant growth mode in the continuation value. The unconstrained quote then grows linearly in the remaining horizon. Financially, the trader is increasingly willing to wait for a better execution price, because the marginal continuation value of holding one additional unit of inventory grows over long horizons.

If instead \( \overline A_q=\overline A_{q-1}, \) then there is no linear growth in the optimal quote. In the non-degenerate case, the quote remains bounded as \(T-t\to\infty\), with its limiting value determined by the ratio \(B_q/B_{q-1}\). Financially, this means that the additional unit of inventory does not create a new dominant long-run continuation value, so the trader does not become increasingly patient as the horizon grows.

If the admissible quote set is bounded, then
\[
    \delta^*(t,q)
    =
    \Pi_{[\delta_{\min},\delta_{\max}]}
    \bigl(\delta^{\mathrm{unc}}(t,q)\bigr).
\]
Hence, whenever the unconstrained quote diverges to \(+\infty\), the admissible
quote converges to the upper bound \(\delta_{\max}\).

\subsection{Fully degenerate case}
\label{subsec:asymptotic_degenerate}

We next consider the fully degenerate case
\[
    A_q=0,
    \qquad q=0,1,\dots,Q_0.
\]
In this case, the continuation value grows polynomially rather than
exponentially in the remaining time.

\begin{proposition}[Logarithmic long-horizon growth in the fully degenerate case]
\label{prop:asymptotic_quotes_degenerate}
Assume that
\(
    A_q=0, q=0,1,\dots,Q_0,
\)
with \(C>0\) and \(G_0=1\). Then, for each fixed \(q\geq1\),
\begin{equation}
\label{eq:degenerate_log_ratio_growth_rate}
    \lim_{T-t\to\infty}
    \frac{1}{\log(T-t)}
    \log\frac{w(t,q)}{w(t,q-1)}
    =
    1.
\end{equation}
Consequently, for all four cases,
\begin{equation}
\label{eq:degenerate_quote_growth_rate}
    \lim_{T-t\to\infty}
    \frac{\delta^{\mathrm{unc}}(t,q)}{\log(T-t)}
    =
    \frac{1}{\kappa}.
\end{equation}
\end{proposition}

\begin{proof}
In the fully degenerate case, Lemma~\ref{lem:explicit_solution_degenerate_case}
gives
\[
    w(t,q)
    =
    \sum_{i=0}^q
    \frac{C^i}{i!}
    G_{q-i}(T-t)^i.
\]
Since \(C>0\) and \(G_0=1\), the leading term as \(T-t\to\infty\) is the term
with \(i=q\). Hence
\[
    w(t,q)
    =
    \frac{C^q}{q!}(T-t)^q
    +
    O\bigl((T-t)^{q-1}\bigr).
\]
Similarly,
\[
    w(t,q-1)
    =
    \frac{C^{q-1}}{(q-1)!}(T-t)^{q-1}
    +
    O\bigl((T-t)^{q-2}\bigr).
\]
Therefore,
\[
    \frac{w(t,q)}{w(t,q-1)}
    =
    \frac{C}{q}(T-t)(1+o(1)).
\]
Taking logarithms yields
\[
    \log\frac{w(t,q)}{w(t,q-1)}
    =
    \log(T-t)
    +
    \log\frac{C}{q}
    +
    o(1).
\]
Dividing by \(\log(T-t)\) gives
\eqref{eq:degenerate_log_ratio_growth_rate}.

For Cases I and II,
\[
    \delta^{\mathrm{unc}}(t,q)
    =
    \frac{1}{\kappa}
    +
    \frac{a}{b}
    +
    \frac{1}{\kappa}
    \log\frac{w(t,q)}{w(t,q-1)}.
\]
For Cases III and IV,
\[
    \delta^{\mathrm{unc}}(t,q)
    =
    \frac{1}{b\gamma}
    \log\left(\frac{\kappa+b\gamma}{\kappa}\right)
    +
    \frac{a}{b}
    +
    \frac{1}{\kappa}
    \log\frac{w(t,q)}{w(t,q-1)}.
\]
In both cases, all terms except the logarithmic continuation-value ratio are
finite constants independent of \(T-t\). After division by \(\log(T-t)\), these
constants vanish. Hence
\[
    \lim_{T-t\to\infty}
    \frac{\delta^{\mathrm{unc}}(t,q)}{\log(T-t)}
    =
    \frac{1}{\kappa}
    \lim_{T-t\to\infty}
    \frac{1}{\log(T-t)}
    \log\frac{w(t,q)}{w(t,q-1)}
    =
    \frac{1}{\kappa}.
\]
This proves \eqref{eq:degenerate_quote_growth_rate}.
\end{proof}

Thus, in the fully degenerate case, the optimal quote grows only logarithmically
with the remaining horizon. This contrasts with the non-degenerate case, where
a new dominant exponential growth mode can lead to linear growth in \(T-t\).
Financially, when all \(A_q\)'s vanish, there is no drift, running-cost, or
risk-aversion contribution generating exponential growth or decay in the
continuation value. The only long-horizon effect comes from the increasing
number of possible execution opportunities, which produces polynomial growth in
\(w(t,q)\) and hence logarithmic growth in the optimal quote.

\subsection{Implications for the four cases in the non-degenerate case}
\label{subsec:implications_four_cases_nondegenerate}

We now interpret the non-degenerate long-horizon growth rate in the four
optimization criteria. In the fully degenerate case \(A_q=0\) for all
\(q=0,\dots,Q_0\), Proposition~\ref{prop:asymptotic_quotes_degenerate} shows
that the unconstrained quote grows logarithmically with universal rate
\(1/\kappa\), independently of the particular form of \(A_q\). Therefore, the
case-specific financial effects of signal drift, running inventory cost, and
price risk enter the long-horizon growth rate only in the non-degenerate case,
which is the focus of this subsection.

The four cases differ through the coefficient \(A_q\). It is useful to write
\[
    A_q=\frac{\kappa}{b}\Psi(q),
\]
where \(\Psi(q)\) represents the effective continuation value of holding
inventory level \(q\). In the four cases,
\[
\begin{array}{lll}
\text{Case I:}
&
\displaystyle
\Psi(q)=g(s)q,
\\[0.8em]
\text{Case II:}
&
\displaystyle
\Psi(q)=g(s)q-J(q),
\\[0.8em]
\text{Case III:}
&
\displaystyle
\Psi(q)=g(s)q-\frac{1}{2}\sigma^2\gamma q^2,
\\[0.8em]
\text{Case IV:}
&
\displaystyle
\Psi(q)=g(s)q-\frac{1}{2}\sigma^2\gamma q^2-J(q).
\end{array}
\]
Assume throughout this subsection that
\(A_i\neq A_j,\) and \(i\neq j.\)
Then Proposition~\ref{prop:asymptotic_quotes_nondegenerate} gives
\[
    \lim_{T-t\to\infty}
    \frac{\delta^{\mathrm{unc}}(t,q)}{T-t}
    =
    \frac{\overline A_q-\overline A_{q-1}}{\kappa}.
\]
Since
\[
    \overline A_q
    =
    \frac{\kappa}{b}
    \max_{0\le i\le q}\Psi(i),
\]
the growth rate can be written as
\begin{equation}
\label{eq:general_growth_rate_Psi}
    \frac{\overline A_q-\overline A_{q-1}}{\kappa}
    =
    \frac{1}{b}
    \left(
        \max_{0\le i\le q}\Psi(i)
        -
        \max_{0\le i\le q-1}\Psi(i)
    \right).
\end{equation}
Thus the unconstrained quote grows linearly in the remaining horizon precisely
when the additional inventory level \(q\) creates a new maximum of the effective
continuation value \(\Psi\). If it does not, the quote has zero linear growth
rate in the non-degenerate case.

\begin{enumerate}
\item In Case I,
\(
    \Psi(q)=g(s)q.
\)
\begin{itemize}
\item If \(g(s)>0\), then \(\Psi(q)\) is increasing, and
\[
    \max_{0\le i\le q}\Psi(i)-\max_{0\le i\le q-1}\Psi(i)
    =
    g(s).
\]
Thus
\[
    \lim_{T-t\to\infty}
    \frac{\delta^{\mathrm{unc}}(t,q)}{T-t}
    =
    \frac{g(s)}{b}.
\]
Equivalently,
\[
    \delta^{\mathrm{unc}}(t,q)
    =
    \frac{g(s)}{b}(T-t)+O(1),
    \qquad T-t\to\infty.
\]
A positive drift signal makes holding inventory attractive, so the trader becomes more patient and quotes further away from the reference price as the horizon increases.

\item If \(g(s)<0\), then \(\Psi(q)<0\) for all \(q\geq1\), while \(\Psi(0)=0\). Hence the maximum is attained at \(0\), and
\[
    \overline A_q=\overline A_{q-1}=0.
\]
The optimal quote remains bounded as \(T-t\to\infty\). Financially, a negative drift signal makes waiting unattractive, so the trader does not increase the quote depth indefinitely even over a long horizon.
\end{itemize}
\item In Case II,
\(
    \Psi(q)=g(s)q-J(q).
\)
The running inventory penalty reduces the effective value of waiting. Even if \(g(s)>0\), a sufficiently large penalty \(J(q)\) can prevent \(\Psi(q)\) from creating new record values as \(q\) increases. In that case the optimal quote remains bounded. If, however, the favourable drift dominates the running inventory cost at the relevant inventory level, then \(\Psi(q)\) may attain a new maximum at \(q\), and the optimal quote grows linearly with the remaining horizon. Thus Case II balances the benefit of waiting for favourable price movements against the cost of carrying inventory.

\item In Case III,
\(
    \Psi(q)=g(s)q-\frac{1}{2}\sigma^2\gamma q^2.
\)
The second term is the certainty-equivalent cost of mark-to-market inventory risk under exponential utility. It is quadratic in \(q\), increasing in volatility \(\sigma\), and increasing in risk aversion \(\gamma\). Therefore, even when \(g(s)>0\), the effective value \(\Psi(q)\) eventually decreases for large \(q\). The quote grows linearly with the horizon only at inventory levels where adding the \(q\)-th unit creates a new maximum of \(\Psi\). Once the volatility-risk penalty dominates, the optimal quote becomes bounded. Financially, risk aversion prevents the trader from becoming indefinitely patient at large inventory levels.

\item In Case IV,
\(
    \Psi(q)
    =
    g(s)q
    -
    \frac{1}{2}\sigma^2\gamma q^2
    -
    J(q).
\)
The running inventory penalty and the exponential-utility volatility penalty are additive. Both reduce the continuation value of holding inventory and therefore make the trader more aggressive. Consequently, Case IV is the most conservative of the four cases: it is least likely that \(\Psi(q)\) creates a new maximum at high inventory levels, and therefore least likely that the unconstrained quote grows linearly with the horizon.
\end{enumerate}

\section{Numerical Experiments}\label{sec:numerics}

We now present numerical experiments for the four cases studied above. Unless
otherwise stated, the parameter values are given in
Table~\ref{tab:baseline-parameters}. The experiments vary the signal
specification, the inventory level, and the main model parameters in order to
illustrate their effects on the optimal quotes. More specifically, we compare
constant and time-dependent signals, different inventory levels, different
terminal liquidation penalties \(I(q)=\alpha q\), different running inventory
penalties \(J(q)=\beta q^2\), different volatility and risk-aversion parameters
\((\sigma,\gamma)\), and different terminal horizons \(T\). The purpose of the experiments is to illustrate the qualitative behaviour and
parameter sensitivity of the explicit optimal quotes.

\begin{table}[t]
\centering
\begin{tabular}{ll}
\hline
Parameter & Baseline value / description \\
\hline
$T$ & $30$, except in the horizon-comparison experiment \\
$q$ & $q=1$ in the baseline plots; $q=2$ in inventory-comparison plots \\
$a,b$ & $a=0$, $b=1$ \\
$\lambda$ & $5/6$ \\
$\kappa$ & $1000$ \\
$\alpha$ & $\alpha\in\{0.001,0.005\}$ for $I(q)=\alpha q$ \\
$\beta$ & $\beta\in\{0.0001,0.0005,0.001\}$ for $J(q)=\beta q^2$ \\
$\gamma$ & $\gamma\in\{0.01,0.05\}$ in Cases III and IV \\
$\sigma$ & $\sigma\in\{0.01,0.1\}$ in Cases III and IV \\
Signal levels & $g(s)\in\{-3,-2,-1,0,1,2,3\}\times 10^{-4}$ in the line plots \\
Constant signal & $g(s,t)=g(s)$ \\
Time-decay signal & $g(s,t)=g(s)e^{-0.01t}$ \\
Delayed-decay signal & $g(s,t)=g(s)e^{-0.01|t-10|}$ \\
\hline
\end{tabular}
\caption{Parameter values used in the numerical experiments.}
\label{tab:baseline-parameters}
\end{table}

The constant-signal figures are computed from the explicit optimal quotes in
Proposition~\ref{prop:explicit_optimal_quotes_four_cases}. The figures show
the unconstrained explicit quotes. If finite quote bounds are imposed in
practice, the constrained optimal quote is obtained by projection onto
\([\delta_{\min},\delta_{\max}]\). In all cases, the
optimal quote is computed from the ratio \(w(t,q)/w(t,q-1)\) with coefficients
chosen according to \eqref{eq:case_coefficients_A_C}. The precise additive terms differ between Cases I--II and
Cases III--IV, as given in
Proposition~\ref{prop:explicit_optimal_quotes_four_cases}. For 
time-dependent signal experiments, we use the implementation of frozen-signals to obtain optimal signal-adaptive quotes: at each time \(t\),
the instantaneous signal level \(g(s,t)\) is inserted into the
constant-coefficient formula.

The numerical results show that the optimal quotes respond systematically to
the signal, inventory, time-to-maturity, and the model parameters. Across all
four cases, favourable signals shift the optimal quotes upward, while adverse
signals shift them downward. Thus, when the signal predicts a favourable price
movement, the trader quotes more patiently; when the signal is adverse, the
trader quotes more aggressively to increase the probability of execution. The
separation between signal levels is most visible away from maturity. Close to
the terminal time, the quotes move closer together and become more aggressive,
reflecting the increasing need to complete the liquidation.

Figures~\ref{fig:caseI-q1-price-impact} and
\ref{fig:caseI-q2-price-impact} illustrate this effect in Case I under constant
signals. Positive signals lead to larger quotes, whereas negative signals lead
to lower quotes. Increasing the terminal liquidation penalty parameter
$\alpha$ shifts the quote curves downward, since holding inventory until
maturity becomes more costly. The same qualitative pattern appears for both
$q=1$ and $q=2$, although the inventory level changes the scale of the quote
profiles. This confirms that even in the risk-neutral benchmark, both the signal
and the liquidation penalty have a visible impact on the optimal quoting strategy.

Figures~\ref{fig:caseI-decay-inventory} and
\ref{fig:caseI-delayed-decay-alpha001-inventory} show that the strategy also
adapts to the time-dependent signal. Under the exponential decay
specification $g(s,t)=g(s)e^{-0.01t}$, the signal has its largest effect early
in the horizon and gradually weakens over time. Under the delayed-decay
specification $g(s,t)=g(s)\xi(t)$, with
$\xi(t)=\exp(-0.01|t-10|)$, the quote profiles reflect the changing signal
strength around the delay point. These experiments show that the framework can
capture not only different signal levels, but also time-varying signal
strength.

Figures~\ref{fig:caseII-q1-running-penalty} and
\ref{fig:caseII-q2-running-penalty} show the effect of the running inventory
penalty in Case II. Increasing $\beta$ in $J(q)=\beta q^2$ makes the quotes
more aggressive, because carrying inventory during the trading period becomes
more costly. This effect is stronger for larger inventory, as expected from the
quadratic form of the running penalty. Nevertheless, the ordering across signal
levels is preserved: favourable signals lead to higher quotes, while adverse
signals lead to lower quotes.

Figures~\ref{fig:caseIII-q1-risk-aversion}--\ref{fig:caseIV-q2-volatility}
show the impact of price risk, risk aversion, and volatility. In Case III,
increasing the risk-aversion parameter $\gamma$ reduces the value of waiting
and leads to more aggressive quotes. In Case IV, the effects of running
inventory costs and price risk act together. Larger running penalties, higher
volatility, and stronger risk aversion all push the strategy toward faster
execution. Positive signals partly offset this pressure by increasing the value
of waiting, while negative signals reinforce the need to execute more
aggressively.

Figure~\ref{fig:caseIV-q2-terminal-time} highlights the role of the trading
horizon. A shorter horizon leads to more aggressive quotes, since there is less
time to complete the liquidation. With a longer horizon, the trader can quote
more patiently, and the effect of the signal remains visible over a larger part
of the trading period. This confirms that signal information is most valuable
when there is sufficient time to exploit it, while terminal liquidation pressure
dominates near maturity.

Finally, Figure~\ref{fig:caseIII-heatmaps-signal-levels} illustrates the joint
dependence of the optimal quote on time and inventory in Case III under four
constant signal levels. The heatmaps show that the signal affects the full
time--inventory structure of the policy, rather than only shifting quotes at a
fixed inventory level. Relative to the zero-signal benchmark, a negative signal
shifts the quote surface downward and leads to more aggressive execution, while
positive signals shift the surface upward and allow the trader to quote more
patiently. The effect becomes stronger as the signal level increases. Overall,
the experiments show that ignoring predictive signals can lead to materially
different quoting decisions, especially away from maturity.

\begin{figure}[t]
    \centering
    \begin{subfigure}[t]{0.49\textwidth}
        \centering
        \includegraphics[width=\textwidth]{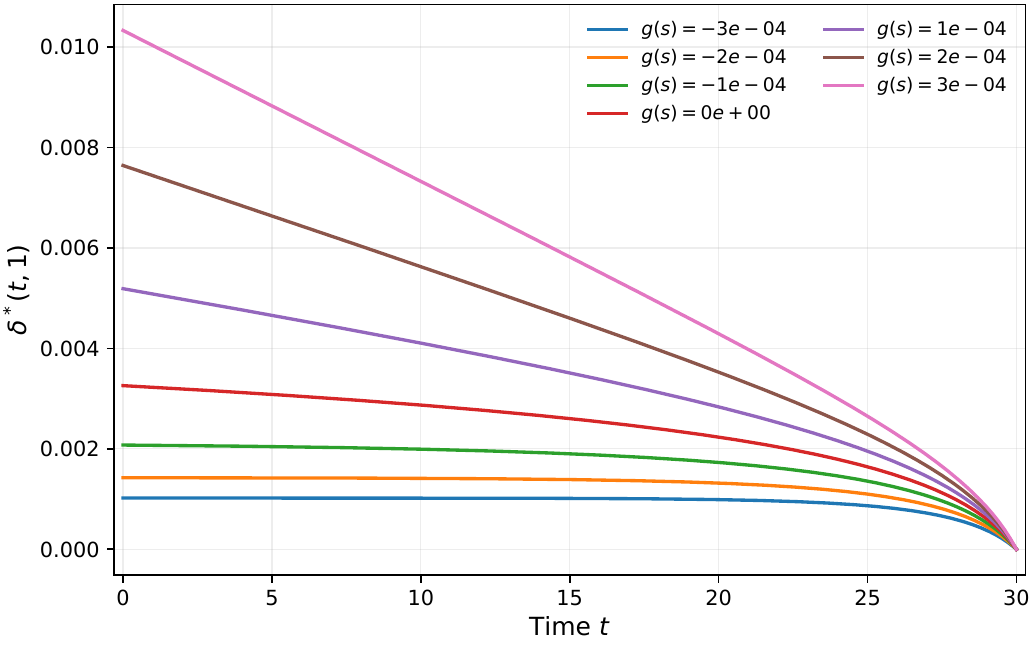}
        \caption{$\alpha=0.001$.}
        \label{fig:caseI-q1-alpha001}
    \end{subfigure}
    \hfill
    \begin{subfigure}[t]{0.49\textwidth}
        \centering
        \includegraphics[width=\textwidth]{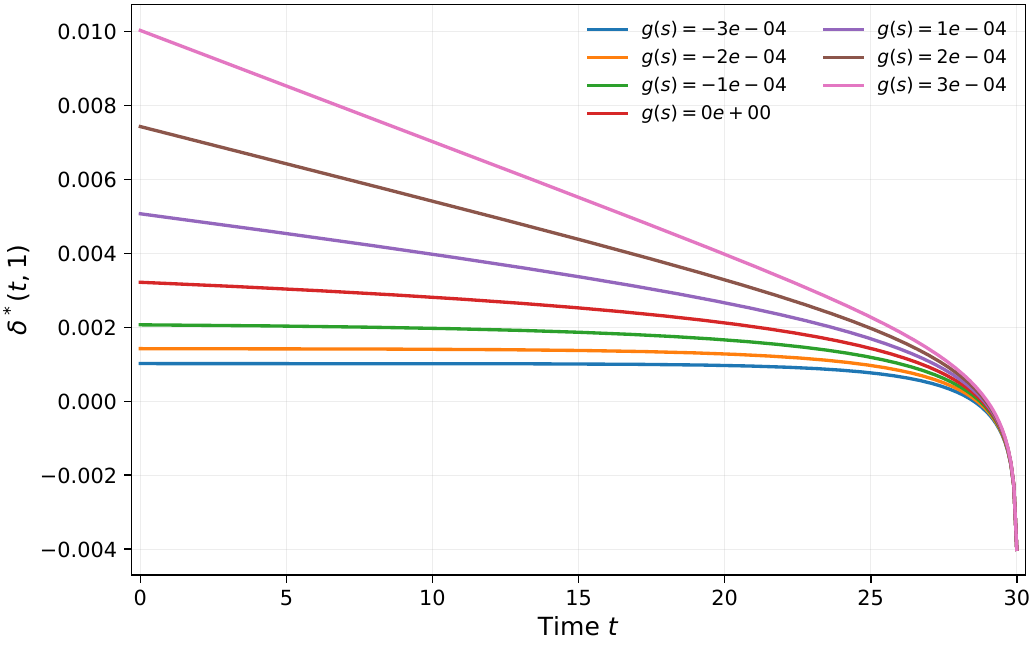}
        \caption{$\alpha=0.005$.}
        \label{fig:caseI-q1-alpha005}
    \end{subfigure}

    \caption{Case I: optimal quotes $\delta^\star(t,1)$ for different signal levels. The panels compare two values of the liquidation penalty parameter $\alpha$ in $I(q)=\alpha q$.}
    \label{fig:caseI-q1-price-impact}
\end{figure}

\begin{figure}[t]
    \centering
    \begin{subfigure}[t]{0.49\textwidth}
        \centering
        \includegraphics[width=\textwidth]{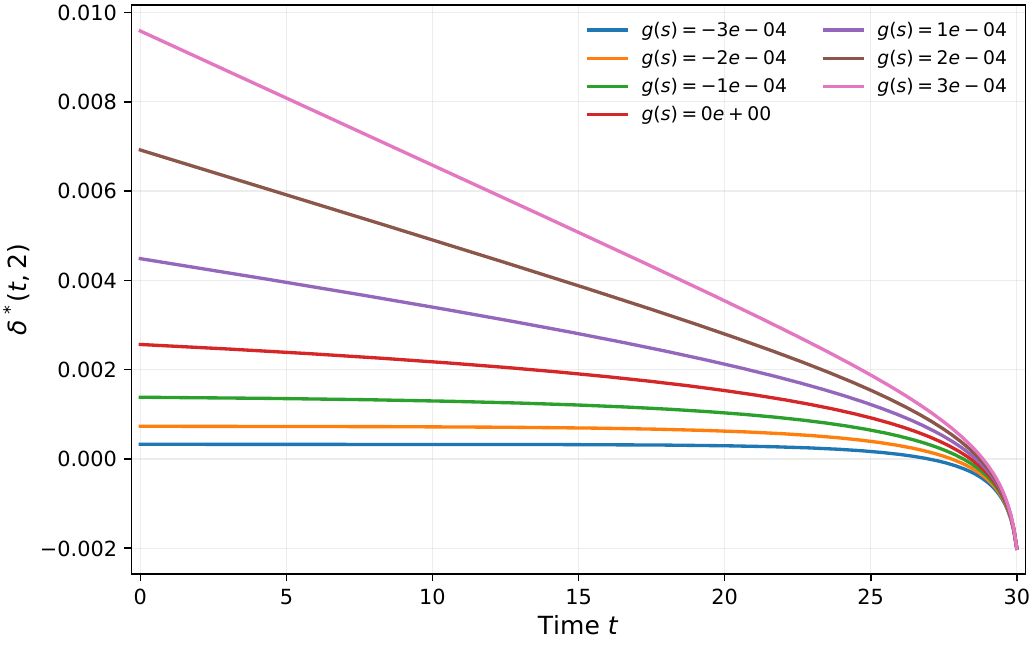}
        \caption{$\alpha=0.001$.}
        \label{fig:caseI-q2-alpha001}
    \end{subfigure}
    \hfill
    \begin{subfigure}[t]{0.49\textwidth}
        \centering
        \includegraphics[width=\textwidth]{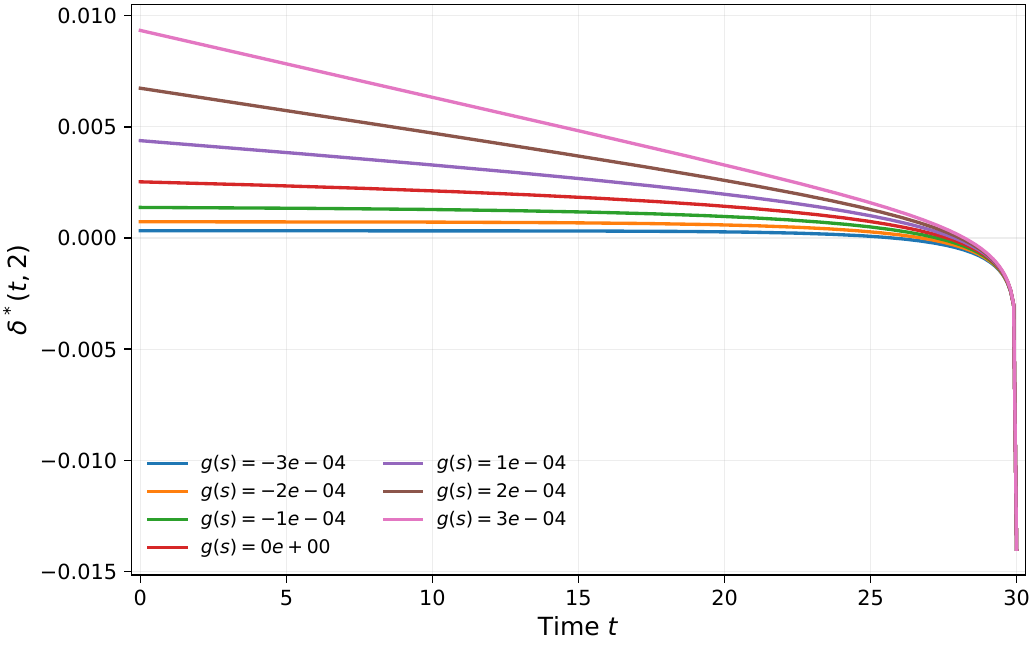}
        \caption{$\alpha=0.005$.}
        \label{fig:caseI-q2-alpha005}
    \end{subfigure}

    \caption{Case I: optimal quotes $\delta^\star(t,2)$ for different signal levels. The panels compare two values of the liquidation penalty parameter $\alpha$ in $I(q)=\alpha q$.}
    \label{fig:caseI-q2-price-impact}
\end{figure}

\begin{figure}[t]
    \centering
    \begin{subfigure}[t]{0.49\textwidth}
        \centering
        \includegraphics[width=\textwidth]{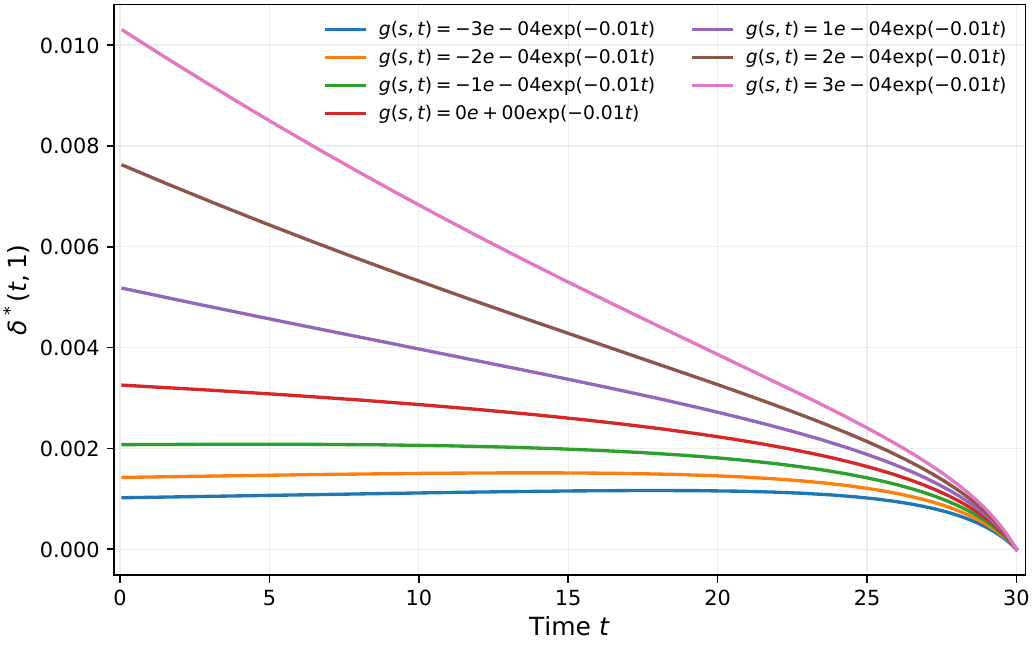}
        \caption{$q=1$.}
        \label{fig:caseI-q1-decay}
    \end{subfigure}
    \hfill
    \begin{subfigure}[t]{0.49\textwidth}
        \centering
        \includegraphics[width=\textwidth]{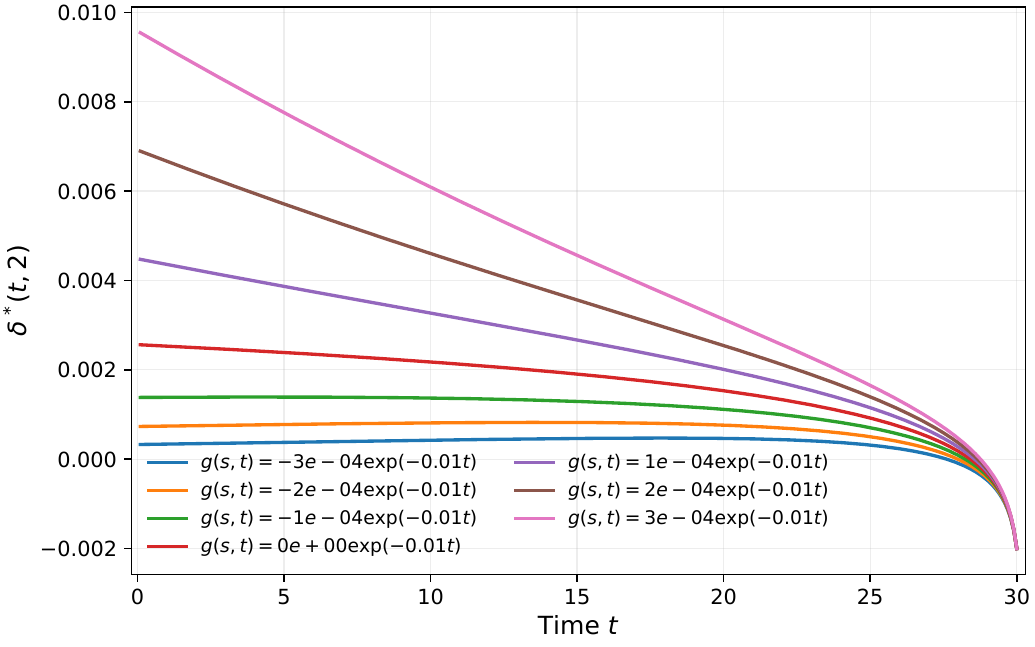}
        \caption{$q=2$.}
        \label{fig:caseI-q2-decay}
    \end{subfigure}
    \caption{Case I: optimal quotes for different signal levels under the exponential signal decay $g(s,t)=g(s)e^{-0.01t}$, with $\alpha=0.001$. The panels compare two initial inventory levels.}
    \label{fig:caseI-decay-inventory}
\end{figure}

\begin{figure}[t]
    \centering
    \begin{subfigure}[t]{0.49\textwidth}
        \centering
        \includegraphics[width=\textwidth]{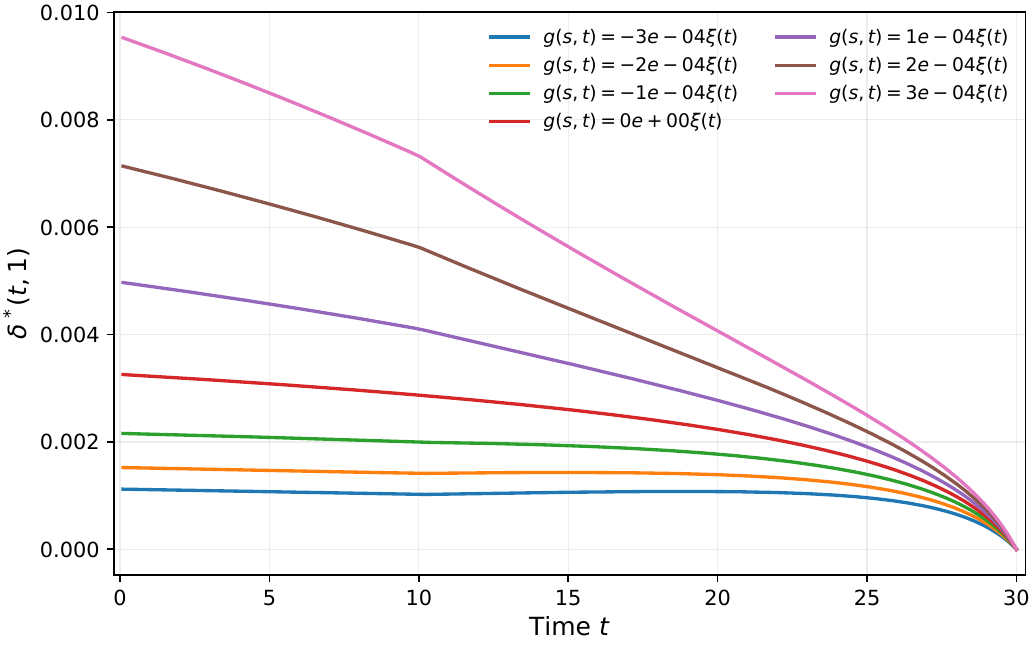}
        \caption{$q=1$.}
        \label{fig:caseI-q1-alpha001-delayed-decay}
    \end{subfigure}
    \hfill
    \begin{subfigure}[t]{0.49\textwidth}
        \centering
        \includegraphics[width=\textwidth]{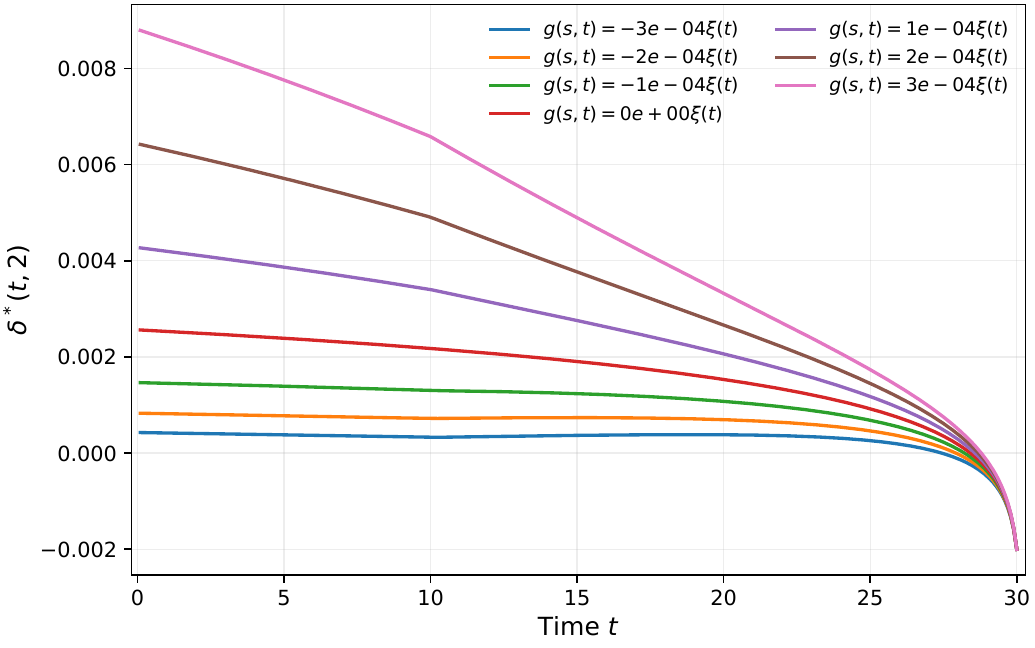}
        \caption{$q=2$.}
        \label{fig:caseI-q2-alpha001-delayed-decay}
    \end{subfigure}

    \caption{Case I: optimal quotes for different signal levels under the time-dependent signal $g(s,t)=g(s)\xi(t)$, with $\alpha=0.001$, where $\xi(t)=\exp(-0.01|t-10|)$. The panels compare two initial inventory levels.}
    \label{fig:caseI-delayed-decay-alpha001-inventory}
\end{figure}

\begin{figure}[t]
    \centering
    \begin{subfigure}[t]{0.49\textwidth}
        \centering
        \includegraphics[width=\textwidth]{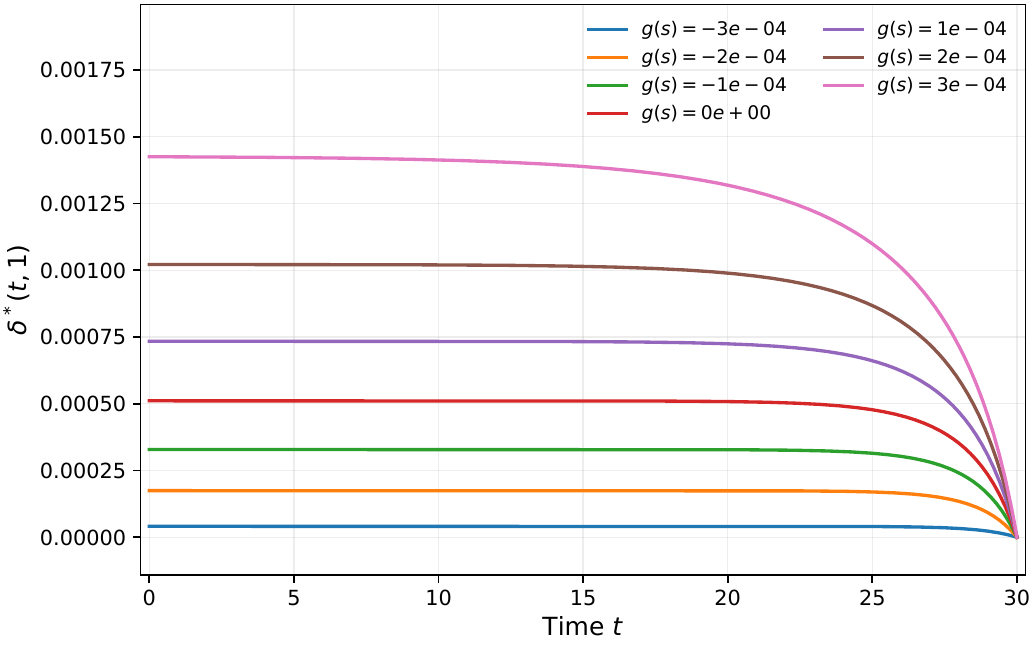}
        \caption{$\beta=0.0005$.}
        \label{fig:caseII-q1-beta0005}
    \end{subfigure}
    \hfill
    \begin{subfigure}[t]{0.49\textwidth}
        \centering
        \includegraphics[width=\textwidth]{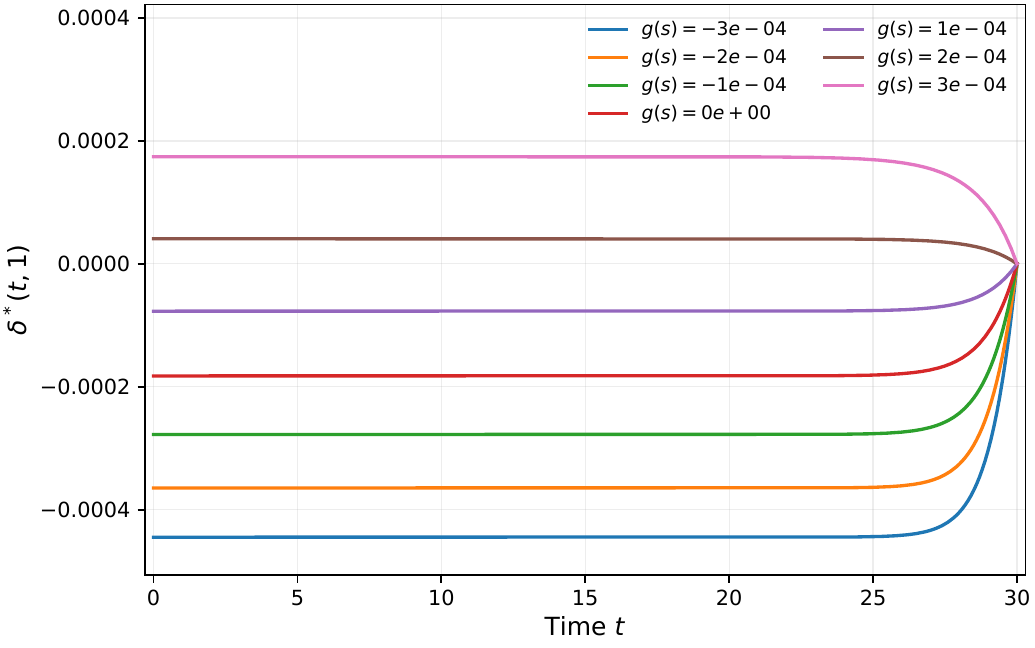}
        \caption{$\beta=0.001$.}
        \label{fig:caseII-q1-beta001}
    \end{subfigure}

    \caption{Case II: optimal quotes $\delta^\star(t,1)$ for different signal levels. The panels compare two values of the running inventory penalty parameter $\beta$ in $J(q)=\beta q^2$, with $\alpha=0.001$.}
    \label{fig:caseII-q1-running-penalty}
\end{figure}

\begin{figure}[t]
    \centering
    \begin{subfigure}[t]{0.49\textwidth}
        \centering
        \includegraphics[width=\textwidth]{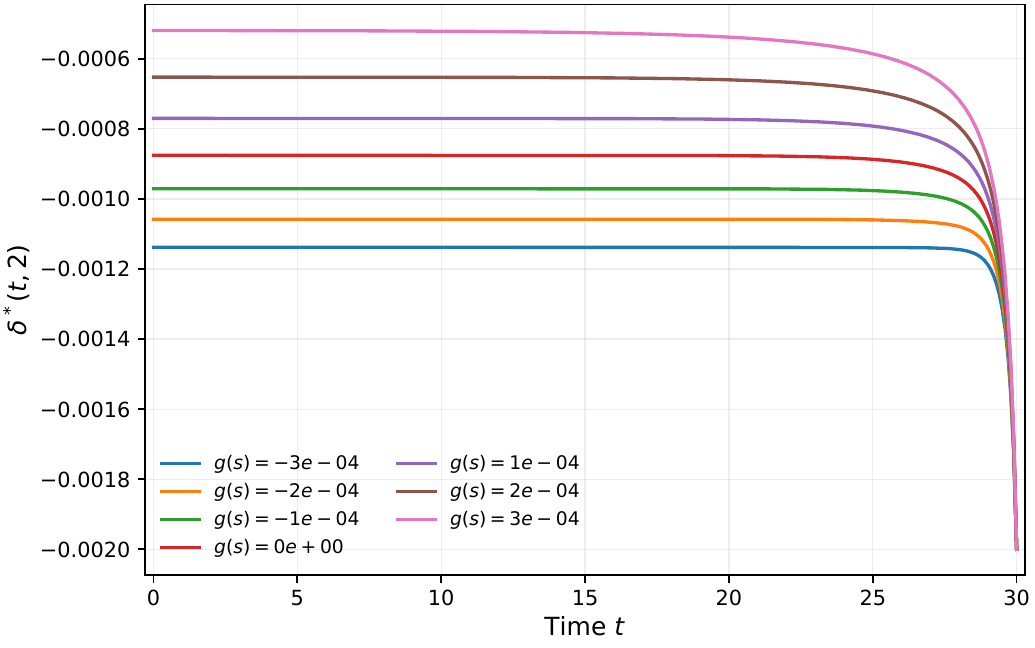}
        \caption{$\beta=0.0005$.}
        \label{fig:caseII-q2-beta0005}
    \end{subfigure}
    \hfill
    \begin{subfigure}[t]{0.49\textwidth}
        \centering
        \includegraphics[width=\textwidth]{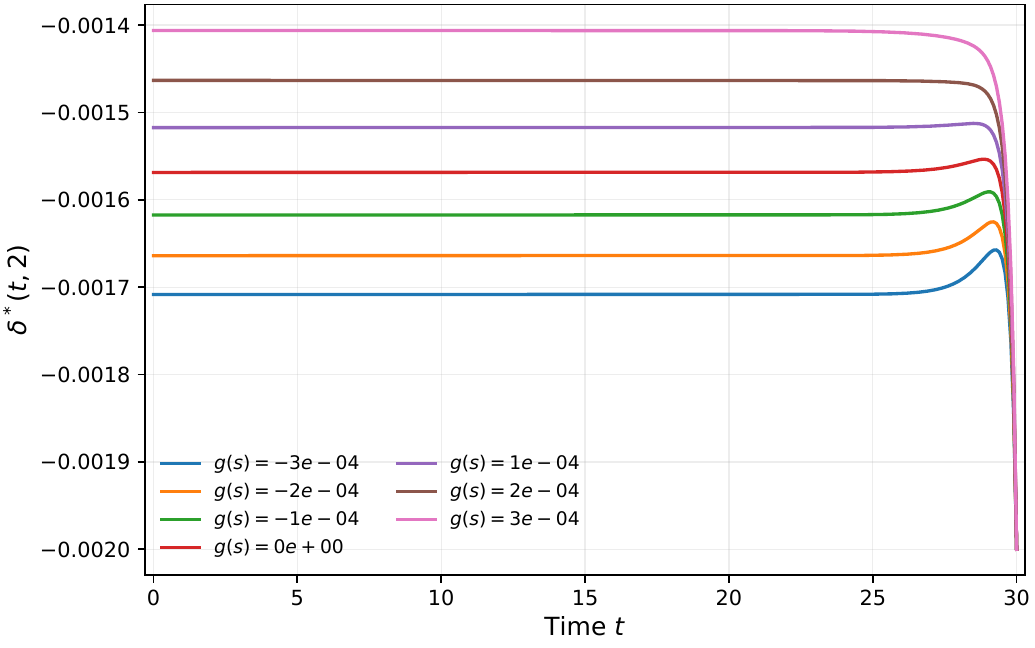}
        \caption{$\beta=0.001$.}
        \label{fig:caseII-q2-beta001}
    \end{subfigure}

    \caption{Case II: optimal quotes $\delta^\star(t,2)$ for different signal levels. The panels compare two values of the running inventory penalty parameter $\beta$ in $J(q)=\beta q^2$, with $\alpha=0.001$.}
    \label{fig:caseII-q2-running-penalty}
\end{figure}

\begin{figure}[t]
    \centering
    \begin{subfigure}[t]{0.49\textwidth}
        \centering
        \includegraphics[width=\textwidth]{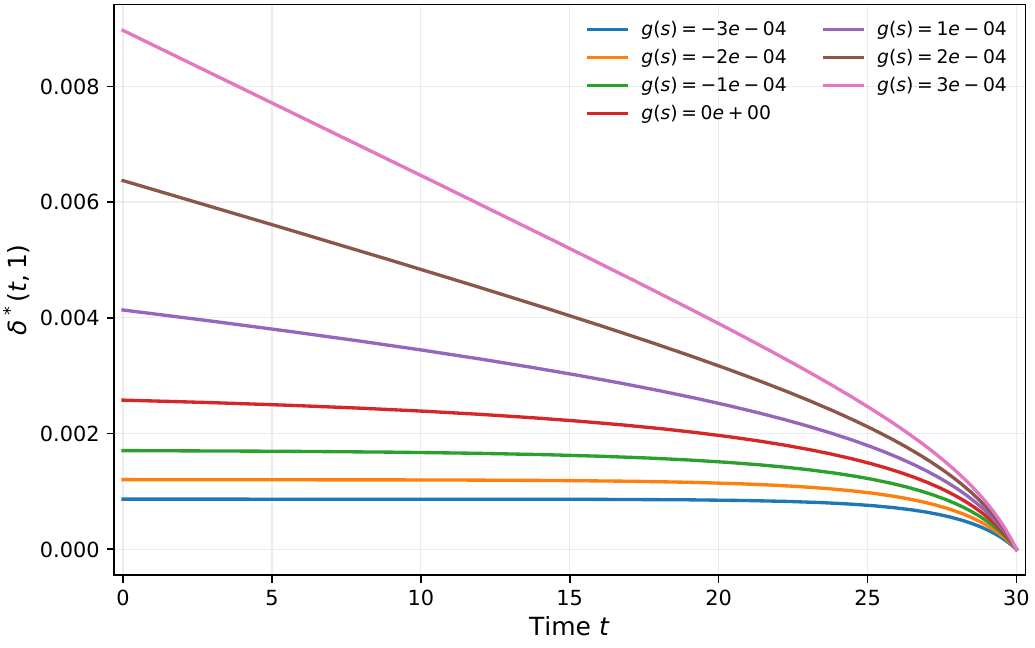}
        \caption{$\gamma=0.01$.}
        \label{fig:caseIII-q1-gamma001}
    \end{subfigure}
    \hfill
    \begin{subfigure}[t]{0.49\textwidth}
        \centering
        \includegraphics[width=\textwidth]{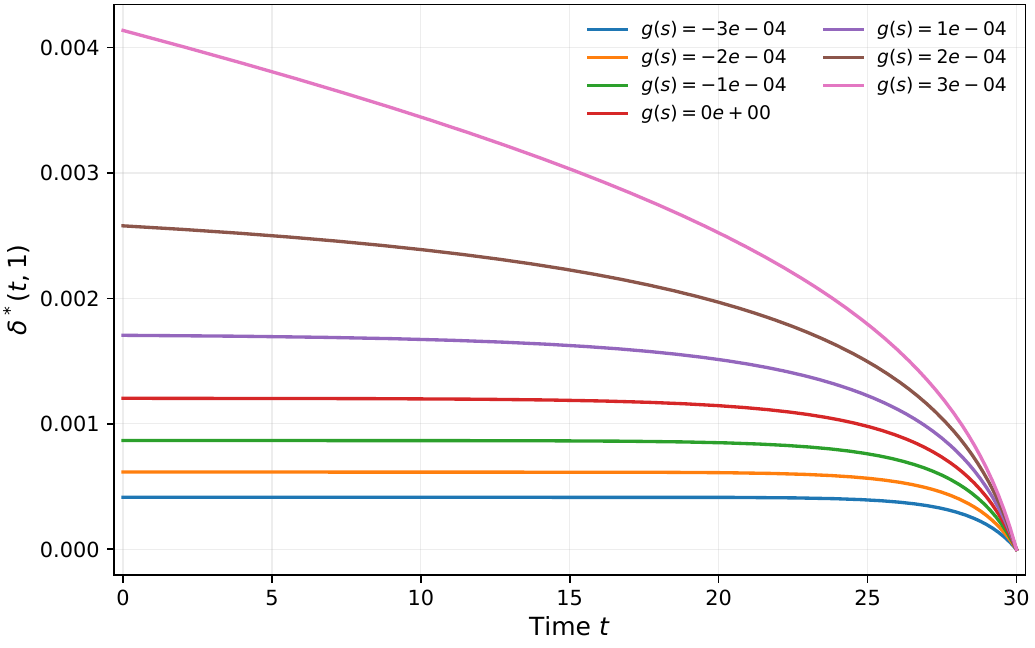}
        \caption{$\gamma=0.05$.}
        \label{fig:caseIII-q1-gamma005}
    \end{subfigure}

    \caption{Case III: optimal quotes $\delta^\star(t,1)$ for different signal levels. The panels compare two values of the risk-aversion parameter $\gamma$, with $\alpha=0.001$ and $\sigma=0.1$.}
    \label{fig:caseIII-q1-risk-aversion}
\end{figure}

\begin{figure}[t]
    \centering
    \begin{subfigure}[t]{0.49\textwidth}
        \centering
        \includegraphics[width=\textwidth]{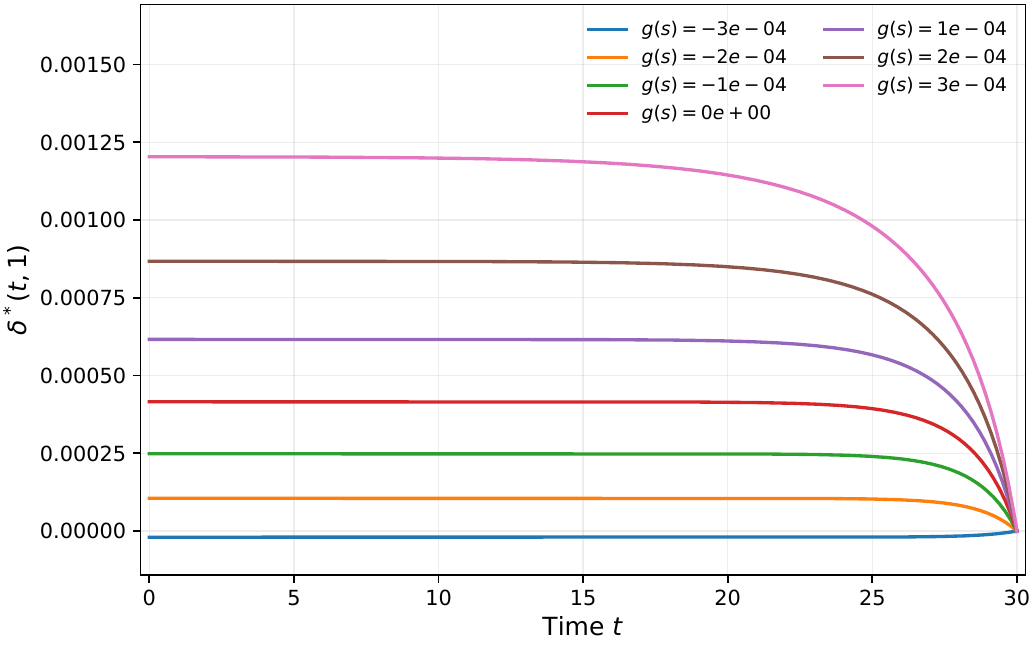}
        \caption{$\beta=0.0005$.}
        \label{fig:caseIV-q1-beta0005}
    \end{subfigure}
    \hfill
    \begin{subfigure}[t]{0.49\textwidth}
        \centering
        \includegraphics[width=\textwidth]{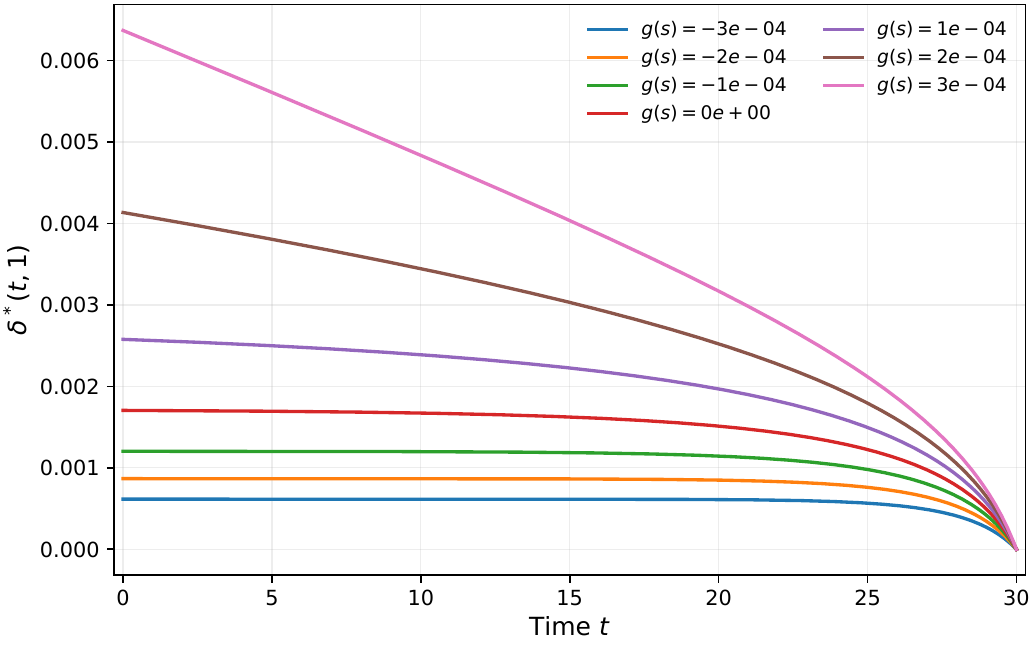}
        \caption{$\beta=0.001$.}
        \label{fig:caseIV-q1-beta0001}
    \end{subfigure}

    \caption{Case IV: optimal quotes $\delta^\star(t,1)$ for different signal levels. The panels compare two values of the running inventory penalty parameter $\beta$ in $J(q)=\beta q^2$, with $\alpha=0.001$, $\sigma=0.1$, and $\gamma=0.01$.}
    \label{fig:caseIV-q1-running-penalty}
\end{figure}

\begin{figure}[t]
    \centering
    \begin{subfigure}[t]{0.49\textwidth}
        \centering
        \includegraphics[width=\textwidth]{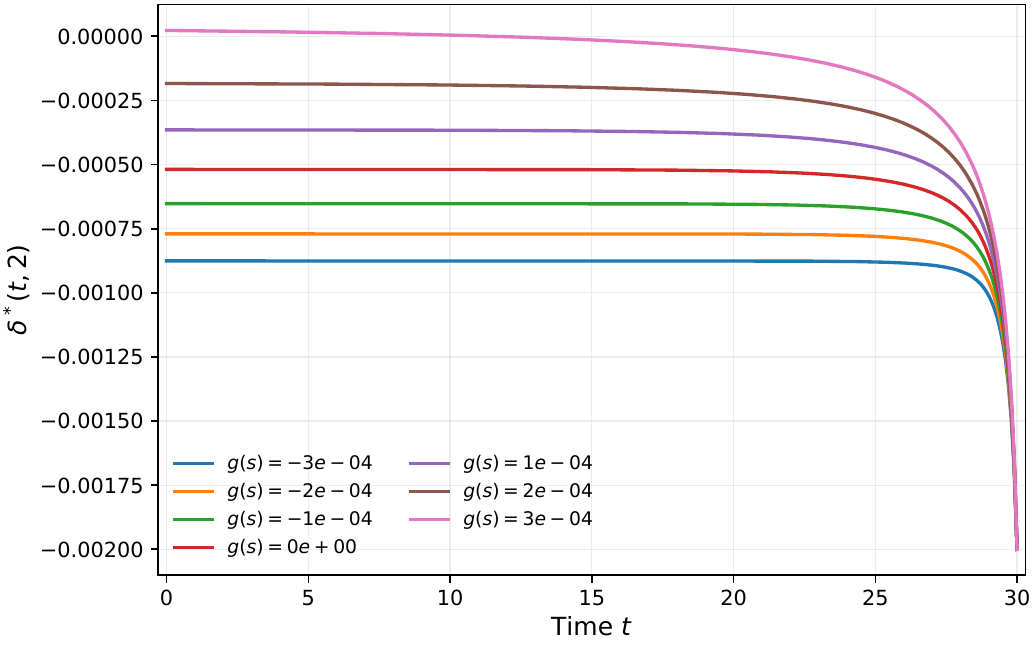}
        \caption{$\sigma=0.1$.}
        \label{fig:caseIV-q2-sigma01}
    \end{subfigure}
    \hfill
    \begin{subfigure}[t]{0.49\textwidth}
        \centering
        \includegraphics[width=\textwidth]{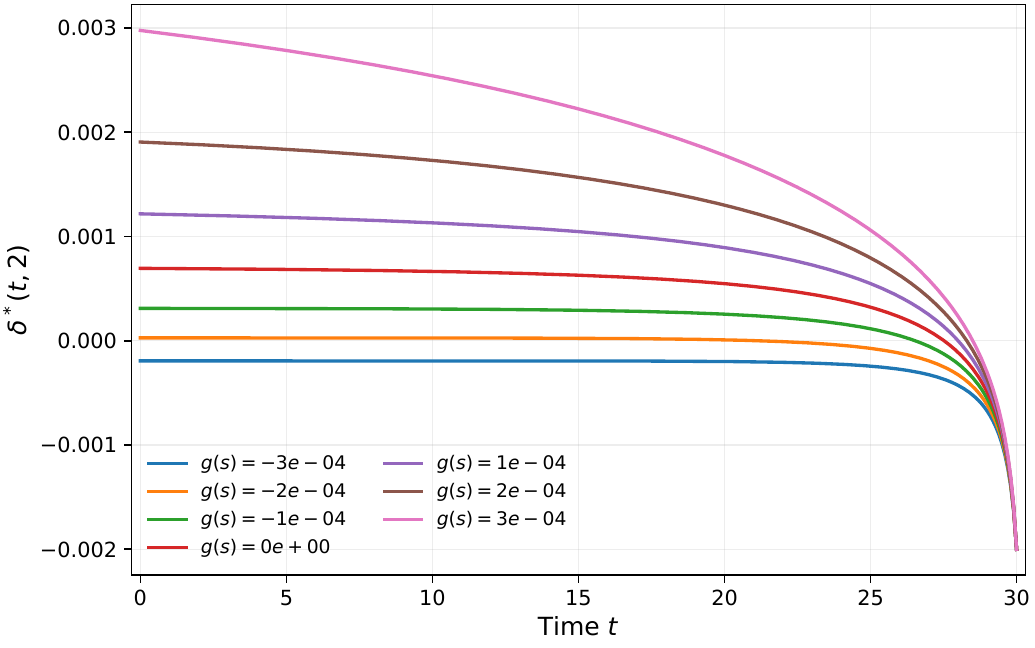}
        \caption{$\sigma=0.01$.}
        \label{fig:caseIV-q2-sigma001}
    \end{subfigure}

    \caption{Case IV: optimal quotes $\delta^\star(t,2)$ for different signal levels. The panels compare two values of the volatility parameter $\sigma$, with $\alpha=0.001$, $\beta=0.0001$, and $\gamma=0.05$.}
    \label{fig:caseIV-q2-volatility}
\end{figure}

\begin{figure}[t]
    \centering
    \begin{subfigure}[t]{0.49\textwidth}
        \centering
        \includegraphics[width=\textwidth]{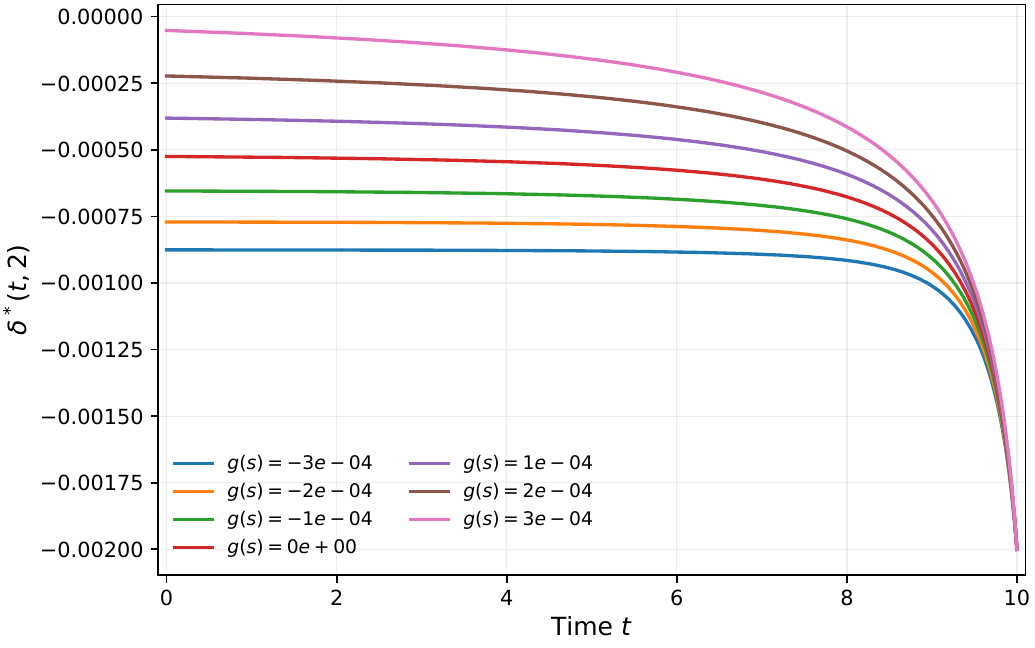}
        \caption{$T=10$.}
        \label{fig:caseIV-q2-T10}
    \end{subfigure}
    \hfill
    \begin{subfigure}[t]{0.49\textwidth}
        \centering
        \includegraphics[width=\textwidth]{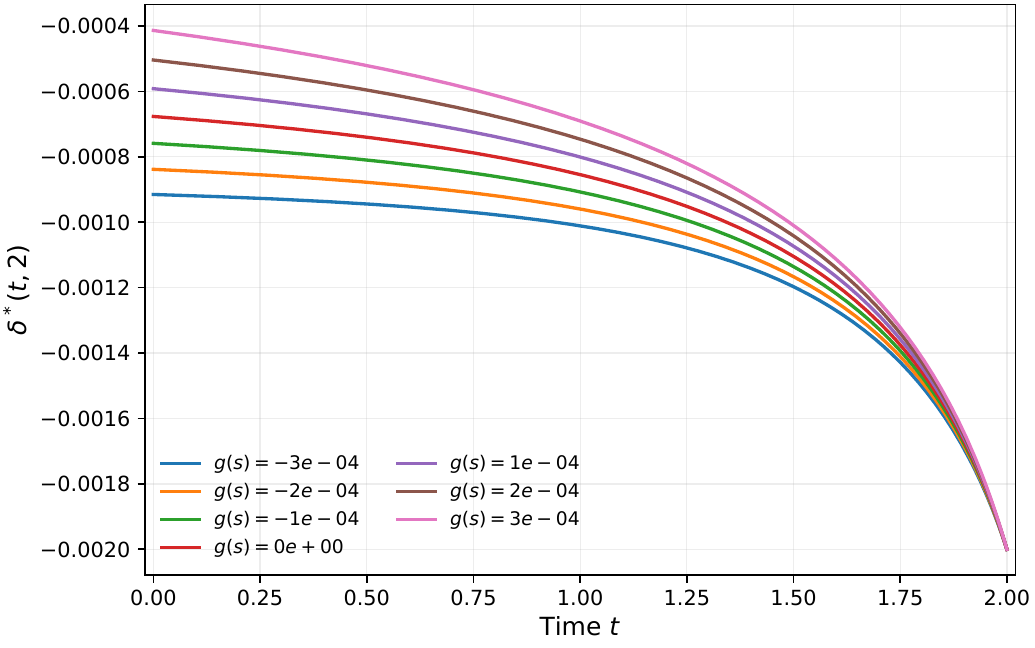}
        \caption{$T=2$.}
        \label{fig:caseIV-q2-T02}
    \end{subfigure}

    \caption{Case IV: optimal quotes $\delta^\star(t,2)$ for different signal levels. The panels compare two values of the terminal time $T$, with $\alpha=0.001$, $\beta=0.0001$, $\sigma=0.1$, and $\gamma=0.05$.}
    \label{fig:caseIV-q2-terminal-time}
\end{figure}

\begin{figure}[t]
    \centering

    \begin{subfigure}[t]{0.49\textwidth}
        \centering
        \includegraphics[width=\textwidth]{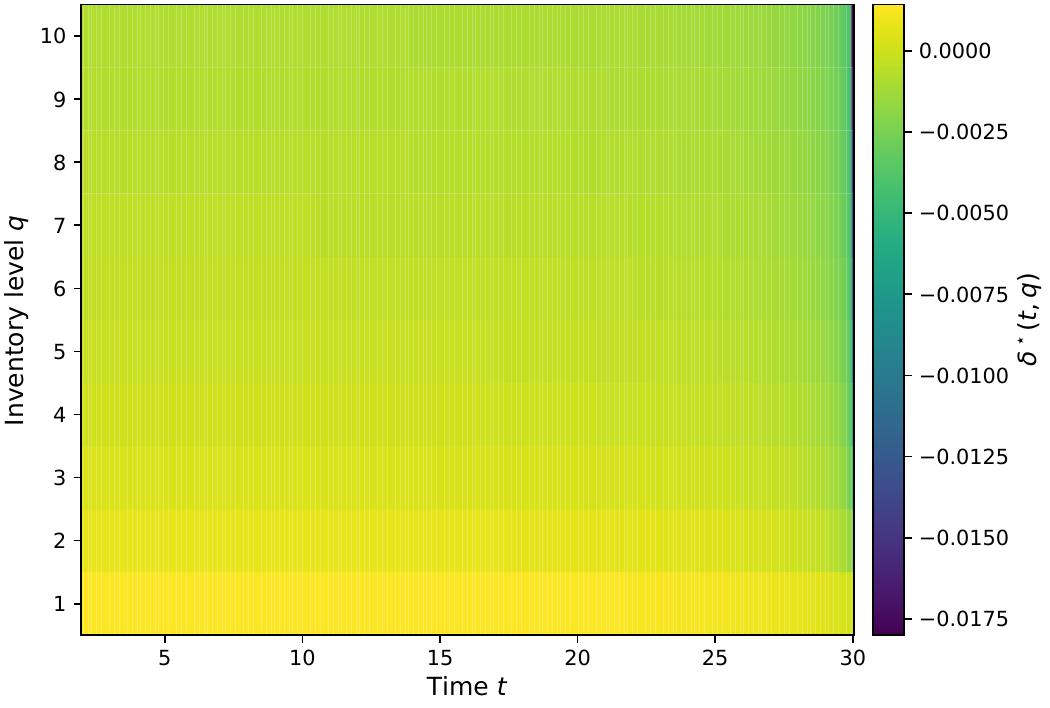}
        \caption{$g(s)=-2\times 10^{-4}$.}
        \label{fig:caseIII-heatmap-gm2e-4}
    \end{subfigure}
    \hfill
    \begin{subfigure}[t]{0.49\textwidth}
        \centering
        \includegraphics[width=\textwidth]{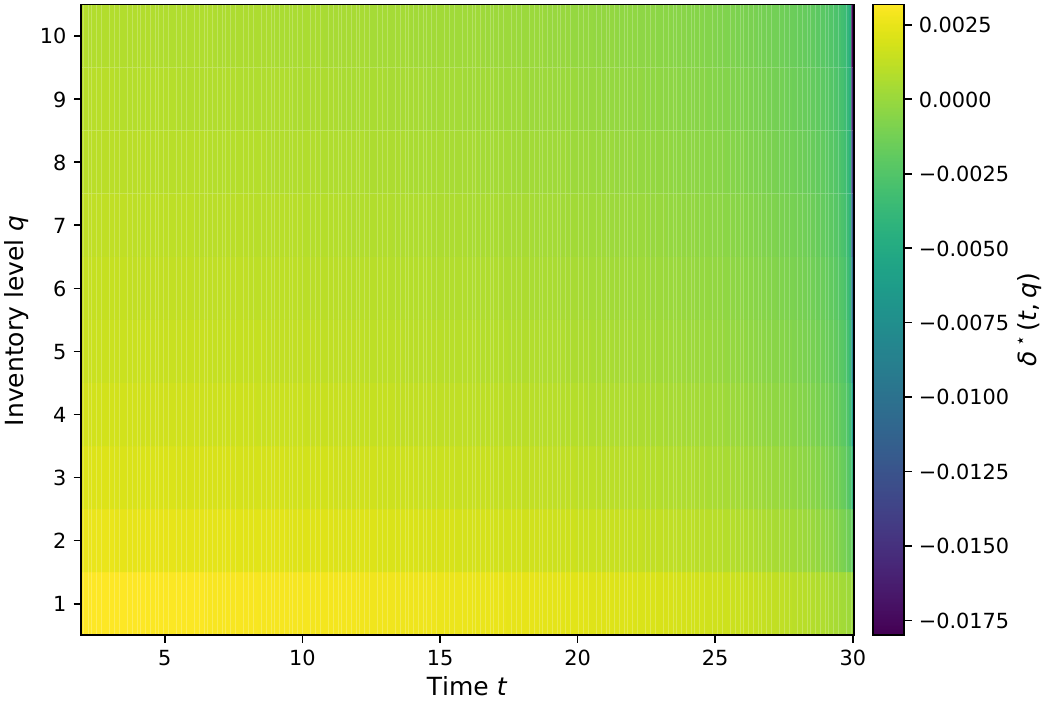}
        \caption{$g(s)=0$.}
        \label{fig:caseIII-heatmap-g0}
    \end{subfigure}

    \vspace{0.8em}

    \begin{subfigure}[t]{0.49\textwidth}
        \centering
        \includegraphics[width=\textwidth]{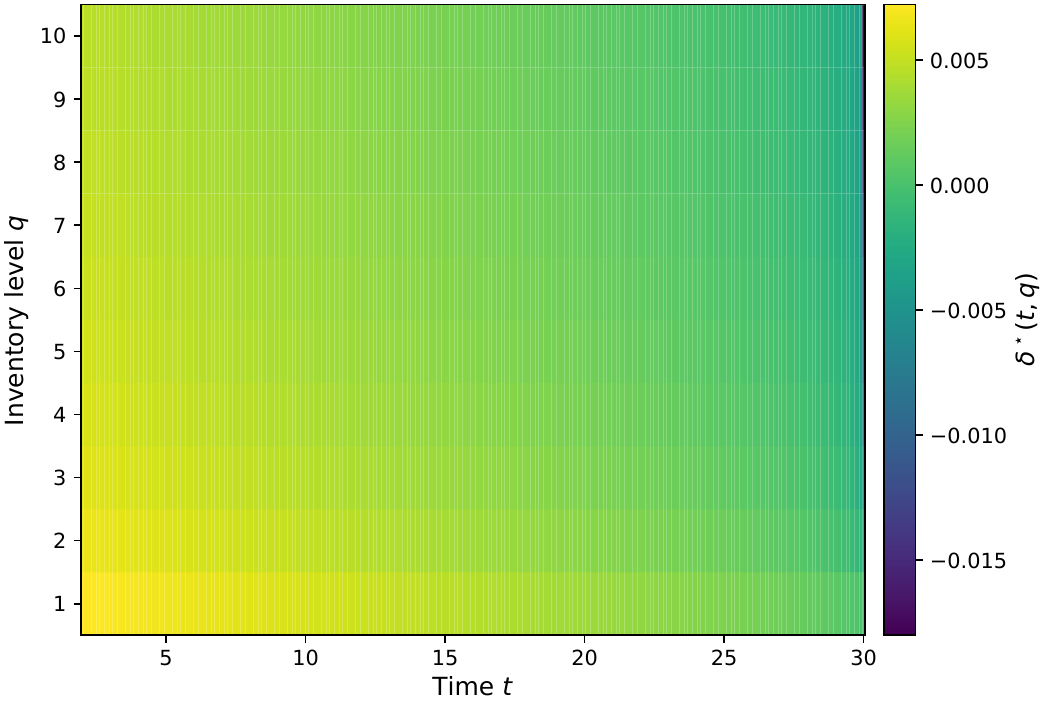}
        \caption{$g(s)=2\times 10^{-4}$.}
        \label{fig:caseIII-heatmap-g2e-4}
    \end{subfigure}
    \hfill
    \begin{subfigure}[t]{0.49\textwidth}
        \centering
        \includegraphics[width=\textwidth]{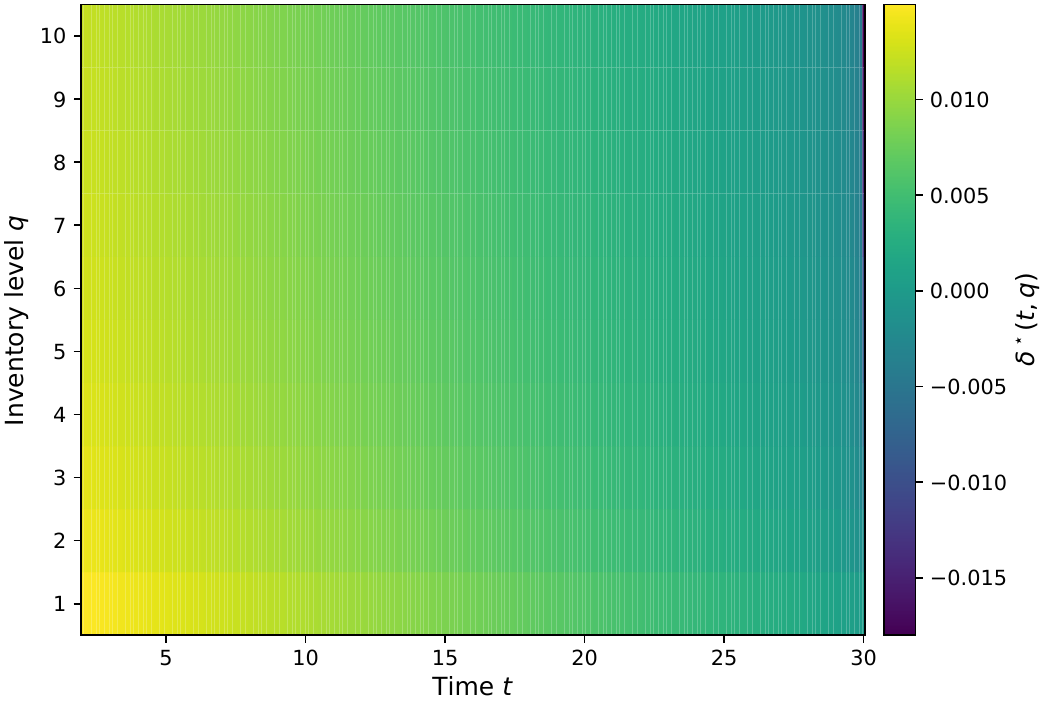}
        \caption{$g(s)=5\times 10^{-4}$.}
        \label{fig:caseIII-heatmap-g5e-4}
    \end{subfigure}

    \caption{Case III: heatmaps of the optimal quote $\delta^\star(t,q)$ over time and inventory levels for different constant signal levels. The parameters are $\alpha=0.001$, $\sigma=0.01$, and $\gamma=0.01$.}
    \label{fig:caseIII-heatmaps-signal-levels}
\end{figure}

\section{Conclusion}
\label{sec:conclusion}

This paper develops a unified explicit solution theory for optimal execution
through sequential limit-order placement. The key result is that a broad class
of controlled-jump execution problems, including quote-dependent fill
intensities, price impact, inventory risk, and signal-dependent price dynamics,
can be reduced to finite-dimensional triangular systems. This yields explicit
value functions and optimal quotes for four execution criteria based on
terminal liquidation wealth and CARA utility, with and without running inventory
penalties.

The paper also provides a rigorous foundation for these explicit strategies by
establishing well-posedness, admissibility, and verification results. The
explicit formulas reveal structural connections between the optimal quotes under
different criteria, in particular between inventory-penalized objectives and
utility-based risk aversion, and make long-horizon asymptotic analysis
tractable.

The numerical results highlight the practical relevance of the framework:
signal-dependent drift can substantially change the optimal quoting strategy.
Thus, predictive information from limit-order-book data can play an essential
role in execution decisions. Overall, the paper provides a tractable and rigorous framework that combines
explicit solvability, verification, asymptotic analysis, and financial
interpretation for signal-adaptive optimal execution quoting.

\bibliographystyle{plain} 
\bibliography{refs} 

@article{cartea2018enhancing,
  title={Enhancing trading strategies with order book signals},
  author={Cartea, \'Alvaro and Donnelly, Ryan and Jaimungal, Sebastian},
  journal={Applied Mathematical Finance},
  volume={25},
  number={1},
  pages={1--35},
  year={2018},
  publisher={Taylor \& Francis}
}

@article{capponi2020multi,
  title={Multi-asset market impact and order flow commonality},
  author={Capponi, Francesco and Cont, Rama},
  journal={Available at SSRN 3706390},
  year={2020}
}

@article{bechler2015optimal,
  title={Optimal execution with dynamic order flow imbalance},
  author={Bechler, Kyle and Ludkovski, Michael},
  journal={SIAM Journal on Financial Mathematics},
  volume={6},
  number={1},
  pages={1123--1151},
  year={2015},
  publisher={SIAM}
}

@article{almgren2001optimal,
  title={Optimal execution of portfolio transactions},
  author={Almgren, Robert and Chriss, Neil},
  journal={Journal of Risk},
  volume={3},
  pages={5--40},
  year={2001}
}

@article{bertsimas1998optimal,
  title={Optimal control of execution costs},
  author={Bertsimas, Dimitris and Lo, Andrew W},
  journal={Journal of financial markets},
  volume={1},
  number={1},
  pages={1--50},
  year={1998},
  publisher={Elsevier}
}

@article{kolm2023deep,
  title={Deep order flow imbalance: Extracting alpha at multiple horizons from the limit order book},
  author={Kolm, Petter N and Turiel, Jeremy and Westray, Nicholas},
  journal={Mathematical Finance},
  volume={33},
  number={4},
  pages={1044--1081},
  year={2023},
  publisher={Wiley Online Library}
}

@article{cont2014price,
  title={The price impact of order book events},
  author={Cont, Rama and Kukanov, Arseniy and Stoikov, Sasha},
  journal={Journal of Financial Econometrics},
  volume={12},
  number={1},
  pages={47--88},
  year={2014},
  publisher={Oxford University Press}
}

@book{cartea2015algorithmic,
  title={Algorithmic and high-frequency trading},
  author={Cartea, {\'A}lvaro and Jaimungal, Sebastian and Penalva, Jos{\'e}},
  year={2015},
  publisher={Cambridge University Press}
}

@article{lokin2026probabilitieslimitorderbook,
      title={Fill Probabilities in a Limit Order Book with State-Dependent Stochastic Order Flows}, 
      author={Felix Lokin and Fenghui Yu},
      year={2026},
      eprint={2403.02572},
      archivePrefix={arXiv},
      primaryClass={q-fin.TR},
      url={https://arxiv.org/abs/2403.02572}, 
      journal={Available at arXiv 2403.02572}
}

@article{cartea2015optimal,
  title={Optimal execution with limit and market orders},
  author={Cartea, {\'A}lvaro and Jaimungal, Sebastian},
  journal={Quantitative Finance},
  volume={15},
  number={8},
  pages={1279--1291},
  year={2015},
  publisher={Taylor \& Francis}
}

@article{hawkes2018hawkes,
  title={Hawkes processes and their applications to finance: a review},
  author={Hawkes, Alan G},
  journal={Quantitative Finance},
  volume={18},
  number={2},
  pages={193--198},
  year={2018},
  publisher={Taylor \& Francis}
}

@article{lehalle2021optimal,
  title={Optimal liquidity-based trading tactics},
  author={Lehalle, Charles-Albert and Mounjid, Othmane and Rosenbaum, Mathieu},
  journal={Stochastic Systems},
  volume={11},
  number={4},
  pages={368--390},
  year={2021},
  publisher={INFORMS}
}

@article{cartea2016incorporating,
  title={Incorporating order-flow into optimal execution},
  author={Cartea, {\'A}lvaro and Jaimungal, Sebastian},
  journal={Mathematics and Financial Economics},
  volume={10},
  pages={339--364},
  year={2016},
  publisher={Springer}
}

@article{guilbaud2013optimal,
  title={Optimal high-frequency trading with limit and market orders},
  author={Guilbaud, Fabien and Pham, Huyen},
  journal={Quantitative Finance},
  volume={13},
  number={1},
  pages={79--94},
  year={2013},
  publisher={Taylor \& Francis}
}

@article{cont2017optimal,
  title={Optimal order placement in limit order markets},
  author={Cont, Rama and Kukanov, Arseniy},
  journal={Quantitative Finance},
  volume={17},
  number={1},
  pages={21--39},
  year={2017},
  publisher={Taylor \& Francis}
}

@article{Guant2012,
   author = {Olivier Guéant and Charles Albert Lehalle and Joaquin Fernandez-Tapia},
   doi = {10.1137/110850475},
   issn = {1945497X},
   issue = {1},
   journal = {SIAM Journal on Financial Mathematics},
   month = {11},
   pages = {740-764},
   publisher = {Society for Industrial and Applied Mathematics},
   title = {Optimal Portfolio Liquidation with Limit Orders},
   volume = {3},
   url = {https://epubs.siam.org/doi/10.1137/110850475},
   year = {2012}
}

@article{Bayraktar2014,
   author = {Erhan Bayraktar and Michael Ludkovski},
   doi = {10.1111/J.1467-9965.2012.00529.X},
   issn = {1467-9965},
   issue = {4},
   journal = {Mathematical Finance},
   month = {10},
   pages = {627-650},
   publisher = {John Wiley \& Sons, Ltd},
   title = {LIQUIDATION IN LIMIT ORDER BOOKS WITH CONTROLLED INTENSITY},
   volume = {24},
   url = {https://onlinelibrary.wiley.com/doi/full/10.1111/j.1467-9965.2012.00529.x https://onlinelibrary.wiley.com/doi/abs/10.1111/j.1467-9965.2012.00529.x https://onlinelibrary.wiley.com/doi/10.1111/j.1467-9965.2012.00529.x},
   year = {2014}
}

@article{Guant2015,
   author = {Olivier Guéant and Charles Albert Lehalle},
   doi = {10.1111/MAFI.12052},
   issn = {1467-9965},
   issue = {3},
   journal = {Mathematical Finance},
   month = {7},
   pages = {457-495},
   publisher = {John Wiley \& Sons, Ltd},
   title = {GENERAL INTENSITY SHAPES IN OPTIMAL LIQUIDATION},
   volume = {25},
   url = {https://onlinelibrary.wiley.com/doi/full/10.1111/mafi.12052 https://onlinelibrary.wiley.com/doi/abs/10.1111/mafi.12052 https://onlinelibrary.wiley.com/doi/10.1111/mafi.12052},
   year = {2015}
}

@article{Guant2013,
   author = {Olivier Guéant and Charles Albert Lehalle and Joaquin Fernandez-Tapia},
   doi = {10.1007/S11579-012-0087-0/METRICS},
   issn = {18629679},
   issue = {4},
   journal = {Mathematics and Financial Economics},
   month = {9},
   pages = {477-507},
   publisher = {Springer},
   title = {Dealing with the inventory risk: A solution to the market making problem},
   volume = {7},
   url = {https://link-springer-com.tudelft.idm.oclc.org/article/10.1007/s11579-012-0087-0},
   year = {2013}
}

@article{Lehalle2017,
   author = {Lehalle, Charles-Albert and Mounjid, Othmane},
title = {Limit Order Strategic Placement with Adverse Selection Risk and the Role of Latency},
journal = {Market Microstructure and Liquidity},
volume = {03},
number = {01},
pages = {1750009},
year = {2017},
doi = {10.1142/S2382626617500095},

URL = {https://doi.org/10.1142/S2382626617500095},
eprint = {https://doi.org/10.1142/S2382626617500095}
}

@article{barzykin2025optimalquotingadverseselection,
      title={Optimal Quoting under Adverse Selection and Price Reading}, 
      author={Alexander Barzykin and Philippe Bergault and Olivier Guéant and Malo Lemmel},
      year={2025},
      eprint={2508.20225},
      archivePrefix={arXiv},
      primaryClass={q-fin.TR},
      url={https://arxiv.org/abs/2508.20225}, 
      journal={Available at arXiv 2508.20225}
}

@article{Barzykin2025,
   author = {Alexander Barzykin and Robert Boyce and Eyal Neuman},
   doi = {10.2139/SSRN.5859484},
   title = {FX Market Making with Internal Liquidity},
   url = {https://papers.ssrn.com/abstract=5859484},
   journal = {Available at SSRN 5859484},
   year = {2025}
}

@article{Bergault2021,
   author = {Philippe Bergault and David Evangelista and Olivier Guéant and Douglas Vieira},
   doi = {10.1080/1350486X.2021.1949359},
   issn = {14664313},
   issue = {2},
   journal = {Applied Mathematical Finance},
   month = {3},
   pages = {101-142},
   publisher = {Routledge},
   title = {Closed-form Approximations in Multi-asset Market Making},
   volume = {28},
   url = {https://scholar.google.com/scholar_url?url=https://www.tandfonline.com/doi/pdf/10.1080/1350486X.2021.1949359%3Fcasa_token%3DDCPp1JwW6fEAAAAA:NQp_ghvEYwLixcwwYVz-MT9F-ldvUuUyy-SNuDizkHTS8jqTYaXxOWBlz3uiNrFkxEY1vUy_gNTSBA&hl=nl&sa=T&oi=ucasa&ct=ucasa&ei=1LEIaqP5F_i9ieoP3-TG8As&scisig=AFyMTJXwTN8DGyFCeOWV93koQxE7},
   year = {2021}
}

@article{Cartea2014,
author = {Cartea, \'{A}lvaro and Jaimungal, Sebastian and Ricci, Jason},
title = {Buy Low, Sell High: A High Frequency Trading Perspective},
journal = {SIAM Journal on Financial Mathematics},
volume = {5},
number = {1},
pages = {415-444},
year = {2014},
doi = {10.1137/130911196},
URL = {https://doi.org/10.1137/130911196
},
eprint = { https://doi.org/10.1137/130911196}
}

@article{HO198147,
title = {Optimal dealer pricing under transactions and return uncertainty},
journal = {Journal of Financial Economics},
volume = {9},
number = {1},
pages = {47-73},
year = {1981},
issn = {0304-405X},
doi = {https://doi.org/10.1016/0304-405X(81)90020-9},
url = {https://www.sciencedirect.com/science/article/pii/0304405X81900209},
author = {Thomas Ho and Hans R. Stoll}
}

@article{Avellaneda01042008,
author = {Marco Avellaneda and Sasha Stoikov},
title = {High-frequency trading in a limit order book},
journal = {Quantitative Finance},
volume = {8},
number = {3},
pages = {217--224},
year = {2008},
publisher = {Routledge},
doi = {10.1080/14697680701381228},
URL = {https://doi.org/10.1080/14697680701381228},
eprint = {https://doi.org/10.1080/14697680701381228   }
}

@article{Boyce2025,
author = {Boyce, Robert and Herdegen, Martin and S\'{a}nchez-Betancourt, Leandro},
title = {Market Making with Exogenous Competition},
journal = {SIAM Journal on Financial Mathematics},
volume = {16},
number = {2},
pages = {692-706},
year = {2025},
doi = {10.1137/24M1679811},

URL = {https://doi.org/10.1137/24M1679811},
eprint = {https://doi.org/10.1137/24M1679811}
}
\nocite{*}

\appendix

\section{Proofs}

\subsection{Proof of Lemma \ref{lem:optimal_quote_case_II}}
\label{app:optimal-quote-caseII-proof}
The terminal and boundary conditions in \eqref{eq:HJB_equations_inventory_risk}
suggest the affine ansatz
\[
    R(t,x,M,q)=x+qM+r(t,q).
\]
The term \(x+qM\) is the mark-to-market wealth, while \(r(t,q)\) captures the
additional value of optimal liquidation net of the running inventory cost.

Using
\[
    \partial_t R=\partial_t r,
    \qquad
    \partial_M R=q,
    \qquad
    \partial_{MM}R=0,
\]
and substituting the ansatz into \eqref{eq:HJB_equations_inventory_risk}, we
obtain, for \(q\geq1\),
\begin{align}
\label{eq:HJB_equations_runningrisk_ansatz}
\left\{
\begin{array}{ll}
\displaystyle
\partial_t r(t,q)
+
g(s)q
-
J(q)
+
\sup_{\delta\in\mathcal A}
\left\{
    \lambda e^{-\kappa\delta}
    \left[
        -a+b\delta+r(t,q-1)-r(t,q)
    \right]
\right\}
=0,
& t<T,
\\[1.1em]
\displaystyle
r(T,q)=-qI(q),
& q\geq 1,
\\[0.6em]
\displaystyle
r(t,0)=0.
&
\end{array}
\right.
\end{align}
The terminal condition follows by comparing
\[
    R(T,x,M,q)=x+q(M-I(q))
\]
with
\[
    R(T,x,M,q)=x+qM+r(T,q).
\]

For fixed \((t,q)\), define
\[
    \Delta r(t,q):=r(t,q)-r(t,q-1).
\]
The term to be maximized is
\[
    \lambda e^{-\kappa\delta}
    \left[
        -a+b\delta-\Delta r(t,q)
    \right].
\]
The first-order condition for an interior maximizer is the same as in Case I
with \(h\) replaced by \(r\), and gives
\begin{equation}
\label{eq:optimal_delta_feedback_running_risk_unconstrained}
    \delta^{\mathrm{unc}}(t,q)
    =
    \frac{1}{\kappa}
    +
    \frac{a}{b}
    +
    \frac{r(t,q)-r(t,q-1)}{b}.
\end{equation}

Substituting \eqref{eq:optimal_delta_feedback_running_risk_unconstrained} into
\eqref{eq:HJB_equations_runningrisk_ansatz} yields
\begin{align}
\label{eq:HJB_equations_runningrisk_feedback}
\left\{
\begin{array}{ll}
\displaystyle
\partial_t r(t,q)
+
g(s)q
-
J(q)
+
\frac{\lambda b}{e\kappa}
\exp\left\{
    -\frac{\kappa a}{b}
    -
    \frac{\kappa}{b}
    \bigl(r(t,q)-r(t,q-1)\bigr)
\right\}
=0,
& q\geq 1,\; t<T,
\\[1.1em]
\displaystyle
r(T,q)=-qI(q),
& q\geq 1,
\\[0.6em]
\displaystyle
r(t,0)=0.
&
\end{array}
\right.
\end{align}

Now set
\[
    r(t,q)=\frac{b}{\kappa}\log w(t,q),
    \qquad w(t,q)>0.
\]
Then
\[
    r(t,q)-r(t,q-1)
    =
    \frac{b}{\kappa}
    \log\frac{w(t,q)}{w(t,q-1)}.
\]
Therefore the feedback quote becomes
\[
    \delta^{\mathrm{unc}}(t,q)
    =
    \frac{1}{\kappa}
    \left(
        1+\log\frac{w(t,q)}{w(t,q-1)}
    \right)
    +
    \frac{a}{b}.
\]
Substituting the logarithmic transformation into
\eqref{eq:HJB_equations_runningrisk_feedback} gives
\[
\begin{aligned}
0
&=
\frac{b}{\kappa}
\frac{\partial_t w(t,q)}{w(t,q)}
+
g(s)q
-
J(q)
+
\frac{\lambda b}{e\kappa}
\exp\left\{-\frac{\kappa a}{b}\right\}
\frac{w(t,q-1)}{w(t,q)}.
\end{aligned}
\]
Multiplying by \(\kappa w(t,q)/b\) gives
\eqref{eq:PDE_running_risk}. The terminal and boundary conditions follow from
\[
    r(T,q)=-qI(q),
    \qquad
    r(t,0)=0,
\]
which imply
\[
    w(T,q)
    =
    \exp\left\{
        -\frac{\kappa}{b}qI(q)
    \right\},
    \qquad
    w(t,0)=1.
\]
If the admissible set is bounded, the admissible maximizer is the projection of
the unconstrained maximizer onto \([\delta_{\min},\delta_{\max}]\).

\subsection{Proof of Lemma \ref{lem:optimal_quote_case_III}}
\label{app:optimal-quote-caseIII-proof}
The terminal and boundary conditions in \eqref{eq:HJB_equations_utility}
suggest the CARA ansatz
\[
    U(t,x,M,q)
    =
    -\exp\left\{
        -\gamma\bigl(x+qM+u(t,q)\bigr)
    \right\}.
\]
The quantity \(x+qM\) is the mark-to-market wealth, while \(u(t,q)\) is the
certainty-equivalent adjustment associated with the remaining liquidation
problem.

For notational simplicity, write \(U=U(t,x,M,q)\). Then
\[
    \partial_t U=-\gamma U\,\partial_t u,
    \qquad
    \partial_M U=-\gamma qU,
    \qquad
    \partial_{MM}U=\gamma^2q^2U.
\]
Moreover, after an execution at quote depth \(\delta\),
\[
\begin{aligned}
&U\bigl(t,x+M-a+b\delta,M,q-1\bigr)  \\
&\qquad
=
U(t,x,M,q)
\exp\left\{
    -\gamma
    \bigl(
        -a+b\delta+u(t,q-1)-u(t,q)
    \bigr)
\right\}.
\end{aligned}
\]
Substituting these expressions into \eqref{eq:HJB_equations_utility} and
dividing by \(-\gamma U>0\), we obtain, for \(q\geq1\),
\begin{align}
\label{eq:HJB_equations_utility_ansatz}
\left\{
\begin{array}{ll}
\displaystyle
\partial_t u(t,q)
+
g(s)q
-
\frac{1}{2}\sigma^2\gamma q^2
& \\[0.5em]
\displaystyle\qquad
+
\sup_{\delta\in\mathcal A}
\left\{
    \frac{\lambda}{\gamma}e^{-\kappa\delta}
    \left[
        1
        -
        \exp\left\{
            -\gamma
            \bigl(
                -a+b\delta+u(t,q-1)-u(t,q)
            \bigr)
        \right\}
    \right]
\right\}
=0,
& t<T,
\\[1.2em]
\displaystyle
u(T,q)=-qI(q),
& q\geq 1,
\\[0.6em]
\displaystyle
u(t,0)=0.
&
\end{array}
\right.
\end{align}
The term
\(
    -\frac{1}{2}\sigma^2\gamma q^2
\)
is the certainty-equivalent penalty for mark-to-market inventory risk generated
by the volatility of \(qM_t\).

For fixed \((t,q)\), define
\[
    \Delta u(t,q):=u(t,q)-u(t,q-1).
\]
The expression to be maximized is
\[
    \frac{\lambda}{\gamma}e^{-\kappa\delta}
    \left[
        1
        -
        \exp\left\{
            -\gamma
            \bigl(
                -a+b\delta-\Delta u(t,q)
            \bigr)
        \right\}
    \right].
\]
The first-order condition for an interior maximizer is
\[
\begin{aligned}
0
&=
\frac{\partial}{\partial\delta}
\left[
    e^{-\kappa\delta}
    \left(
        1
        -
        \exp\left\{
            -\gamma
            \bigl(
                -a+b\delta-\Delta u(t,q)
            \bigr)
        \right\}
    \right)
\right]
\\
&=
e^{-\kappa\delta}
\left[
    -\kappa
    +
    (\kappa+b\gamma)
    \exp\left\{
        -\gamma
        \bigl(
            -a+b\delta-\Delta u(t,q)
        \bigr)
    \right\}
\right].
\end{aligned}
\]
Thus the unconstrained maximizer is
\begin{equation}
\label{eq:optimal_delta_feedback_utility_unconstrained}
    \delta^{\mathrm{unc}}(t,q)
    =
    \frac{1}{b\gamma}
    \log\left(\frac{\kappa+b\gamma}{\kappa}\right)
    +
    \frac{a}{b}
    +
    \frac{u(t,q)-u(t,q-1)}{b}.
\end{equation}

Substituting \eqref{eq:optimal_delta_feedback_utility_unconstrained} into
\eqref{eq:HJB_equations_utility_ansatz} gives
\begin{align}
\label{eq:HJB_equations_utility_optimal}
\left\{
\begin{array}{ll}
\displaystyle
\partial_t u(t,q)
+
g(s)q
-
\frac{1}{2}\sigma^2\gamma q^2
& \\[0.4em]
\displaystyle\quad
+
\frac{b\lambda}{\kappa}
\left(
    \frac{\kappa}{\kappa+b\gamma}
\right)^{\frac{\kappa}{b\gamma}+1}
e^{-\kappa a/b}
\exp\left\{
    -\frac{\kappa}{b}
    \bigl(u(t,q)-u(t,q-1)\bigr)
\right\}
=0,
& q\geq 1,\; t<T,
\\[1.2em]
\displaystyle
u(T,q)=-qI(q),
& q\geq 1,
\\[0.6em]
\displaystyle
u(t,0)=0.
&
\end{array}
\right.
\end{align}

Finally, set
\[
    u(t,q)=\frac{b}{\kappa}\log w(t,q),
    \qquad w(t,q)>0.
\]
Then
\[
    u(t,q)-u(t,q-1)
    =
    \frac{b}{\kappa}
    \log\frac{w(t,q)}{w(t,q-1)}.
\]
Therefore the feedback quote becomes
\[
    \delta^{\mathrm{unc}}(t,q)
    =
    \frac{1}{b\gamma}
    \log\left(\frac{\kappa+b\gamma}{\kappa}\right)
    +
    \frac{a}{b}
    +
    \frac{1}{\kappa}
    \log\frac{w(t,q)}{w(t,q-1)}.
\]
Substituting the logarithmic transformation into
\eqref{eq:HJB_equations_utility_optimal}, and using the definition of
\(\widehat{\lambda}\), gives \eqref{eq:PDE_utility}. The terminal and boundary
conditions follow from
\[
    u(T,q)=-qI(q),
    \qquad
    u(t,0)=0,
\]
which imply
\[
    w(T,q)=
    \exp\left\{
        -\frac{\kappa}{b}qI(q)
    \right\},
    \qquad
    w(t,0)=1.
\]
If the admissible set is bounded, the admissible maximizer is the projection of
the unconstrained maximizer onto \([\delta_{\min},\delta_{\max}]\).

\subsection{Proof of Lemma \ref{lem:optimal_quote_case_IV}}
\label{app:optimal-quote-caseIV-proof}
The terminal and boundary conditions in
\eqref{eq:HJB_equations_utility_inventory} suggest the CARA ansatz
\[
    F(t,x,M,q)
    =
    -\exp\left\{
        -\gamma\bigl(x+qM+f(t,q)\bigr)
    \right\}.
\]
Here \(x+qM\) is the mark-to-market wealth, while \(f(t,q)\) is the
certainty-equivalent adjustment generated by optimal liquidation, price risk,
and the running inventory cost.

For notational simplicity, write \(F=F(t,x,M,q)\). Then
\[
    \partial_t F=-\gamma F\,\partial_t f,
    \qquad
    \partial_M F=-\gamma qF,
    \qquad
    \partial_{MM}F=\gamma^2q^2F.
\]
Moreover,
\[
\begin{aligned}
&F\bigl(t,x+M-a+b\delta,M,q-1\bigr)  \\
&\qquad
=
F(t,x,M,q)
\exp\left\{
    -\gamma
    \bigl(
        -a+b\delta+f(t,q-1)-f(t,q)
    \bigr)
\right\}.
\end{aligned}
\]
Substituting these expressions into
\eqref{eq:HJB_equations_utility_inventory} and dividing by
\(-\gamma F>0\), we obtain
\begin{align}
\label{eq:HJB_equations_utility_inventory_ansatz}
\left\{
\begin{array}{ll}
\displaystyle
\partial_t f(t,q)
+
g(s)q
-
\frac{1}{2}\sigma^2\gamma q^2
-
J(q)
& \\[0.5em]
\displaystyle\qquad
+
\sup_{\delta\in\mathcal A}
\left\{
    \frac{\lambda}{\gamma}e^{-\kappa\delta}
    \left[
        1
        -
        \exp\left\{
            -\gamma
            \bigl(
                -a+b\delta+f(t,q-1)-f(t,q)
            \bigr)
        \right\}
    \right]
\right\}
=0,
& q\geq 1,\; t<T,
\\[1.2em]
\displaystyle
f(T,q)=-qI(q),
& q\geq 1,
\\[0.6em]
\displaystyle
f(t,0)=0.
&
\end{array}
\right.
\end{align}
Thus, compared with Case III, the only additional term in the reduced equation
is \(-J(q)\).

The maximization over \(\delta\) is the same as in Case III. Defining
\[
    \Delta f(t,q):=f(t,q)-f(t,q-1),
\]
the first-order condition for an interior maximizer gives
\[
    \delta^{\mathrm{unc}}(t,q)
    =
    \frac{1}{b\gamma}
    \log\left(\frac{\kappa+b\gamma}{\kappa}\right)
    +
    \frac{a}{b}
    +
    \frac{f(t,q)-f(t,q-1)}{b}.
\]
Substituting this optimizer into
\eqref{eq:HJB_equations_utility_inventory_ansatz} yields
\begin{align}
\label{eq:HJB_equations_utility_inventory_optimal}
\left\{
\begin{array}{ll}
\displaystyle
\partial_t f(t,q)
+
g(s)q
-
\frac{1}{2}\sigma^2\gamma q^2
-
J(q)
+
\frac{b}{\kappa}\widehat{\lambda}
\exp\left\{
    -\frac{\kappa}{b}
    \bigl(f(t,q)-f(t,q-1)\bigr)
\right\}
=0,
& q\geq 1,\; t<T,
\\[1.2em]
\displaystyle
f(T,q)=-qI(q),
& q\geq 1,
\\[0.6em]
\displaystyle
f(t,0)=0.
&
\end{array}
\right.
\end{align}

Finally, set
\[
    f(t,q)=\frac{b}{\kappa}\log w(t,q),
    \qquad w(t,q)>0.
\]
Then
\[
    f(t,q)-f(t,q-1)
    =
    \frac{b}{\kappa}
    \log\frac{w(t,q)}{w(t,q-1)}.
\]
This gives the feedback quote in
\eqref{eq:optimal_delta_feedback_utility_inventory_w}. Substituting the same
transformation into
\eqref{eq:HJB_equations_utility_inventory_optimal} gives
\eqref{eq:PDE_utility_inventory}. The terminal and boundary conditions follow
from
\[
    f(T,q)=-qI(q),
    \qquad
    f(t,0)=0,
\]
which imply
\[
    w(T,q)=
    \exp\left\{
        -\frac{\kappa}{b}qI(q)
    \right\},
    \qquad
    w(t,0)=1.
\]
If the admissible set is bounded, the admissible maximizer is obtained by
projecting the unconstrained maximizer onto
\([\delta_{\min},\delta_{\max}]\).

\subsection{Proof of Proposition \ref{prop:explicit_solution_triangular_system}}
\label{app:triangular-system-proof}

In this appendix, we provide the detailed induction argument leading to the
explicit solution formula for the triangular system
\begin{align}
\label{eq:appendix_general_triangular_system}
\left\{
\begin{array}{ll}
\displaystyle
\partial_t w(t,q)+A_q w(t,q)+Cw(t,q-1)=0,
& q\geq 1,\; t<T,
\\[0.8em]
\displaystyle
w(T,q)=G_q,
& q\geq 1,
\\[0.6em]
\displaystyle
w(t,0)=G_0=1.
&
\end{array}
\right.
\end{align}
We assume throughout this appendix that
\(
    A_i\neq A_j,\) and 
    \( i\neq j,
\)
so that the denominators appearing below are nonzero. Set
\[
    \tau:=T-t,
    \qquad
    v_q(\tau):=w(T-\tau,q),
    \qquad 0\leq \tau\leq T.
\]
Then \eqref{eq:appendix_general_triangular_system} is equivalent to the
forward triangular system
\begin{equation}
\label{eq:appendix_forward_triangular_system}
    \frac{\dd}{\dd\tau}v_q(\tau)
    =
    A_qv_q(\tau)+Cv_{q-1}(\tau),
    \qquad
    v_q(0)=G_q,
    \qquad
    v_0(\tau)=1.
\end{equation}

For fixed \(q\geq 1\), this is a linear inhomogeneous ODE with forcing term
\(Cv_{q-1}(\tau)\). Multiplying by the integrating factor \(e^{-A_q\tau}\)
gives
\[
    \frac{\dd}{\dd\tau}
    \left(
        e^{-A_q\tau}v_q(\tau)
    \right)
    =
    Ce^{-A_q\tau}v_{q-1}(\tau).
\]
Integrating from \(0\) to \(\tau\), we obtain the variation-of-constants
formula
\begin{equation}
\label{eq:appendix_vq_variation_constants}
    v_q(\tau)
    =
    e^{A_q\tau}G_q
    +
    C\int_0^\tau e^{A_q(\tau-u)}v_{q-1}(u)\,\dd u.
\end{equation}
This formula shows explicitly that \(v_q\) can be computed once \(v_{q-1}\)
is known.

We prove by induction that
\begin{equation}
\label{eq:appendix_vq_explicit_solution}
    v_q(\tau)
    =
    \sum_{r=0}^q
    C^{q-r}G_r
    \sum_{i=r}^{q}
    \frac{e^{A_i\tau}}
    {\displaystyle
        \prod_{\substack{j=r\\ j\neq i}}^{q}(A_i-A_j)
    }.
\end{equation}
For \(q=0\), the formula reduces to \(v_0(\tau)=G_0=1\), where the empty
product is interpreted as \(1\). Hence the base case holds.

Assume that the formula holds for \(q-1\), namely
\begin{equation}
\label{eq:appendix_induction_hypothesis}
    v_{q-1}(u)
    =
    \sum_{r=0}^{q-1}
    C^{q-1-r}G_r
    \sum_{i=r}^{q-1}
    \frac{e^{A_i u}}
    {\displaystyle
        \prod_{\substack{j=r\\ j\neq i}}^{q-1}(A_i-A_j)
    }.
\end{equation}
Substituting \eqref{eq:appendix_induction_hypothesis} into
\eqref{eq:appendix_vq_variation_constants} gives
\[
\begin{aligned}
    v_q(\tau)
    &=
    e^{A_q\tau}G_q
    +
    C\int_0^\tau e^{A_q(\tau-u)}
    \sum_{r=0}^{q-1}
    C^{q-1-r}G_r
    \sum_{i=r}^{q-1}
    \frac{e^{A_i u}}
    {\displaystyle
        \prod_{\substack{j=r\\ j\neq i}}^{q-1}(A_i-A_j)
    }
    \,\dd u
    \\
    &=
    e^{A_q\tau}G_q
    +
    \sum_{r=0}^{q-1}
    C^{q-r}G_r
    \sum_{i=r}^{q-1}
    \frac{1}
    {\displaystyle
        \prod_{\substack{j=r\\ j\neq i}}^{q-1}(A_i-A_j)
    }
    \int_0^\tau e^{A_q(\tau-u)}e^{A_i u}\,\dd u.
\end{aligned}
\]
Since \(A_i\neq A_q\) for \(i<q\),
\[
\begin{aligned}
    \int_0^\tau e^{A_q(\tau-u)}e^{A_i u}\,\dd u
    &=
    e^{A_q\tau}
    \int_0^\tau e^{(A_i-A_q)u}\,\dd u
    \\
    &=
    e^{A_q\tau}
    \frac{e^{(A_i-A_q)\tau}-1}{A_i-A_q}
    \\
    &=
    \frac{e^{A_i\tau}-e^{A_q\tau}}{A_i-A_q}.
\end{aligned}
\]
Therefore,
\begin{equation}
\label{eq:appendix_vq_before_rearrangement}
\begin{aligned}
    v_q(\tau)
    &=
    e^{A_q\tau}G_q
    +
    \sum_{r=0}^{q-1}
    C^{q-r}G_r
    \sum_{i=r}^{q-1}
    \frac{
        e^{A_i\tau}-e^{A_q\tau}
    }{
        \displaystyle
        (A_i-A_q)
        \prod_{\substack{j=r\\ j\neq i}}^{q-1}(A_i-A_j)
    }.
\end{aligned}
\end{equation}

It remains to rewrite the inner sum in the desired divided-difference form.
For fixed \(r\leq q-1\), define
\[
    D_i^{(r,q)}
    :=
    \prod_{\substack{j=r\\ j\neq i}}^{q}(A_i-A_j),
    \qquad r\leq i\leq q.
\]
For \(i<q\),
\[
    D_i^{(r,q)}
    =
    (A_i-A_q)
    \prod_{\substack{j=r\\ j\neq i}}^{q-1}(A_i-A_j).
\]
Hence the coefficients of \(e^{A_i\tau}\) for \(i=r,\dots,q-1\) in
\eqref{eq:appendix_vq_before_rearrangement} are already of the form
\(1/D_i^{(r,q)}\). The coefficient of \(e^{A_q\tau}\) is
\[
    -\sum_{i=r}^{q-1}
    \frac{1}
    {
        (A_i-A_q)
        \displaystyle
        \prod_{\substack{j=r\\ j\neq i}}^{q-1}(A_i-A_j)
    }
    =
    -\sum_{i=r}^{q-1}\frac{1}{D_i^{(r,q)}}.
\]
We now use the identity
\begin{equation}
\label{eq:appendix_lagrange_identity}
    \sum_{i=r}^{q}\frac{1}{D_i^{(r,q)}}=0.
\end{equation}
This identity follows, for instance, from the Lagrange interpolation formula
applied to the constant polynomial \(1\). Indeed, the Lagrange basis
polynomials associated with the nodes \(A_r,\dots,A_q\) satisfy
\[
    1
    =
    \sum_{i=r}^q
    \prod_{\substack{j=r\\j\neq i}}^q
    \frac{x-A_j}{A_i-A_j}.
\]
The coefficient of \(x^{q-r}\) on the right-hand side is precisely
\[
    \sum_{i=r}^q
    \frac{1}{\prod_{\substack{j=r\\j\neq i}}^q(A_i-A_j)},
\]
whereas the coefficient of \(x^{q-r}\) on the left-hand side is zero. This
proves \eqref{eq:appendix_lagrange_identity}.

Using \eqref{eq:appendix_lagrange_identity}, the coefficient of \(e^{A_q\tau}\)
is
\[
    -\sum_{i=r}^{q-1}\frac{1}{D_i^{(r,q)}}
    =
    \frac{1}{D_q^{(r,q)}}.
\]
Consequently,
\[
\begin{aligned}
&\sum_{i=r}^{q-1}
    \frac{
        e^{A_i\tau}-e^{A_q\tau}
    }{
        \displaystyle
        (A_i-A_q)
        \prod_{\substack{j=r\\ j\neq i}}^{q-1}(A_i-A_j)
    }
=
\sum_{i=r}^{q}
    \frac{e^{A_i\tau}}
    {\displaystyle
        \prod_{\substack{j=r\\ j\neq i}}^{q}(A_i-A_j)
    }.
\end{aligned}
\]
Substituting this back into \eqref{eq:appendix_vq_before_rearrangement}, we get
\[
    v_q(\tau)
    =
    e^{A_q\tau}G_q
    +
    \sum_{r=0}^{q-1}
    C^{q-r}G_r
    \sum_{i=r}^{q}
    \frac{e^{A_i\tau}}
    {\displaystyle
        \prod_{\substack{j=r\\ j\neq i}}^{q}(A_i-A_j)
    }.
\]
The first term is exactly the \(r=q\) term, since the corresponding product is
empty and therefore equal to \(1\):
\[
    G_qe^{A_q\tau}
    =
    C^0G_q
    \sum_{i=q}^{q}
    \frac{e^{A_i\tau}}
    {\displaystyle
        \prod_{\substack{j=q\\ j\neq i}}^{q}(A_i-A_j)
    }.
\]
Thus \eqref{eq:appendix_vq_explicit_solution} holds for \(q\). By induction, it
holds for all \(q=0,\dots,Q_0\).

Finally, since \(w(t,q)=v_q(T-t)\), we obtain
\begin{equation}
\label{eq:appendix_w_explicit_solution}
    w(t,q)
    =
    \sum_{r=0}^q
    C^{q-r}G_r
    \sum_{i=r}^{q}
    \frac{e^{A_i(T-t)}}
    {\displaystyle
        \prod_{\substack{j=r\\ j\neq i}}^{q}(A_i-A_j)
    },
    \qquad q=0,\dots,Q_0.
\end{equation}
This is the explicit solution used in \eqref{eq:general_explicit_solution}.

\end{document}